\documentclass[11pt]{article}

\usepackage{amsmath,amssymb,amsthm,mathtools}
\usepackage{geometry}
\usepackage{url}
\newcommand{\email}[1]{\href{mailto:#1}{#1}}

\usepackage{float}

\usepackage{adjustbox}
\usepackage{diagbox}
\usepackage{slashbox}
\usepackage{multirow,booktabs}
\usepackage{rotating}

\theoremstyle{definition}
\newtheorem{definition}{Definition}[section]
\newtheorem{remark}[definition]{Remark}
\newtheorem{example}[definition]{Example}
\theoremstyle{plain}
\newtheorem{theorem}[definition]{Theorem}
\newtheorem{lemma}[definition]{Lemma}
\newtheorem{proposition}[definition]{Proposition}
\newtheorem{corollary}[definition]{Corollary}
\usepackage{physics}

\usepackage[most]{tcolorbox}

\tcolorboxenvironment{theorem}{
  enhanced,
  colback=white,
  colframe=black,
  boxrule=0.8pt,
  arc=4pt,
  left=6pt,right=6pt,top=6pt,bottom=6pt,
  before skip=\topsep, after skip=\topsep
}
\definecolor{lightgray}{gray}{0.85}
\tcolorboxenvironment{definition}{
  enhanced,
  borderline west={3pt}{0pt}{lightgray}, 
  frame hidden,                       
  colback=white,                      
  left=6pt, right=6pt, top=6pt, bottom=6pt,
  before skip=\topsep, after skip=\topsep,
      breakable   
}

\usepackage{fontawesome5}

\tcolorboxenvironment{proof}{
  enhanced,
  borderline west={3pt}{0pt}{gray},   
  frame hidden,                             
  colback=white,
  left=6pt, right=6pt, top=6pt, bottom=6pt,
  before skip=\topsep, after skip=\topsep,
  breakable                                 
}

\newcommand{\C}{\mathbb{C}}
\newcommand{\PP}{\mathbb{P}}
\newcommand{\GL}{\mathrm{GL}}
\newcommand{\PGL}{\mathrm{PGL}}
\newcommand{\End}{\mathrm{End}}
\newcommand{\id}{\mathrm{id}}
\newcommand{\rk}{\mathrm{rk}}

\DeclareMathOperator{\Vol}{Vol}

\DeclareMathOperator{\Hom}{Hom}
\DeclareMathOperator{\Sing}{Sing}
\DeclareMathOperator{\Gr}{Gr}

\usepackage{amsmath,amssymb}

\providecommand{\cO}{\mathcal{O}}
\providecommand{\Gm}{\mathbb{G}_{\mathrm m}}
\providecommand{\PP}{\mathbb{P}}
\providecommand{\Spec}{\operatorname{Spec}}
\providecommand{\Proj}{\operatorname{Proj}}
\providecommand{\PP}{\mathbb{P}}

\providecommand{\GL}{\operatorname{GL}}
\providecommand{\PGL}{\operatorname{PGL}}
\providecommand{\SL}{\operatorname{SL}}

\providecommand{\per}{\operatorname{per}}

\providecommand{\End}{\operatorname{End}}
\providecommand{\Mat}{\operatorname{Mat}}

\providecommand{\codim}{\operatorname{codim}}
\providecommand{\Hilb}{\operatorname{Hilb}}

\providecommand{\Schur}{\mathbb{S}}
\providecommand{\len}{\ell}

\newcommand{\Z}{\mathbb{Z}}

\newcommand{\im}{\mathrm{im}}

\newcommand{\sslash}{\mathbin{/\mkern-6mu/}}

\definecolor{QGLBlue}{RGB}{65,105,225}   
   
\usepackage[unicode]{hyperref} 
\hypersetup{
  colorlinks=true,
  linkcolor=QGLBlue,
  citecolor=QGLBlue,
  urlcolor=QGLBlue,
  breaklinks=true
}

\title{Introduction to Quantum Entanglement Geometry\\
{\small Entanglement Filtration in Azumaya Algebras and Geometric Entanglement in Quantum Systems}}
\author{Kazuki Ikeda \thanks{Email: \email{kazuki.ikeda@umb.edu} University of Massachusetts Boston}}

\begin{document}
\maketitle

\begin{abstract}
This article is an expository account aimed at viewing entanglement in finite-dimensional quantum many-body systems as a phenomenon of \emph{global geometry}. While the mathematics of general quantum states has been studied extensively, this article focuses specifically on their entanglement. When a quantum system varies over a classical parameter space $X$, each fiber may look like the same Hilbert space, yet there may be no global identification because of twisting in the gluing data. Describing this situation by an Azumaya algebra, one always obtains the family of pure-state spaces as a Severi--Brauer scheme $SB(A)\to X$.

The main focus is to characterize the condition under which the subsystem decomposition required to define entanglement exists globally and compatibly, by a reduction to the stabilizer subgroup $G_d$ of the Segre variety, and to explain that the obstruction appears in the Brauer class. As a consequence, quantum states yield a natural filtration dictated by entanglement on $SB(A)$.

Using a spin system on a torus as an example, we show concretely that the holonomy of the gluing can produce an entangling quantum gate, and can appear as an obstruction class distinct from the usual Berry numbers or Chern numbers. For instance, even for quantum systems that have traditionally been regarded as having no topological band structure, the entanglement of their eigenstates can be related to global geometric universal quantities, reflecting the background geometry. 
\end{abstract}

\setcounter{tocdepth}{2} 
\tableofcontents\

\section{What Is Quantum Entanglement Geometry?}

\subsection{Parameterized quantum systems}
In quantum information textbooks (for example \cite{nielsen2000quantum}), one typically begins by fixing a single Hilbert space
\[
H \simeq \C^{n}
\]
and then discusses states, measurements, and operations on that fixed stage. This is ideal when one wants to treat the same system repeatedly on the same stage.

In actual physics, however, quantum systems often appear as depending on \emph{classical} parameters. Typical situations include
\begin{itemize}
  \item control parameters such as coupling constants, external fields, temperature, and boundary conditions;
  \item spacetime parameters such as time $t$, position $x$, and momentum $k$. 
\end{itemize}
In such cases, for each point $x\in X$ of a parameter space $X$, one has a quantum system at that point. This naturally leads to an assignment
\[
x\in X \quad\longmapsto\quad H_x
\]
where $H_x$ is the Hilbert space at parameter $x$.

The simplest situation is when the $H_x$ arise as fibers of a single global vector bundle $E\to X$, so that $H_x\simeq E_x$. Then one can treat bases and operators coherently over $X$. In practice, however, one easily encounters situations where
\begin{center}
\emph{locally $H_x\simeq \C^n$, yet the gluing of bases on overlaps fails to be globally consistent.}
\end{center}
Moreover, the Hilbert spaces $H_x$ and $H_y$ for different $x,y\in X$ need not even have the same dimension.

This note aims to construct a geometric description of quantum states, with a particular focus on quantum entanglement.

A key point is that what matters in quantum information is not a vector itself, but the state obtained after identifying vectors up to an overall scalar. A pure state is represented as
\[
\text{nonzero vectors } H_x\setminus\{0\}\ \text{modulo scalars}
\quad\longleftrightarrow\quad \PP(H_x).
\]
This leads to an important observation:
\begin{center}
\emph{even when a global vector bundle $E$ does not exist, a family of projective spaces often \\
does exist globally.}
\end{center}
However, not every quantum state is of interest. In quantum theory, product states are valid quantum states, but they can be regarded as classical in the sense that they exhibit no quantum correlations. Among pure states, only entangled states possess genuinely nonclassical (quantum) features. Our first aim is to formulate entanglement as a geometric object.

In what follows, keeping in mind this fact that the pure-state space exists globally, we regard pure states not as a mere collection of points, but as a (twisted) projective bundle over $X$. 

\vskip0.3cm
The following three objects play central roles in this note.
\begin{enumerate}
  \item \textbf{A family of pure states:} fiberwise one sees $\PP(H_x)\simeq \PP^{n-1}$, but globally it may be twisted.
  \item \textbf{The locus of product states:} once a subsystem decomposition is given, product states appear as a Segre variety.
  \item \textbf{A filtration of quantum entanglement:} in the bipartite case, Schmidt rank forms a filtration given by determinantal loci (minors).
\end{enumerate}
All of this is familiar over a fixed Hilbert space $H$, but for parameterized systems a new issue arises: (2) and (3) do not automatically globalize. The goal of this work is to determine when they \emph{can} (or \emph{cannot}) be globalized. This idea was initiated and developed in \cite{Ikeda:2026ojm}, and the remainder of this article is devoted to explaining the underlying concepts and the main results.

\subsection{What is entanglement?}
In quantum information, the entanglement of a pure state is defined relative to a fixed subsystem decomposition
\[
H \simeq H_{1}\otimes\cdots\otimes H_{r}.
\]
A state $|\psi\rangle$ is entangled (with respect to this decomposition) if it cannot be written in the product form
\[
|\psi\rangle = |\psi_{1}\rangle\otimes\cdots\otimes|\psi_{r}\rangle.
\]

Geometrically, the set of all product states is the image of the Segre embedding:
\[
\Sigma_{d}\;=\;\PP^{d_1-1}\times\cdots\times\PP^{d_r-1}\ \subset\ \PP(H)
\qquad(n=\prod_i d_i).
\]
Thus the product locus appears as a distinguished subvariety $\Sigma_d\subset \PP(H)$ (rank-$1$ tensors), and the entangled states are those outside it. Entanglement is an intrinsically quantum correlation present in a state, and a central theme in quantum information and quantum physics is to quantify how much (how strongly) entanglement is present.

For parameterized systems, however, before one even asks whether a state is entangled, a more fundamental question appears:
\begin{quote}
\begin{enumerate}
    \item \emph{How should one choose a decomposition into subsystems?}
    \item \emph{Can one choose the tensor decomposition itself coherently over the entire parameter space?}
\end{enumerate}
\end{quote}

Even if locally one can identify
\[
H_x \simeq \C^{d_1}\otimes\cdots\otimes \C^{d_r},
\]
if the gluing on overlaps becomes a \emph{general} change of basis, then the Segre variety (the locus of product states) visible locally will not be preserved. As a result, over the whole space $X$ one may be forced into the situation that
\begin{center}
\emph{product/entangled cannot be defined coherently as global geometric objects.}
\end{center}
This is the main theme of the present paper. The quantum entanglement geometry discussed here is an
approach that treats, in an algebro--geometric way, both
\begin{center}
\emph{state entanglement} \quad and \quad \emph{the obstruction to the existence of a subsystem decomposition}
\end{center}
simultaneously \cite{Ikeda:2026ojm}.

\subsection{\label{sec:elementary_example}An elementary example of quantum information in algebraic geometry}
\noindent\textbf{Concurrence.}
Our aim is to present a geometric framework for quantum information in a concrete and elementary way. To explain the underlying philosophy, let us consider the most elementary example. Any two-qubit pure state can be written as
\[
|\psi\rangle = a|00\rangle+b|01\rangle+c|10\rangle+d|11\rangle,\qquad
\Psi=\begin{pmatrix}a&b\\ c&d\end{pmatrix}.
\]
Since a global scalar is physically meaningless, it is natural to treat the state as a point in the projective space
\[
[\,\psi\,]=[a:b:c:d]\in \mathbb P(\mathbb C^2\otimes \mathbb C^2)\cong \mathbb P^3.
\]

There are several ways to decide whether this state is entangled. Here we consider the concurrence $C$ \cite{PhysRevLett.78.5022}. For a \emph{pure state} of two qubits, under the normalization $\langle\psi|\psi\rangle=1$, one has
\[
C(|\psi\rangle)=2\,|ad-bc| \;=\; 2\,|\det(\Psi)|.
\]
It is known that $C$ takes values in the range  $0\le C\le 1$. A state is entangled if and only if $C\neq0$. For instance, the Bell state $(|00\rangle+|11\rangle)/\sqrt2$  has $\det(\Psi)=1/2~(C=1)$, detecting maximal entanglement.

\medskip
\noindent\textbf{The matrix criterion for separability and the Segre variety.}
A state is a product if there exist $|\phi\rangle=\alpha|0\rangle+\beta|1\rangle$ and $|\chi\rangle=\gamma|0\rangle+\delta|1\rangle$ such that
\[
|\psi\rangle = |\phi\rangle\otimes|\chi\rangle.
\]
In terms of coefficients this means
\[
(a,b,c,d)=(\alpha\gamma,\alpha\delta,\beta\gamma,\beta\delta),
\]
and it immediately follows that $ad-bc=0$. Conversely, if $ad-bc=0$ then $\Psi$ has rank $\le1$, so it admits a factorization $\Psi=(\alpha,\beta)^{\mathsf T}(\gamma,\delta)$ and hence corresponds to a product state. Thus
\[
\text{product} \;\Longleftrightarrow\; \mathrm{rank}(\Psi)=1
\;\Longleftrightarrow\; ad-bc=0.
\]
In projective-geometric terms, the set of product states is the image of the Segre embedding
\begin{equation}\label{eq:segre-2qubit}
\mathrm{Seg}:\mathbb P^1\times \mathbb P^1 \hookrightarrow \mathbb P^3,    
\end{equation}
and this image is a smooth quadric surface defined by the single quadratic equation
\[
ad-bc=0.
\]
Therefore, whether a state is a product state is described purely algebro--geometrically as
\emph{whether the point lies on a certain algebraic variety} (see also Remark \ref{rem:entangled_is_OR}). The concurrence is precisely an elementary entanglement certificate given by  the defining equation of this Segre variety.

\medskip
\noindent\textbf{Local operations and symmetry.}
The action of a local unitary transformation $U_A\otimes U_B$ corresponds, in the matrix representation, to
\[
\Psi\longmapsto (U_A)\,\Psi\,(U_B)^{\mathsf T}.
\]
Hence
\[
\det(\Psi)\longmapsto \det(U_A)\det(U_B)\det(\Psi),
\]
and in particular, if $U_A,U_B\in SU(2)$ then $\det(\Psi)$ (and therefore $C$) is invariant. In this sense, $\det(\Psi)$ is an algebraic quantity that behaves well under local operations, and the vanishing condition $\det(\Psi)=0$ is a geometric condition invariant under local operations (this is the first manifestation of the fact that, although entanglement depends on the choice of subsystems, it is not destroyed by local operations).

\medskip
\noindent\textbf{Outlook to higher dimensions.}
For two qubits, separability was expressed by a single determinant. Similarly, for a general bipartite system $\mathbb C^{d_A}\otimes\mathbb C^{d_B}$, if one arranges the coefficients into a $d_A\times d_B$ matrix, then product states are exactly those of rank $1$, characterized by the vanishing of all $2\times2$ minors (the Segre variety). More generally, Schmidt rank $\le k$ is given by the vanishing of all $(k+1)\times(k+1)$ minors. Thus, the entanglement filtration of pure states can be described as
\[
\text{minors vanish}\quad\Longleftrightarrow\quad
\text{the state lies on a determinantal variety},
\]
namely in the language of polynomial equations / algebraic varieties.

The philosophy extracted from this elementary example is:
\begin{quote}
The basic structures of quantum information (separability and polynomial invariants) naturally appear as the Segre variety in projective space and its defining equations.
\end{quote}
In this article we extend this viewpoint to families (where quantum systems vary along classical parameters), and further to twisted backgrounds (Severi--Brauer fibrations), and we formulate algebro--geometrically the question: \emph{when (and with what obstructions) can the Segre variety be defined globally?} And when possible, we construct a global filtration.

\vskip0.3cm
\subsection{Organization of this note.}
In the rest of the paper, in order to make the discussion above precise, we gradually introduce tools for describing the twisting of projective bundles, gluing that preserves the Segre embedding (reduction of structure group), and the relative formulation of determinantal loci in the bipartite case.

Section~\ref{sec:2} reviews the projective-geometric description of pure states and the Segre embedding (\emph{standard background}). Section~\ref{sec:integer} introduces a partition-lattice viewpoint on separability patterns together with a brief
number-theoretic analogy (\emph{original to this note}). Section~\ref{sec:secant-tensor-rank} recalls tensor rank, secant varieties, and border rank as an additive refinement of the entanglement filtration (\emph{standard background}; see ref. \cite{landsberg2011tensors} for example). Section~\ref{sec:5} specializes to the bipartite case, relating product states to matrix rank and determinantal loci (\emph{standard background}).

Sections \ref{sec:6}--\ref{sec:spin-chain-toy} are based on ref. \cite{Ikeda:2026ojm} and relevant established facts. Sections~6--8 turn to families of quantum systems over a parameter space $X$ and explain how local subsystem decompositions correspond to reductions of structure group to the stabilizer of the Segre variety, and how the obstruction is encoded by the Brauer class. Sections~9--14 develop the geometry of entangled states on the Severi--Brauer scheme $\mathrm{SB}(A)\to X$. Section~15 presents a concrete spin-system example on a torus where the gluing holonomy yields an entangling gate.

Sections~\ref{sec:16}--\ref{sec:general_types} propose a spectral/Hecke-theoretic criterion of entanglement (\emph{original to this note}). We define $d$-product Satake parameters as the tensor-product image in the dual group and explain how this image condition is detected by finitely many invariant polynomials, with explicit computations in the cases $d=(2,2)$ and $d=(2,2,2)$. Section~\ref{sec:summary} summarizes the geometric and spectral viewpoints, and the appendices collect auxiliary material on determinantal resolutions (\ref{sec:supplement-3points}) and on the gluing computations in the torus example (\ref{sec:torus}).

\section{Geometry of Pure States}\label{sec:2}

\subsection{Pure states are projective space}
A state vector $\psi\in H$ representing a pure state is physically unchanged if we multiply it by an overall nonzero scalar $\lambda$. Thus we regard a pure state as a point in projective space.

\begin{definition}[Projective space]
For a complex vector space $H\simeq \C^n$, define the projective space $\PP(H)$ by
\[
\PP(H)\;:=\;(H\setminus\{0\})/\sim,
\qquad
\psi\sim \lambda\psi\ \ (\lambda\in\C^\times).
\]
It means that a point of $\PP(H)$ is a $1$-dimensional linear subspace (a line) in $H$.
\end{definition}

\begin{proposition}[Two equivalent descriptions of pure states]
Let $H$ be a finite-dimensional complex inner-product space. The following sets are naturally in bijection:
\begin{enumerate}
  \item points of the projective space $\PP(H)$,
  \item rank-one orthogonal projections (density matrices)
  \[
  \Bigl\{\,P\in\End(H)\ \bigm|\ P^2=P,\ P^\dagger=P,\ \rk(P)=1\,\Bigr\}.
  \]
\end{enumerate}
\end{proposition}

\begin{proof}
\textbf{(Step 1) Define a map $\PP(H)\to$ rank-one projectors.}
Take $[\psi]\in\PP(H)$ and choose a representative $\psi\in H\setminus\{0\}$. Define
\[
P_\psi \;:=\;\frac{\ket{\psi}\bra{\psi}}{\braket{\psi}{\psi}}.
\]

\smallskip
\textbf{(Step 2) Show it is independent of the representative.}
If we choose another representative $\psi'=\lambda\psi$ with $\lambda\in\C^\times$, then
\[
\ket{\psi'}\bra{\psi'}
=
\ket{\lambda\psi}\bra{\lambda\psi}
=
(\lambda\ket{\psi})(\overline{\lambda}\bra{\psi})
=
|\lambda|^2\ket{\psi}\bra{\psi},
\]
and also
\[
\braket{\psi'}{\psi'}
=
\braket{\lambda\psi}{\lambda\psi}
=
\overline{\lambda}\lambda\braket{\psi}{\psi}
=
|\lambda|^2\braket{\psi}{\psi}.
\]
Hence
\[
P_{\psi'}
=
\frac{|\lambda|^2\ket{\psi}\bra{\psi}}{|\lambda|^2\braket{\psi}{\psi}}
=
\frac{\ket{\psi}\bra{\psi}}{\braket{\psi}{\psi}}
=
P_\psi.
\]
Therefore $P_\psi$ depends only on $[\psi]$.

\smallskip
\textbf{(Step 3) Show $P_\psi$ is a rank-one orthogonal projection.}
We compute
\[
P_\psi^2
=
\frac{\ket{\psi}\bra{\psi}}{\braket{\psi}{\psi}}
\frac{\ket{\psi}\bra{\psi}}{\braket{\psi}{\psi}}
=
\frac{\ket{\psi}(\bra{\psi}\ket{\psi})\bra{\psi}}{\braket{\psi}{\psi}^2}
=
\frac{\ket{\psi}\braket{\psi}{\psi}\bra{\psi}}{\braket{\psi}{\psi}^2}
=
P_\psi.
\]
Taking the adjoint gives
\[
P_\psi^\dagger
=
\left(\frac{\ket{\psi}\bra{\psi}}{\braket{\psi}{\psi}}\right)^\dagger
=
\frac{(\ket{\psi}\bra{\psi})^\dagger}{\overline{\braket{\psi}{\psi}}}
=
\frac{\ket{\psi}\bra{\psi}}{\braket{\psi}{\psi}}
=
P_\psi
\]
(since $\braket{\psi}{\psi}$ is real), so $P_\psi$ is self-adjoint. Moreover, $\mathrm{im}(P_\psi)=\C\psi$ is $1$-dimensional, hence $\rk(P_\psi)=1$.

\smallskip
\textbf{(Step 4) Construct the inverse map.}
Let $P$ be a rank-one orthogonal projection. Then $\mathrm{im}(P)\subset H$ is a $1$-dimensional subspace. We associate to $P$ the point $[\psi]:=\mathrm{im}(P)\in\PP(H)$.

\smallskip
\textbf{(Step 5) Show the two maps are inverse to each other.}
Starting from $[\psi]$, the image of the projector $P_\psi$ is $\C\psi$, so applying the inverse map returns $[\psi]$. Conversely, starting from $P$ and letting $[\psi]=\mathrm{im}(P)$, choose a representative $\psi$ of this line. Since $P$ is the orthogonal projection onto $\C\psi$, we have $P=P_\psi$. Thus the correspondence is a bijection.
\end{proof}

The identification above corresponds to a Bloch sphere representation for the general case.

\subsection{Product states and the Segre map}
For the moment we focus on the bipartite case. The multipartite generalization appears in \S\ref{sec:general_types}. 

\begin{definition}[Product states in the bipartite case]
Let $H_A\simeq\C^{d_A}$ and $H_B\simeq\C^{d_B}$, and set $H:=H_A\otimes H_B$.
A pure state $[\psi]\in\PP(H)$ is called \emph{product} if there exist
$a\in H_A\setminus\{0\}$ and $b\in H_B\setminus\{0\}$ such that
\[
[\psi]=[a\otimes b].
\]
\end{definition}

Choose bases $\{e_i\}_{i=1}^{d_A}$ of $H_A$ and $\{f_j\}_{j=1}^{d_B}$ of $H_B$, and write
\[
a=\sum_{i=1}^{d_A} a_i e_i,\qquad
b=\sum_{j=1}^{d_B} b_j f_j.
\]
Then
\[
a\otimes b
=
\left(\sum_i a_i e_i\right)\otimes\left(\sum_j b_j f_j\right)
=
\sum_{i,j} a_i b_j (e_i\otimes f_j).
\]
That is, with respect to the basis $\{e_i\otimes f_j\}$ of $H$, the coefficients have the ``product form''
\[
\psi_{ij}=a_i b_j.
\]

Projectivizing this construction gives the Segre map:
\[
\sigma:\PP(H_A)\times\PP(H_B)\longrightarrow \PP(H_A\otimes H_B),
\qquad
([a],[b])\longmapsto [a\otimes b].
\]
Its image $\Sigma_{A,B}:=\mathrm{im}(\sigma)$ is the Segre variety (the locus of product states).

\subsection{Product states and matrix rank}
\begin{definition}[Flattening of a tensor]
With respect to the chosen bases, write any $\psi\in H_A\otimes H_B$ as
\[
\psi=\sum_{i=1}^{d_A}\sum_{j=1}^{d_B}\psi_{ij}\,(e_i\otimes f_j).
\]
The $d_A\times d_B$ matrix obtained by arranging the coefficients,
\[
M(\psi):=(\psi_{ij})\in \mathrm{Mat}_{d_A\times d_B}(\C),
\]
is called the flattening matrix of $\psi$.
\end{definition}

\begin{proposition}[Product state $\Leftrightarrow$ matrix rank $1$]\label{prop:matrixrank1}
Assume $\psi\neq 0$. The following are equivalent:
\begin{enumerate}
  \item $[\psi]$ is product ($[\psi]=[a\otimes b]$),
  \item $\rk(M(\psi))=1$.
\end{enumerate}
\end{proposition}

\begin{proof}
\textbf{(1)$\Rightarrow$(2).}
If $[\psi]=[a\otimes b]$, then by the computation above we have $\psi_{ij}=a_i b_j$, and hence
\[
M(\psi)=
\begin{pmatrix}
a_1\\ \vdots\\ a_{d_A}
\end{pmatrix}
\begin{pmatrix}
b_1 & \cdots & b_{d_B}
\end{pmatrix}
=
a\,b^{\mathsf{T}}
\]
Every column is a scalar multiple of $a$, so the column space is $1$-dimensional. Thus $\rk(M(\psi))=1$.

\smallskip
\textbf{(2)$\Rightarrow$(1).}
Assume $\rk(M(\psi))=1$. Then all columns lie in a single $1$-dimensional subspace of $\C^{d_A}$, so there exists $a\neq 0$ such that each column is a scalar multiple of $a$. Writing the $j$-th column as $c^{(j)}$, there exists $b_j\in\C$ such that
\[
c^{(j)} = b_j a.
\]
In components this means
\[
\psi_{ij} = a_i b_j.
\]
Let $b:=(b_1,\dots,b_{d_B})$. Then
\[
\psi=\sum_{i,j}\psi_{ij}(e_i\otimes f_j)
=\sum_{i,j}a_i b_j(e_i\otimes f_j)
=\left(\sum_i a_i e_i\right)\otimes\left(\sum_j b_j f_j\right)
=a\otimes b,
\]
so $[\psi]=[a\otimes b]$ is product.
\end{proof}

The rank-$1$ condition is equivalent to the vanishing of all $2\times 2$ minors. This can be shown elementarily as follows.

\begin{theorem}
For a $d_A\times d_B$ matrix $M$, the following are equivalent:
\begin{enumerate}
  \item $\rk(M)\le 1$,
  \item for any $1\le i_1<i_2\le d_A$ and $1\le j_1<j_2\le d_B$,
  \[
  \det
  \begin{pmatrix}
  M_{i_1 j_1} & M_{i_1 j_2}\\
  M_{i_2 j_1} & M_{i_2 j_2}
  \end{pmatrix}
  =0.
  \]
\end{enumerate}
\end{theorem}

\begin{proof}
\textbf{(1)$\Rightarrow$(2).}
If $\rk(M)\le 1$, then any two columns are linearly dependent, so any $2\times 2$ submatrix has rank $\le 1$ and hence determinant $0$. More directly, if $\rk(M)\le 1$ then $M=a b^{\mathsf T}$ for some vectors $a,b$, so
\[
\det
\begin{pmatrix}
a_{i_1}b_{j_1} & a_{i_1}b_{j_2}\\
a_{i_2}b_{j_1} & a_{i_2}b_{j_2}
\end{pmatrix}
=
a_{i_1}a_{i_2}
\det
\begin{pmatrix}
b_{j_1} & b_{j_2}\\
b_{j_1} & b_{j_2}
\end{pmatrix}
=0.
\]

\smallskip
\textbf{(2)$\Rightarrow$(1).}
We argue by contradiction. If $\rk(M)\ge 2$, then at least two columns are linearly independent. Let these columns be $u,v\in\C^{d_A}$. Since $u$ and $v$ are linearly independent, there exist indices $i_1,i_2$ such that
\[
\det
\begin{pmatrix}
u_{i_1} & v_{i_1}\\
u_{i_2} & v_{i_2}
\end{pmatrix}\neq 0
\]
(otherwise $u$ and $v$ would be linearly dependent). But $u$ and $v$ are the columns corresponding to some $j_1,j_2$, so this determinant is exactly a nonzero $2\times 2$ minor, contradicting (2). Hence $\rk(M)\le 1$.
\end{proof}

\begin{remark}
The locus defined by minors is, geometrically, the image of the Segre embedding of $\PP^{d_A-1}\times\PP^{d_B-1}$ (the Segre variety). This generalizes the elementary example (eq. \eqref{eq:segre-2qubit}) to the generic case.
\end{remark}

\subsection{Ideals for multipartite quantum systems}
\label{subsec:partition-ideal}

In this section we extend, from the bipartite case, the correspondence
\[
\text{product states}
\;\Longleftrightarrow\;
\text{lying on a Segre variety}
\;\Longleftrightarrow\;
\text{all $2\times2$ minors vanish}
\]
to the multipartite setting. Moreover, we prepare a dictionary that organizes the logic
\begin{quote}
\emph{for which partitions does the state factorize as a product?}
\end{quote}
in terms of ideals.

\medskip
\noindent
The basic philosophy is the same for a general subsystem decomposition $H\simeq H_1\otimes\cdots\otimes H_r$: product states form a Segre variety, and entangled states lie outside it. In the multipartite case, there is much more freedom in how one chooses subsystems, and entanglement with respect to subsystems is no longer captured only by conditions of the form minors vanish. The aim of this section is therefore to set up as much general theory as possible that remains valid without relying heavily on minors.

\subsubsection{Partitions and $\pi$-product states}
Fix the number of subsystems $N\ge 2$ and a finite-dimensional complex vector space
\[
H \;=\; H_1\otimes H_2\otimes\cdots\otimes H_N
\qquad(\dim H_i=d_i).
\]
Pure states are points of the projective space $P(H)$.

\begin{definition}[Partition and refinement]
A \emph{partition} $\pi$ of the set $[N]:=\{1,2,\dots,N\}$ is a collection of pairwise disjoint nonempty subsets $B_1,\dots,B_r\subset [N]$ such that
\[
[N]=B_1\sqcup B_2\sqcup\cdots\sqcup B_r.
\]
We call each $B_j$ a \emph{block}.

For two partitions $\pi,\rho$, we define
\[
\pi\preceq \rho
\quad\xLeftrightarrow[]{\text{def}}\quad
\pi \text{ is a refinement of }\rho,
\]
meaning that each block of $\pi$ is contained in some block of $\rho$.
\end{definition}

\begin{definition}[$\pi$-product states and the Segre variety $\Sigma_\pi$]
For a partition $\pi=\{B_1,\dots,B_r\}$, set
\[
H_{B_j}:=\bigotimes_{i\in B_j} H_i.
\]
A pure state $[\psi]\in P(H)$ is said to be \emph{$\pi$-product} if there exist nonzero vectors $0\neq \psi_{B_j}\in H_{B_j}$ such that
\[
[\psi]=[\psi_{B_1}\otimes\cdots\otimes \psi_{B_r}].
\]
We denote the set of all such states by
\[
\Sigma_\pi \subset \PP(H)
\]
and call it the \emph{$\pi$-product locus}.
\end{definition}

Since $\Sigma_\pi$ is the image of the Segre embedding
\[
\PP(H_{B_1})\times\cdots\times\PP(H_{B_\ell})
\hookrightarrow \PP(H),
\]
it is a closed subset and moreover irreducible. Hence (in the sense of the homogeneous coordinate ring) its vanishing ideal
\[
I(\Sigma_\pi)\subset k[\PP(H)]
\]
is a \emph{prime ideal}.

If a partition $\pi$ refines $\pi'$ ( $\pi\preceq \pi'$, meaning each block of $\pi'$ is a union of blocks of $\pi$), then 
\[
\Sigma_\pi \subset \Sigma_{\pi'}.
\]
By contravariance this yields
\[
I(\Sigma_{\pi'})\subset I(\Sigma_\pi).
\]

\begin{remark}[Relation to the type $d$ (subsystem type)]
For each block $B$, set $D_B:=\dim H_B=\prod_{i\in B} d_i$. Then a partition $\pi=\{B_1,\dots,B_r\}$ induces a tuple of dimensions
\[
d(\pi):=(D_{B_1},\dots,D_{B_r})
\qquad\Bigl(\ \prod_{j=1}^r D_{B_j}=\dim H\ \Bigr).
\]
In this sense, $\Sigma_\pi$ is the same as the Segre variety of type $d(\pi)$ (the Segre image of $P^{D_{B_1}-1}\times\cdots\times P^{D_{B_r}-1}$. The order of the blocks is not essential).
\end{remark}

\begin{proposition}[Coarser partitions give larger $\Sigma_\pi$ / smaller ideals]
\label{prop:refinement-inclusion}
If partitions $\pi,\rho$ satisfy $\pi\preceq \rho$ (thus $\pi$ is finer), then
\[
\Sigma_\pi\subseteq \Sigma_\rho
\qquad\text{and}\qquad
I(\Sigma_\pi)\supseteq I(\Sigma_\rho).
\]
\end{proposition}

\begin{proof}
Saying that $\pi$ is finer means that the condition can be written as a product is stronger because it requires a decomposition into more tensor factors. Suppose $\pi=\{B_1,\dots,B_r\}$ and $\psi=\psi_{B_1}\otimes\cdots\otimes\psi_{B_r}$. Since $\rho$ is obtained by grouping some blocks of $\pi$, each block $C$ of $\rho$ can be written as $C=B_{i_1}\sqcup\cdots\sqcup B_{i_m}$. Setting
\[
\psi_C:=\psi_{B_{i_1}}\otimes\cdots\otimes\psi_{B_{i_m}}\in H_C
\]
gives $\psi=\bigotimes_{C\in \rho}\psi_C$, so $\psi$ is $\rho$-product. Hence $\Sigma_\pi\subseteq\Sigma_\rho$. In general, if $X\subseteq Y$ then $I(X)\supseteq I(Y)$, so the ideal inclusion follows in the opposite direction.
\end{proof}

\subsubsection{Vanishing ideals and $2\times2$ minors (the bipartition case)}

In projective geometry one should think of a point $[v]\in \mathbf{P}(H)$ as a line $\mathbf{C}v\subset H$ rather than a specific vector $v\in H\setminus\{0\}$. On the other hand, a homogeneous polynomial $f\in \mathrm{Sym}^t(H^\vee)$ is naturally a polynomial function on the vector space $H$. The role of the affine cone below is to bridge these two viewpoints.

\begin{definition}[Zero locus and vanishing ideal]
Let $P(H)\cong P^{D-1}$ ($D=\prod_i d_i$), and let its homogeneous coordinate ring be
\[
S:=\mathrm{Sym}(H^\vee)=\bigoplus_{t\ge 0}\mathrm{Sym}^t(H^\vee)
\]
graded by $\mathbb{Z}_{\ge0}$.

For a subset $X\subset P(H)$, define its (homogeneous) affine cone by
\[
\widehat{X}:=\{\,v\in H \mid v=0\ \text{or}\ [v]\in X\,\}\subset H.
\]
Then the \emph{vanishing ideal} of $X$ is defined by
\[
I(X):=\{\,f\in S\ \text{homogeneous}\mid f(v)=0\ \text{for all }v\in\widehat{X}\,\}.
\]
(For homogeneous $f$, the condition $f([v])=0$ is well-defined, independent of the choice of representative $v$.)

Conversely, for a homogeneous ideal $I\subset S$, define the zero locus
\[
V(I):=\{\, [v]\in P(H)\mid f(v)=0\ \text{for all homogeneous }f\in I\,\}.
\]
\end{definition}

If $f\in \mathrm{Sym}^t(H^\vee)$ is homogeneous of degree $t$, then for any $\lambda\in \mathbf{C}$ and
$v\in H$ one has
\[
f(\lambda v)=\lambda^t f(v).
\]
In particular, for $v\neq 0$ the condition $f(v)=0$ depends only on the projective class $[v]$, because $f(\lambda v)=0$ if and only if $f(v)=0$. This is why we restrict to homogeneous polynomials when defining $I(X)$ and $V(I)$.

Also, note that including $v=0$ in $\widehat{X}$ is harmless: every homogeneous polynomial of positive degree vanishes at $0$, and the only degree-$0$ homogeneous polynomial that vanishes at $0$ is the zero polynomial.  So $I(X)$ contains no nonzero constants and is automatically a homogeneous ideal.

\begin{example}[A point in $\mathbf{P}^1$]
Let $H=\mathbf{C}^2$ and choose dual coordinates $(x_0,x_1)$ on $H$, so $S=\mathbf{C}[x_0,x_1]$.
Let
\[
X=\{[1:0]\}\subset \mathbf{P}(H)=\mathbf{P}^1.
\]
Then the affine cone is the line through $(1,0)$:
\[
\widehat{X}=\{(t,0)\mid t\in \mathbf{C}\}\subset \mathbf{C}^2.
\]
A homogeneous polynomial $f(x_0,x_1)$ vanishes on $\widehat{X}$ if and only if $f(t,0)=0$ for all $t$. Hence
\[
I(X)=(x_1)\subset \mathbf{C}[x_0,x_1],
\qquad
V\bigl((x_1)\bigr)=\{[1:0]\}=X.
\]
So, in this simplest case, the projective subset $X$ is recovered as $V(I(X))$.
\end{example}

\medskip
In the bipartite case, $\Sigma_{A|B}$ is characterized by matrix rank $1$, and this condition is the vanishing of all $2\times2$ minors. The same holds for a bipartition $A|B$ in the multipartite setting.

\begin{definition}[Flattening]\label{def:Flattening}
Fix a bipartition $A|B$ (so $A\sqcup B=[N]$), and set
\[
H_A:=\bigotimes_{i\in A}H_i,\qquad H_B:=\bigotimes_{j\in B}H_j.
\]
After choosing bases, the coefficient tensor can be written with multi-indices $x_{a_1\cdots a_N}$. By grouping indices into the $A$-side and the $B$-side, one arranges the coefficients into a $D_A\times D_B$ matrix ($D_A=\dim H_A$, $D_B=\dim H_B$)
\[
M_{A|B}=\bigl(x_{\alpha,\beta}\bigr),
\]
which is called the \emph{flattening}.
\end{definition}

The vanishing ideal of the product locus for the bipartition $A|B$,
\[
\Sigma_{A|B}:=\mathrm{Seg}\bigl(P(H_A)\times P(H_B)\bigr)\subset P(H),
\]
is given by
\[
I_{A|B}:=I(\Sigma_{A|B})
=\big\langle \text{all $2\times2$ minors of }M_{A|B}\big\rangle.
\]

\begin{remark}
$V(I_{A|B})=\Sigma_{A|B}$ is the Segre variety of rank-$1$ matrices, and $I_{A|B}$ is a typical \emph{determinantal ideal}. Since $\Sigma_{A|B}$ is irreducible, $I_{A|B}$ is a \emph{prime ideal}.
\end{remark}

\begin{example}[Two qubits]
Let $H=\mathbf{C}^2\otimes \mathbf{C}^2$ and choose bases $\{e_0,e_1\}$ and $\{f_0,f_1\}$.
Write
\[
\psi=\sum_{i,j\in\{0,1\}} \psi_{ij}\, e_i\otimes f_j,
\qquad
[a:b:c:d]:=[\psi_{00}:\psi_{01}:\psi_{10}:\psi_{11}]\in \mathbf{P}(H)\cong \mathbf{P}^3,
\]
so $S=\mathbf{C}[a,b,c,d]$.

The \emph{flattening matrix} is
\[
M(\psi)=\begin{pmatrix}
a & b\\
c & d
\end{pmatrix}.
\]
The product locus (Segre variety) is
\[
\Sigma_{2,2}=\mathrm{Seg}\bigl(\mathbf{P}^1\times \mathbf{P}^1\bigr)\subset \mathbf{P}^3,
\]
and, by a linear-algebra fact, it follows that
\[
[\psi]\in \Sigma_{2,2}
\ \Longleftrightarrow\
\mathrm{rank}\,M(\psi)\le 1
\ \Longleftrightarrow\
\det M(\psi)=ad-bc=0.
\]
Thus the vanishing ideal of the product locus is the principal ideal
\[
I(\Sigma_{2,2})=(ad-bc)\subset \mathbf{C}[a,b,c,d],
\qquad
\Sigma_{2,2}=V(ad-bc).
\]
\end{example}

\begin{example}[General bipartite cut]
Let $H=H_A\otimes H_B$ with $\dim H_A=d_A$ and $\dim H_B=d_B$.  After choosing bases, write coordinates $(x_{ij})$ on $H$ so that the flattening matrix is
\[
M(\psi)=(x_{ij})\in \mathrm{Mat}_{d_A\times d_B}(\mathbf{C}).
\]
For indices $1\le i_1<i_2\le d_A$ and $1\le j_1<j_2\le d_B$, define the $2\times 2$ minors
\[
\Delta_{i_1,i_2;\,j_1,j_2}
:=
x_{i_1 j_1}x_{i_2 j_2}-x_{i_1 j_2}x_{i_2 j_1}
\ \in\ S_2.
\]
Let $I_{A|B}\subset S$ be the homogeneous ideal generated by all these minors:
\[
I_{A|B}:=\bigl\langle\,\Delta_{i_1,i_2;\,j_1,j_2}\ \bigm|\ i_1<i_2,\ j_1<j_2\,\bigr\rangle.
\]
Then the product locus for the bipartition $A|B$ is exactly
\[
\Sigma_{A|B}=V(I_{A|B})\subset \mathbf{P}(H),
\]
because $\Delta_{i_1,i_2;\,j_1,j_2}(\psi)=0$ for all minors is equivalent to $\mathrm{rank}\,M(\psi)\le 1$. So, \ $\psi$ is decomposable as $\psi=a\otimes b$ up to scale.
\end{example}

\subsubsection{OR/AND and the corresponding ideal operations}
\begin{proposition}[Dictionary between set operations and ideal operations]
\label{prop:ideal-dictionary}
For Zariski closed subsets $X,Y\subset P(H)$ in projective space, the following hold:
\begin{enumerate}
\item \textbf{OR (union):}
\[
I(X\cup Y)=I(X)\cap I(Y).
\]
\item \textbf{AND (intersection):}
\[
I(X\cap Y)=\left(\sqrt{\,I(X)+I(Y)\,}\right)^{\mathrm{sat}}.
\]
\end{enumerate}
Here $\mathfrak m:=S_{>0}$ is the irrelevant ideal, and the saturation of a homogeneous ideal $J$ is defined by
\[
J^{\mathrm{sat}}
:=\bigcup_{k\ge0}(J:\mathfrak m^k)
=\{\,f\in S\mid \exists k,\ \mathfrak m^k f\subset J\,\}.
\]
\end{proposition}

In the projective setting, a homogeneous ideal $J$ and its saturation $J^{\mathrm{sat}}$ may define the same zero locus, but they can differ as vanishing ideals. Statement (2) in the proposition means that to describe the \emph{vanishing ideal} of $X\cap Y$ itself, one needs saturation (and radical). As point sets, however, one always has
\[
V(I(X)+I(Y))=X\cap Y.
\]

\subsubsection{The biproduct locus for an $N$-partite system and primary decomposition}
In this subsection we first organize, for a general $N$-partite system, for which bipartitions does the state factorize? in terms of ideals, and then treat the simple but nontrivial case $N=3$ as an example.

From now on, let the base field be $k=\C$ (or an algebraically closed field), and fix
\[
H=H_1\otimes\cdots\otimes H_N,\qquad d_i:=\dim H_i\ge 2.
\]

\paragraph{Bipartitions and the definition of biproduct / GME states.}

\begin{definition}[Bipartition $A|A^c$ and biproduct / GME states]
\label{def:bisep-GME-general}
For a nonempty proper subset $\emptyset\neq A\subsetneq [N]$, write $A^c:=[N]\setminus A$ for its complement. The product locus
\[
\Sigma_{A|A^c}:=\Sigma_{\{A,A^c\}}\subset\PP(H)
\]
is the image of the Segre embedding
\[
\PP(H_A)\times\PP(H_{A^c})\hookrightarrow \PP(H),
\]
where $H_A:=\bigotimes_{i\in A}H_i$. We define:
\begin{itemize}
\item \emph{biproduct locus:} the set of pure states that are product with respect to \emph{some}
bipartition,
\[
X_{\mathrm{bisep}}^{(N)} :=\bigcup_{\emptyset\neq A\subsetneq [N]} \Sigma_{A|A^c}\subset\PP(H).
\]
\item \emph{GME locus} (genuinely multipartite entangled): the set of pure states that are not product with respect to \emph{any} bipartition,
\[
X_{\mathrm{GME}}^{(N)}:=\PP(H)\setminus X_{\mathrm{bisep}}^{(N)}.
\]
\end{itemize}
\end{definition}

Since $\Sigma_{A|A^c}=\Sigma_{A^c|A}$, one can avoid double counting by restricting to, for example,
\[
X_{\mathrm{bisep}}^{(N)}
=\bigcup_{\substack{\emptyset\neq A\subsetneq[N]\\ 1\in A}}\Sigma_{A|A^c},
\]
which gives the same set.

\begin{proposition}[Primary decomposition of the biproduct locus]
\label{prop:bisep-primary-general}
The vanishing ideal of $X_{\mathrm{bisep}}^{(N)}\subset\PP(H)$ is
\[
I\bigl(X_{\mathrm{bisep}}^{(N)}\bigr)
=\bigcap_{\emptyset\neq A\subsetneq[N]} I\bigl(\Sigma_{A|A^c}\bigr).
\]
Each $I(\Sigma_{A|A^c})$ is a prime ideal, so the right-hand side is a primary decomposition (in fact, an
\emph{intersection of prime ideals}).
\end{proposition}

\begin{proof}
By Proposition~\ref{prop:ideal-dictionary}(1), the ideal of a finite union of closed sets is the intersection:
\[
I\Bigl(\bigcup_{A}\Sigma_{A|A^c}\Bigr)=\bigcap_A I(\Sigma_{A|A^c}).
\]
Moreover, the Segre map is a closed immersion, and $\PP(H_A)\times\PP(H_{A^c})$ is irreducible (a product of irreducible varieties is irreducible), hence its image $\Sigma_{A|A^c}$ is also irreducible. The vanishing ideal of an irreducible projective variety is prime, so $I(\Sigma_{A|A^c})$ is prime. Since prime ideals are primary, the right-hand side gives a primary decomposition.
\end{proof}

\begin{remark}
Proposition~\ref{prop:bisep-primary-general} means that, in general,
\[
X_{\mathrm{bisep}}^{(N)}\ \text{is not irreducible, and is formed as a union of the}\ 
\Sigma_{A|A^c}.
\]
Thus biproductness carries a discrete ``OR'' structure:
\[
\text{\emph{which bipartition does it factorize in?}}
\]
and this becomes visible as an intersection of prime ideals.
\end{remark}

\paragraph{What happens if several partitions are satisfied simultaneously?}

\begin{definition}[Common refinement of partitions]
\label{def:partition-meet}
For two partitions $\pi=\{B_a\}$ and $\pi'=\{C_b\}$, the collection of all nonempty intersections
\[
B_a\cap C_b\qquad (1\le a\le \ell,\ 1\le b\le \ell')
\]
forms a partition again. This is called the \emph{common refinement (meet)} of $\pi$ and $\pi'$, and we
write
\[
\pi\wedge \pi'.
\]
This corresponds to the $\wedge$ operation in the lattice of partitions.
\end{definition}

\begin{lemma}[If a state is a product for two partitions, then it is a product for the meet]
\label{lem:intersection-is-meet}
For any partitions $\pi,\pi'$, as point sets one has
\[
\Sigma_\pi\cap \Sigma_{\pi'}=\Sigma_{\pi\wedge \pi'}.
\]
\end{lemma}

\begin{proof}
The inclusion $\Sigma_{\pi\wedge\pi'}\subset \Sigma_\pi\cap\Sigma_{\pi'}$ is clear, since $\pi\wedge\pi'$ is a common refinement of $\pi$ and $\pi'$: if one can factorize more finely, then one is both $\pi$-product and $\pi'$-product.

Conversely, let $[\psi]\in \Sigma_\pi\cap\Sigma_{\pi'}$ and choose a representative $\psi\neq 0$. Since we are in projective space, we may absorb scalars and assume the equality holds at the level of
vectors:
\[
\psi=\bigotimes_{B\in\pi} u_B \;=\; \bigotimes_{C\in\pi'} v_C
\qquad (u_B\in H_B\setminus\{0\},\ v_C\in H_C\setminus\{0\}).
\]

Fix a block $B_0\in\pi$ and set
\[
u_{B_0^c}:=\bigotimes_{B\in\pi,\ B\neq B_0} u_B \in H_{B_0^c}\setminus\{0\},
\]
where $B_0^c=[N]\setminus B_0$.

\medskip
\noindent
\emph{(Choice of a functional)}
For each block $C$ of $\pi'$, consider $C\cap B_0^c$ (ignore it if empty). Then
\[
H_{B_0^c}\;\cong\;\bigotimes_{C\in\pi'} H_{C\cap B_0^c},
\]
and hence
\[
H_{B_0^c}^\vee\;\cong\;\bigotimes_{C\in\pi'} H_{C\cap B_0^c}^\vee.
\]
Under this identification, \emph{decomposable (pure)} functionals $\varphi=\bigotimes_{C\in\pi'}\varphi_{C\cap B_0^c}$ linearly generate $H_{B_0^c}^\vee$. Since $u_{B_0^c}\neq 0$, there exists such a decomposable functional with
\[
\exists\ \varphi=\bigotimes_{C\in\pi'}\varphi_{C\cap B_0^c}\in H_{B_0^c}^\vee
\quad\text{s.t.}\quad \varphi(u_{B_0^c})\neq 0.
\]
(Otherwise all pure functionals would annihilate $u_{B_0^c}$, hence so would their linear span, i.e.\ all of $H_{B_0^c}^\vee$, a contradiction.) After scaling we may assume $\varphi(u_{B_0^c})=1$.

\medskip
\noindent
\emph{(Factorization along the meet)}
Apply the contraction map
\[
(\id_{H_{B_0}}\otimes \varphi): H_{B_0}\otimes H_{B_0^c}\to H_{B_0}
\]
to $\psi$. Using $\psi=u_{B_0}\otimes u_{B_0^c}$ gives
\[
(\id\otimes\varphi)(\psi)
=(\id\otimes\varphi)(u_{B_0}\otimes u_{B_0^c})
=\varphi(u_{B_0^c})\,u_{B_0}
=u_{B_0}.
\]
On the other hand, writing $\psi=\bigotimes_{C\in\pi'} v_C$ and using $\varphi=\bigotimes_{C\in\pi'}\varphi_{C\cap B_0^c}$, we obtain
\[
(\id\otimes\varphi)(\psi)
=\bigotimes_{C\in\pi'}
\Bigl(\ (\id_{H_{B_0\cap C}}\otimes \varphi_{C\cap B_0^c})(v_C)\ \Bigr)
\in \bigotimes_{C\in\pi'} H_{B_0\cap C}.
\]
Hence $u_{B_0}$ factorizes along the pieces $\{B_0\cap C\}_{C\in\pi'}$.

Doing this for every $B_0\in\pi$, we obtain
\[
\psi=\bigotimes_{B\in\pi} u_B
=\bigotimes_{B\in\pi}\ \bigotimes_{C\in\pi'} w_{B\cap C}
=\bigotimes_{D\in\pi\wedge\pi'} w_D,
\]
so $[\psi]\in \Sigma_{\pi\wedge\pi'}$.
\end{proof}

\begin{remark}
The set theoretic equality $\Sigma_\pi\cap\Sigma_{\pi'}=\Sigma_{\pi\wedge\pi'}$ is correct as a statement in linear algebra. In contrast, the scheme theoretic intersection (the closed subscheme defined by the ideal sum $I(\Sigma_\pi)+I(\Sigma_{\pi'})$) is generally more delicate, and in particular can contain nilpotents if the intersection is not transverse. In this note, when we discuss product/entangled criteria, we first consider the equality at the level of sets.
\end{remark}

\paragraph{A nontrivial example.}
Now let $N=3$ and $H=H_1\otimes H_2\otimes H_3$. There are only three bipartitions:
\[
1|23,\qquad 2|13,\qquad 3|12.
\]
Thus Definition~\ref{def:bisep-GME-general} becomes
\[
X_{\mathrm{bisep}}^{(3)}=\Sigma_{1|23}\cup\Sigma_{2|13}\cup\Sigma_{3|12},
\qquad
X_{\mathrm{GME}}^{(3)}=\PP(H)\setminus X_{\mathrm{bisep}}^{(3)}.
\]
In this note we also call $X_{\mathrm{bisep}}^{(3)}$ the \emph{biproduct locus}.

\begin{proposition}[Primary decomposition of the $3$-partite biproduct locus]
\label{prop:bisep-primary-3}
\[
I\bigl(X_{\mathrm{bisep}}^{(3)}\bigr)
= I(\Sigma_{1|23})\cap I(\Sigma_{2|13})\cap I(\Sigma_{3|12}).
\]
Each $I(\Sigma_{i|jk})$ is a prime ideal, so the right-hand side is a primary decomposition (in fact, a prime decomposition).
\end{proposition}

\begin{proof}
This is just Proposition~\ref{prop:bisep-primary-general} specialized to $N=3$.
\end{proof}

\begin{lemma}
\label{lem:two-cuts-imply-full-3}
Let $\psi\neq 0$. If
\[
[\psi]\in \Sigma_{1|23}\cap \Sigma_{2|13},
\]
then $[\psi]\in \Sigma_{1|2|3}$. Similarly, the intersection of any two distinct bipartitions satisfies
\[
\Sigma_{1|23}\cap \Sigma_{2|13}
=\Sigma_{1|23}\cap \Sigma_{3|12}
=\Sigma_{2|13}\cap \Sigma_{3|12}
=\Sigma_{1|2|3}
\qquad(\text{as point sets}).
\]
\end{lemma}

\begin{proof}
Using the general statement about meets of partitions (Lemma~\ref{lem:intersection-is-meet}), for example,
\[
(1|23)\wedge(2|13)=1|2|3.
\]
For $N=3$, the meet of two distinct bipartitions is always the full partition, hence
\[
\Sigma_{1|23}\cap\Sigma_{2|13}=\Sigma_{1|2|3}.
\]
\end{proof}

\begin{remark}[Why $N=3$ is special]
\label{rem:N3}
For $N=3$ there are only three bipartitions, and if a state is product for two distinct bipartitions, then their meet is the full partition. Hence
\[
\text{(product for one bipartition)}\land \text{(product for another bipartition)}
\ \Rightarrow\ \text{fully product}.
\]
However, for $N\ge 4$, even if two distinct bipartitions are satisfied, their meet need not refine all the way to $N$ blocks. It may become, for instance,
\[
1|2|\text{(the rest)}.
\]
Thus the state need not be fully product. For example, when $N=4$, if
\[
\psi=a_1\otimes a_2\otimes \Phi_{34}
\qquad(\Phi_{34}\in H_3\otimes H_4\ \text{is entangled}),
\]
then
\[
[\psi]\in\Sigma_{1|234}\cap\Sigma_{2|134},
\]
but $[\psi]\notin\Sigma_{1|2|3|4}$. This phenomenon (entanglement may persist in the remaining factors after taking partial marginals) is a genuinely new feature in the multipartite case ($N\ge 4$).
\end{remark}

\subsubsection{Algebraic properties of GME}

\paragraph{Bipartition entanglement is an ``OR of minors $\neq0$.''}

\begin{remark}[Bipartition $A|A^c$: product is AND, entangled is OR]\label{rem:entangled_is_OR}
\label{rem:OR-of-minors-general}
For a bipartition $A|A^c$, the locus $\Sigma_{A|A^c}\subset\PP(H)$ is (via flattening) the set of rank-$1$ tensors, and its ideal is generated by
\[
I(\Sigma_{A|A^c})=\langle\ 2\times2\text{ minors of }M_{A|A^c}\ \rangle. 
\]
Writing the set of generators as $\mathcal{M}_{A|A^c}$, we have
\[
\Sigma_{A|A^c}=V\bigl(I(\Sigma_{A|A^c})\bigr)=\bigcap_{\Delta\in\mathcal{M}_{A|A^c}}V(\Delta).
\]
Hence the complement is
\[
\PP(H)\setminus \Sigma_{A|A^c}
=\bigcup_{\Delta\in\mathcal{M}_{A|A^c}} D(\Delta),
\qquad
D(\Delta):=\{[\psi]\in\PP(H)\mid \Delta(\psi)\neq 0\}.
\]
As logical statements,
\[
\text{(product in $A|A^c$)}\iff \bigwedge_{\Delta\in\mathcal{M}_{A|A^c}}(\Delta=0),
\qquad
\text{(entangled in $A|A^c$)}\iff \bigvee_{\Delta\in\mathcal{M}_{A|A^c}}(\Delta\neq 0).
\]
That is, \emph{the nonvanishing of at least one minor} is an elementary criterion for being entangled with respect to that bipartition.
\end{remark}

\paragraph{$N$-partite GME is an ``AND of OR over all bipartitions.''}

\begin{remark}[GME is an ``AND of OR'' (general $N$)]
\label{rem:AND-of-OR-general}
By Definition~\ref{def:bisep-GME-general},
\[
X_{\mathrm{GME}}^{(N)}
=\PP(H)\setminus\Bigl(\ \bigcup_{\emptyset\neq A\subsetneq[N]}\Sigma_{A|A^c}\ \Bigr)
=\bigcap_{\emptyset\neq A\subsetneq[N]}\bigl(\PP(H)\setminus\Sigma_{A|A^c}\bigr).
\]
Each $\PP(H)\setminus\Sigma_{A|A^c}$ can be written as
$\bigcup_{\Delta\in\mathcal{M}_{A|A^c}}D(\Delta)$ by
Remark~\ref{rem:OR-of-minors-general}, so
\[
[\psi]\in X_{\mathrm{GME}}^{(N)}
\quad\Longleftrightarrow\quad
\bigwedge_{\emptyset\neq A\subsetneq[N]}
\Bigl(\ \bigvee_{\Delta\in\mathcal{M}_{A|A^c}}(\Delta(\psi)\neq 0)\ \Bigr).
\]
Thus GME can be characterized as follows:
\begin{center}
\emph{for every bipartition, at least one $2\times2$ minor is nonzero.}
\end{center}
\end{remark}

\begin{remark}[Important: the entangled locus is not a zero locus]
\label{rem:ent-open-not-closed}
The loci $\Sigma_{A|A^c}$ and $X_{\mathrm{bisep}}^{(N)}$ are \emph{closed sets} defined by equations, so they can be described by ideals (as $V(I)$). In contrast, their complement $X_{\mathrm{GME}}^{(N)}$ is in general an \emph{open set}.

Since projective space $\PP(H)$ is irreducible, the only subsets that are both open and closed are the empty set and the whole space. Therefore,
\[
\boxed{
\begin{minipage}{0.99\linewidth}
it is algebro--geometrically natural to describe non-entangled (product) states by ideals, and view entangled states as the complement (an open locus).
\end{minipage}
}
\]
\end{remark}

\paragraph{GME is Zariski open.}

\begin{proposition}[The GME locus is Zariski open and dense (general $N$)]
\label{prop:GME-open-dense-general}
Assume $d_i\ge 2$. Then $X_{\mathrm{GME}}^{(N)}$ is Zariski open and moreover dense. In particular, $X_{\mathrm{GME}}^{(N)}$ is nonempty.
\end{proposition}

\begin{proof}
The locus $X_{\mathrm{bisep}}^{(N)}$ is a finite union of closed sets (Segre varieties), hence it is closed, and its complement $X_{\mathrm{GME}}^{(N)}$ is open.

Moreover, each $\Sigma_{A|A^c}\subset\PP(H)$ is a Segre variety, so $\dim \Sigma_{A|A^c}=(\dim H_A-1)+(\dim H_{A^c}-1)$, while $\dim\PP(H)=\dim H-1$. If $d_i\ge2$, then $\dim H=(\dim H_A)(\dim H_{A^c})$, and
\[
(\dim H-1)-\dim\Sigma_{A|A^c}=(\dim H_A-1)(\dim H_{A^c}-1)>0.
\]
Thus $\Sigma_{A|A^c}$ is always a proper closed subset. Hence the finite union $X_{\mathrm{bisep}}^{(N)}$ is also proper closed, and therefore its complement $X_{\mathrm{GME}}^{(N)}$ is dense.
\end{proof}

\section{Observation: A brief number-theoretic analogy}\label{sec:integer}

Since the quantum states considered in this note are restricted to pure states, there are places where quantum states can look as if they behave like ``primes''. In particular, the central question of this note---whether a state factorizes (whether it is a product state or an entangled state)---is not entirely unlike the classical question people have pondered since antiquity: whether an integer admits a prime factorization. (The question ``does a quantum state factorize?'' will raise foundational questions in the near future.) In this section we therefore think a little about how far this intuitive analogy goes, and where it breaks down.

Since addition of states as vectors is already defined, we keep it as is (we ignore overall scalars, and if an operation produces the zero vector we regard the state as ``vanishing''). We reinterpret the tensor product as the analogue of multiplication of integers.

\paragraph{Related literature.}
In this note we relate number theory and quantum separability from the viewpoint of ``factorization,'' but the structural appearance of prime factorization and the Chinese remainder theorem (CRT) in finite-dimensional quantum systems has also been discussed in earlier work \cite{EllinasFloratos1999PrimeDecompositionCorrelation}. There are also proposals to organize inclusion relations of finite quantum systems by divisibility $m\mid n$ and to build a $T_0$ topology from this structure \cite{Vourdas2012PartialOrderT0Topology}. Moreover, there is the ``prime state'', a construction that builds quantum states from the prime sequence itself and analyzes their entanglement \cite{LatorreSierra2020PrimeState}.

\medskip
\paragraph{Primes.}
In number theory, the word ``prime'' has at least the following two meanings.
\begin{itemize}
\item \textbf{Prime element--type ``prime'':}
a ``prime'' in the sense that an element cannot be decomposed multiplicatively. A typical example is a prime number $p\in\mathbb Z$, characterized by
\[
p=ab \ \Rightarrow\ a=\pm1 \ \text{or}\ b=\pm1.
\]
\item \textbf{Prime ideal--type ``prime'':}
a \emph{prime ideal} $\mathfrak p$ of a commutative ring $R$ is an ideal satisfying $\mathfrak p\subsetneq R$ and, for all $f,g\in R$,
\[
fg\in\mathfrak p\ \Rightarrow\ f\in\mathfrak p \ \text{or}\ g\in\mathfrak p.
\]
Moreover, if we equip $X=\Spec R$ with the Zariski topology, then
\[
V(\mathfrak p):=\{\mathfrak q\in \Spec R\mid \mathfrak p\subseteq \mathfrak q\}
=\overline{\{\mathfrak p\}}
\]
is an irreducible closed subset, and $\mathfrak p$ is the \emph{generic point} of $V(\mathfrak p)$. Conversely, any irreducible closed subset $Z\subset \Spec R$ has a unique generic point $\eta\in Z$ such that
\[
Z=\overline{\{\eta\}}=V(\eta).
\]
Hence prime ideals and irreducible closed subsets correspond bijectively via
\[
\mathfrak p \longleftrightarrow V(\mathfrak p)
\]
with inverse $Z\mapsto I(Z)$.
\end{itemize}
Since $\mathbb Z$ is a UFD (unique factorization domain), the above notions of ``prime'' coincide in $\mathbb Z$. In general, however, they do not coincide.

\medskip
Accordingly, from now on we consider these two meanings of ``prime'' separately for sets of quantum states.

\paragraph{(1) Prime element--type ``prime'' (relative irreducibility): not decomposable under tensor product (``multiplication'').}

\begin{definition}[Product / entangled with respect to a bipartition (a property of a ``point'')]
Let $H=H_1\otimes\cdots\otimes H_N$. For a nonempty proper subset $\emptyset\neq A\subsetneq [N]:=\{1,\dots,N\}$, set
\[
H_A:=\bigotimes_{i\in A}H_i,\qquad H_{A^c}:=\bigotimes_{j\in A^c}H_j.
\]
A pure state $[\psi]\in\PP(H)$ is said to be
\begin{itemize}
\item \emph{$A|A^c$-product (product in the cut $A|A^c$)} if there exist nonzero vectors $0\neq u\in H_A$ and $0\neq v\in H_{A^c}$ such that
\[
[\psi]=[u\otimes v].
\]
\item \emph{$A|A^c$-entangled (entangled in the cut $A|A^c$)} if it is not $A|A^c$-product.
\end{itemize}
\end{definition}

Primality of an integer is absolute with respect to the multiplicative structure. By contrast, whether a pure state can (or cannot) factorize depends on the \emph{chosen} tensor product structure and also on \emph{which bipartition $A|A^c$ one considers}. Thus, in the sense of
\[
\text{there is no tensor factorization with respect to a fixed cut,}
\]
one may say that an entangled state is ``prime-like.''

\medskip
Just as the prime factorization of a natural number is uniquely determined, the decomposition of a fully product state is also uniquely determined (in projective space).

\begin{proposition}[Projective uniqueness of the fully product decomposition]
\label{prop:rank1-Nfactor}
Let $H=H_1\otimes\cdots\otimes H_N$, and let $v_i,w_i\in H_i\setminus\{0\}$.
If
\[
v_1\otimes\cdots\otimes v_N \;=\; w_1\otimes\cdots\otimes w_N\ \in H,
\]
then there exist $\lambda_1,\dots,\lambda_N\in k^\times$ such that
\[
w_i=\lambda_i v_i\quad(1\le i\le N),
\qquad
\prod_{i=1}^N \lambda_i = 1.
\]
In particular, the projective points are uniquely determined:
\[
[v_i]=[w_i]\in\PP(H_i)\qquad(1\le i\le N).
\]
\end{proposition}

\begin{proof}
Fix $i$, and for each $j\neq i$ choose $\varphi_j\in H_j^\vee$ such that $\varphi_j(v_j)=1$. By the universal property of the tensor product, this defines a contraction map
\[
T_i:=\bigotimes_{j\neq i}\varphi_j:\ H\longrightarrow H_i,
\]
viewing $H=\bigl(\otimes_{j\neq i}H_j\bigr)\otimes H_i$.
Then
\[
T_i(v_1\otimes\cdots\otimes v_N)=v_i.
\]
On the other hand,
\[
T_i(w_1\otimes\cdots\otimes w_N)
=\Bigl(\prod_{j\neq i}\varphi_j(w_j)\Bigr)\,w_i.
\]
Since these are equal, we have $v_i=\alpha_i w_i$ for some $\alpha_i\neq 0$, hence $w_i=\lambda_i v_i$ with $\lambda_i=\alpha_i^{-1}$.

Finally,
\[
w_1\otimes\cdots\otimes w_N
=(\lambda_1\cdots\lambda_N)\,(v_1\otimes\cdots\otimes v_N),
\]
so comparing with the original equality yields $\prod_i\lambda_i=1$.
\end{proof}

\begin{corollary}[Injectivity of the Segre map]
\label{cor:Segre-injective}
The Segre map
\[
\mathrm{Seg}:\ \PP(H_1)\times\cdots\times\PP(H_N)\longrightarrow \PP(H),
\qquad ([v_1],\dots,[v_N])\longmapsto [v_1\otimes\cdots\otimes v_N]
\]
is injective. Consequently, for a fully product state $[\psi]\in\Sigma_{1|\cdots|N}:=\im(\mathrm{Seg})$, the tuple of projective points $([v_1],\dots,[v_N])$ satisfying
\[
[\psi]=[v_1\otimes\cdots\otimes v_N]
\]
is \emph{uniquely} determined.
\end{corollary}

\medskip
Next, we explain that the information of which partitions a state $[\psi]$ can factorize with respect to its separability pattern can be uniquely organized by a finest partition. The statements below are  correct as point sets.

\begin{definition}
\label{def:pattern-finest}
For a state $[\psi]\in\PP(H)$, define
\[
\mathcal P([\psi])\;:=\;\{\ \pi\ \text{a partition}\mid [\psi]\in\Sigma_\pi\ \}
\]
to be the set of partitions for which $[\psi]$ is a product state (it is nonempty since it always contains the trivial partition $\{[N]\}$).

\smallskip
Define the \emph{finest product partition of $[\psi]$} by
\[
\pi_{\min}([\psi])\;:=\;\bigwedge_{\pi\in\mathcal P([\psi])}\ \pi,
\]
where $\wedge$ denotes the common refinement.
\end{definition}

\begin{theorem}[Existence and uniqueness of the finest partition]
\label{thm:finest-partition-unique}
Let $[\psi]\in\PP(H)$ be arbitrary. Then:
\begin{enumerate}
\item (Existence) $[\psi]$ is a product state with respect to $\pi_{\min}([\psi])$:
\[
[\psi]\in \Sigma_{\pi_{\min}([\psi])}.
\]
\item (Uniqueness) The partitions $\rho$ for which $[\psi]$ is a product state are exactly the \emph{coarser partitions} of $\pi_{\min}([\psi])$ (those obtained by merging blocks). That is, for any partition $\rho$,
\[
[\psi]\in\Sigma_\rho
\quad\Longleftrightarrow\quad
\rho \ \text{is obtained by merging blocks of }\ \pi_{\min}([\psi]).
\]
\end{enumerate}
Hence the finest partition (the separability pattern) is \emph{unique}, and all information about the product structure of $[\psi]$ can be recovered from $\pi_{\min}([\psi])$ alone.
\end{theorem}

\begin{proof}
Since the total number of partitions is finite for fixed $N$, the set $\mathcal P([\psi])$ is finite. Hence $\wedge_{\pi\in\mathcal P([\psi])}\pi$ is well-defined.

\smallskip
(1) For any $\pi,\rho\in\mathcal P([\psi])$, we have $[\psi]\in\Sigma_\pi\cap\Sigma_\rho$, so by Lemma~\ref{lem:intersection-is-meet} we get $[\psi]\in\Sigma_{\pi\wedge\rho}$. Repeating this finitely many times yields
\[
[\psi]\in\Sigma_{\wedge_{\pi\in\mathcal P([\psi])}\pi}
=\Sigma_{\pi_{\min}([\psi])}.
\]

\smallskip
(2) If $\rho$ is obtained by merging blocks of $\pi_{\min}([\psi])$ (thus $\rho$ is coarser), then a $\pi_{\min}([\psi])$-product decomposition becomes a $\rho$-product decomposition simply by regrouping some tensor factors, so $[\psi]\in\Sigma_\rho$.

Conversely, assume $[\psi]\in\Sigma_\rho$. Then $\rho\in\mathcal P([\psi])$. By definition, $\pi_{\min}([\psi])$ is the common refinement of all elements of $\mathcal P([\psi])$, hence in particular $\pi_{\min}([\psi])$ refines $\rho$ (equivalently, $\rho$ is coarser than $\pi_{\min}([\psi])$). Therefore $\rho$ is obtained by merging blocks of $\pi_{\min}([\psi])$.
\end{proof}

The statement and meaning of this theorem can be interpreted as follows. If $\pi_{\min}([\psi])=\{B_1,\dots,B_k\}$, then $[\psi]$ can be written as
\[
[\psi]=[\psi_{B_1}\otimes\cdots\otimes \psi_{B_k}],
\]
and this block decomposition is projectively unique.

Moreover, since $\pi_{\min}([\psi])$ is the finest, each block state $[\psi_{B_a}]$ cannot be further factorized inside that block. Indeed, if $B_a$ admitted a nontrivial further partition and $[\psi_{B_a}]$ were product with respect to it, then the whole partition $\pi_{\min}([\psi])$ could be refined further, a contradiction. Thus each $B_a$ can be interpreted as a ``genuinely entangled block''
internally.

\begin{example}[The case $N=3$]
Let us examine the case $N=3$. In this case there are only five partitions:
\[
1|2|3,\quad 1|23,\quad 2|13,\quad 3|12,\quad 123.
\]
Hence $\pi_{\min}([\psi])$ is always one of these.
\begin{itemize}
\item If $\pi_{\min}([\psi])=1|2|3$, then $[\psi]$ is fully product.
\item If $\pi_{\min}([\psi])=1|23$ (or cyclic permutations), then $[\psi]$ is biproduct
(product only in some bipartition).
\item If $\pi_{\min}([\psi])=123$, then $[\psi]$ is GME (not product in any bipartition).
\end{itemize}
As noted in Remark~\ref{rem:N3}, for $N=3$ the meet of two distinct bipartitions such as $1|23$ and $2|13$ is $1|2|3$, so being product in two different bipartitions implies being fully product. For $N\ge 4$ this generally fails (the meet can be of the form $1|2|\text{(the rest)}$).
\end{example}

\paragraph{(2) Prime ideal--type ``prime'': the product locus is an irreducible closed set, and its vanishing ideal is prime.}

As we already saw in Proposition~\ref{prop:bisep-primary-general}, for a bipartition $A|A^c$ the
product locus
\[
\Sigma_{A|A^c}\subset \PP(H)
\]
is an irreducible closed subset, and its vanishing ideal
\[
I_{A|A^c}:=I(\Sigma_{A|A^c})\subset S
\]
is a \textbf{prime ideal}. Moreover, the biproduct locus from
Definition~\ref{def:bisep-GME-general},
\[
X_{\mathrm{bisep}}^{(N)}
:=\bigcup_{\emptyset\neq A\subsetneq [N]}\Sigma_{A|A^c}
\subset \PP(H),
\]
admits a primary decomposition -- corresponding to decomposition into irreducible components-- 
(Proposition~\ref{prop:bisep-primary-general}):
\[
I\!\left(X_{\mathrm{bisep}}^{(N)}\right)
=\bigcap_{\emptyset\neq A\subsetneq [N]} I_{A|A^c}.
\]

Recall the basic example $\Spec\mathbb Z$ in number theory: for an integer $n\neq 0$,
\[
V(n):=\{\mathfrak p\in\Spec\mathbb Z\mid (n)\subset \mathfrak p\}
=\{(p)\mid p\ \text{is prime and}\ p\mid n\},
\]
so
\[
V(n)=\bigcup_{p\mid n}V(p),\qquad
I(V(n))=\bigcap_{p\mid n}(p).
\]

\smallskip
In exactly the same form, on the quantum side we have
\[
X_{\mathrm{bisep}}^{(N)}
=\bigcup_{A|A^c}\Sigma_{A|A^c},
\qquad
I\!\left(X_{\mathrm{bisep}}^{(N)}\right)=\bigcap_{A|A^c} I_{A|A^c}.
\]
In this sense, the analogy is quite close.

\smallskip
However, an important difference is that in $\Spec\mathbb Z$ a ``prime number $p$'' directly generates the ``prime ideal $(p)$'', whereas in $\PP(H)$ a ``state $[\psi]$'' does not naturally produce a prime ideal of the kind above. Another difference is that in number theory $(p)$ is a closed point in $\Spec\mathbb Z$, while the GME locus is Zariski open
(Proposition~\ref{prop:GME-open-dense-general}).

\begin{remark}[A remark for careful readers]
\begin{itemize}
\item $V(n)$ remembers only the set of prime divisors of $n$ and forgets exponents (the $p$-adic valuation). For instance, $V(p^k)=V(p)$. Thus the comparison above is not with prime factorization itself, but rather with the level of radicals / sets of irreducible components.

\item A point $[\psi]\in\PP(H)$ corresponds to the point ideal $I([\psi])$ (a maximal homogeneous ideal), which is prime. However, the prime ideals relevant to quantum states here are those like $I_{A|A^c}=I(\Sigma_{A|A^c})$, thus prime ideals defining product loci (closed sets defined by equations). In that sense, the roles are different: a single state $[\psi]$ does not naturally generate a prime ideal that encodes ``productness'' in the above way.
\end{itemize}
\end{remark}

\paragraph{Summary.}
Summarizing the discussion, one may say the following.
\begin{enumerate}
\item \textbf{Prime element--type ``prime'':}
with respect to a fixed bipartition, or with respect to all bipartitions, one can compare the GME property---\emph{not decomposable under tensor product}---to the notion of an integer being ``indecomposable (prime).'' However, this is a relative irreducibility depending on the chosen partition, not an absolute notion as in number theory. (If a factorization exists, it is unique projectively.)

\item \textbf{Prime ideal--type ``prime'':}
the product locus $\Sigma_{A|A^c}$ is an irreducible closed subset, hence $I_{A|A^c}$ is a prime ideal. The biproduct locus is the union of these loci, and its vanishing ideal admits a primary decomposition as an \emph{intersection of prime ideals}.
\end{enumerate}
Thus, as a slogan,
\[
\boxed{\ \text{GME is the set of points that do not lie in the union of the prime components given by product loci.}\ }
\]

In this section we took a small peek at the worlds of quantum states and number theory through the lens of multiplication. In \S\ref{sec:secant-tensor-rank}, we will make further observations from the viewpoint of addition.

\section{Additive decomposition: tensor rank and secant varieties}
\label{sec:secant-tensor-rank}

In this section, as a viewpoint corresponding to the ``multiplicative side'' discussed so far (whether one can factorize under the tensor product), we introduce the \textbf{additive side} (\emph{superposition}). Just as number theory has the basic pair of problems
\[
\text{multiplication: prime factorization}
\qquad\text{addition: Waring's problem (how many summands?)},
\]
for quantum states one can naturally consider:
\begin{align}
\begin{aligned}
\text{multiplication}&: \text{can the state be decomposed as a product state?}\notag\\
\text{addition}&: \text{how many product states are needed in a sum?}
\end{aligned}
\end{align}
We treated the ``multiplicative'' viewpoint in \S\ref{sec:integer}. Now we turn to the ``additive'' viewpoint. A standard reference for this subsection is \cite{landsberg2011tensors}.

\subsection{Measuring by sums of product states: tensor rank}
\label{subsec:tensor-rank}

Fix a subsystem decomposition
\[
H \;=\; H_1\otimes\cdots\otimes H_N,
\]
and regard pure states as points of the projective space $\PP(H)$.

\begin{definition}[Tensor rank]
For $0\neq \psi\in H$, define the \emph{tensor rank} $\mathrm{rank}_\otimes(\psi)$ by
\[
\mathrm{rank}_\otimes(\psi)
:=
\min\left\{\,r\ \middle|\ 
\psi=\sum_{i=1}^{r} v_{i1}\otimes\cdots\otimes v_{iN}\ \text{for some }v_{ij}\in H_j\right\}.
\]
For a projective state $[\psi]$, this does not change under scalar multiplication, so we may write
\[
\mathrm{rank}_\otimes([\psi]) := \mathrm{rank}_\otimes(\psi).
\]
\end{definition}

In general, even if one writes
\[
\psi=\sum_{i=1}^r c_i\,(v_{i1}\otimes\cdots\otimes v_{iN})
\qquad(c_i\in\C),
\]
one can absorb $c_i$ into (say) $v_{i1}$ in each term:
\[
c_i\,(v_{i1}\otimes\cdots\otimes v_{iN})
=(c_i v_{i1})\otimes v_{i2}\otimes\cdots\otimes v_{iN}.
\]
Thus the coefficient-free form used in the definition is sufficient.

\smallskip
When $N=2$, $H=H_A\otimes H_B$ can be identified with matrices, and rank-$1$ tensors correspond to rank-$1$ matrices. Hence $\mathrm{rank}_\otimes(\psi)$ agrees with the matrix (Schmidt) rank (Theorem~\ref{thm:SR_is_rk}).

\subsection{A filtration as closed sets: secant varieties and border rank}
\label{subsec:secant-variety}

Tensor rank measures whether a state can be written as a \emph{finite sum} of product states, but in multipartite systems an important phenomenon occurs:
\[
\text{the set ``rank $\le r$'' is not Zariski closed.}
\]
Since algebraic geometry naturally treats objects defined by equations (closed sets), we introduce the notion obtained by taking the Zariski closure (border rank).

\begin{remark}[Failure of closedness typically occurs for $N\ge 3$]
In the bipartite case ($N=2$), the set $\{[\psi]\mid \mathrm{rank}_\otimes(\psi)\le r\}$ is a determinantal locus defined by $(r+1)\times(r+1)$ minors, hence Zariski closed (see \S\ref{subsec:secant-bipartite}). In contrast, for $N\ge 3$ this set is generally not closed, and tensor rank can drop in a limit (the latter example of the $W$ state is a typical instance).
\end{remark}

\begin{definition}[$r$-th secant variety]
For a projective variety $X\subset\PP^m$, define its \emph{$r$-th secant variety} $\sigma_r(X)$ by
\[
\sigma_r(X)
:=
\overline{
\bigcup_{x_1,\dots,x_r\in X}
\langle x_1,\dots,x_r\rangle
}
\ \subset\ \PP^m.
\]
Here $\langle x_1,\dots,x_r\rangle$ denotes the smallest projective linear subspace containing them (a secant line for $r=2$, a secant plane for $r=3$), and $\overline{(\ \cdot\ )}$ denotes the Zariski closure.
\end{definition}

\begin{remark}
The bar $\overline{(\cdot)}$ in the definition of $\sigma_r(X)$ is the Zariski closure. Accordingly, the condition $[\psi]\in \sigma_r(X)$ can be rephrased, for example, as follows: there exists a one-parameter family $[\psi(t)]$ ($t\in k^\times$) such that for each $t\neq 0$, $[\psi(t)]$ lies on a projective linear subspace $\langle x_1(t),\dots,x_r(t)\rangle$ spanned by $r$ points of $X$, and (allowing scalar rescaling in projective space) $[\psi]$ appears as the limit $t\to 0$.
\end{remark}

\begin{definition}[Border rank (secant rank)]
For $0\neq \psi\in H$, define the \emph{border rank} (the \emph{secant rank}) by
\[
\underline{\mathrm{rank}}_\otimes(\psi)
:=
\min\{\,r \mid [\psi]\in \sigma_r(\Sigma)\,\}.
\]
The same definition applies to a projective state $[\psi]$.
\end{definition}

The inequality $\underline{\mathrm{rank}}_\otimes(\psi)\le r$ means that $[\psi]$ can be obtained as a limit of points lying in projective spaces spanned by $r$ product states. Hence, in general,
\[
\underline{\mathrm{rank}}_\otimes(\psi)\le \mathrm{rank}_\otimes(\psi),
\]
since allowing limits can reduce the required number of summands.

\begin{proposition}[Tensor rank, border rank, and secant varieties]
\label{prop:rank-vs-secant}
Let $0\neq \psi\in H$. Then:
\begin{enumerate}
\item If $\mathrm{rank}_\otimes(\psi)\le r$, then $[\psi]\in \sigma_r(\Sigma)$.
\item $[\psi]\in \sigma_r(\Sigma)\ \Longleftrightarrow\ \underline{\mathrm{rank}}_\otimes(\psi)\le r$.
\end{enumerate}
\end{proposition}

\begin{proof}
(1) If $\psi=\sum_{i=1}^r \psi_i$ with each $\psi_i$ a product state, then $[\psi]$ lies in the projective linear subspace $\langle[\psi_1],\dots,[\psi_r]\rangle$ spanned by $[\psi_1],\dots,[\psi_r]\in\Sigma$. Hence $[\psi]$ lies in the union of such linear subspaces, and therefore in its Zariski closure $\sigma_r(\Sigma)$.

(2) This is exactly the definition of border rank.
\end{proof}

Since $\sigma_r(\Sigma)$ is Zariski closed (hence a projective variety), its vanishing ideal
\[
I(\sigma_r(\Sigma))
\]
is defined. In this way, the filtration of states that can be approximated by $\le r$ product states is obtained as an object that can be detected by equations (homogeneous polynomials).

\subsection{The bipartite case}
\label{subsec:secant-bipartite}

The bipartite filtration $R_{\le r}$ discussed in \S\ref{sec:variety} coincides exactly with the
secant varieties.

\begin{proposition}[Bipartite case: $\sigma_r(\Sigma)=R_{\le r}$]
\label{prop:secant=det}
Let $H=H_A\otimes H_B$, and let $\Sigma=\PP(H_A)\times\PP(H_B)\subset\PP(H)$ be the Segre variety.
Then for any $r$,
\[
\sigma_r(\Sigma)=R_{\le r}
\]
holds. (When $r\ge \min(d_A,d_B)$, both sides equal $\PP(H)$, so this is trivial.) In particular, in the bipartite case,
\[
\underline{\mathrm{rank}}_\otimes(\psi)=\mathrm{rank}_\otimes(\psi)=\mathrm{SR}(\psi)
\]
holds (rank and border rank coincide).
\end{proposition}

\begin{proof}
$\sigma_r(\Sigma)\subset R_{\le r}$: By definition, $\sigma_r(\Sigma)$ is the Zariski closure of the union of projective subspaces spanned by $r$ rank-$1$ matrices. Any point in this union clearly has matrix rank $\le r$, and the condition ``matrix rank $\le r$'' is a closed condition defined by the vanishing of $(r+1)$-minors. Hence it remains true after taking Zariski closure, so $\sigma_r(\Sigma)\subset R_{\le r}$.

$\sigma_r(\Sigma)\supset R_{\le r}$: Any matrix of rank $\le r$ can be decomposed (by basic linear algebra) as a sum of at most $r$ rank-$1$ matrices. Therefore any point of $R_{\le r}$ lies on a projective linear subspace spanned by $r$
points of $\Sigma$, hence belongs to $\sigma_r(\Sigma)$.
\end{proof}

In the bipartite case, the ideal defined by minors therefore gives
\[
I(\sigma_r(\Sigma))=I(R_{\le r}).
\]

\subsection{In multipartite systems: rank and border rank can differ}
\label{subsec:rank-vs-border-example}

In multipartite systems, since the tensor-rank loci are not closed, a typical phenomenon is that rank and border rank differ. As a nontrivial example, we consider the $W$ state of $3$ qubits and show that $\mathrm{rank}_\otimes=3$ while $\underline{\mathrm{rank}}_\otimes=2$.

Let $H=(\C^2)^{\otimes 3}$ and consider
\[
|W\rangle:=|001\rangle+|010\rangle+|100\rangle.
\]
This state has tensor rank $3$ but border rank $2$.

To see border rank $\le 2$, consider the $t$-dependent state
\[
\psi(t):=(|0\rangle+t|1\rangle)^{\otimes 3}-|000\rangle.
\]
The right-hand side is the difference of two product states, so for $t\neq 0$ we have $\mathrm{rank}_\otimes(\psi(t))\le 2$. Expanding,
\[
\psi(t)=
t(|001\rangle+|010\rangle+|100\rangle)+t^2(\cdots)+t^3|111\rangle.
\]
Since projective space identifies scalar multiples,
\[
[\psi(t)]
=\left[\frac{1}{t}\psi(t)\right]
=\bigl[\,|001\rangle+|010\rangle+|100\rangle+t(\cdots)+t^2|111\rangle\,\bigr]
\qquad(t\neq 0),
\]
so $[\psi(t)]\to [W]$ as $t\to 0$. (And since $W$ is not a product state, this implies $\underline{\mathrm{rank}}_\otimes(W)=2$.)

\begin{remark}[A simple proof that $\mathrm{rank}_\otimes(W)=3$]
From $W=|001\rangle+|010\rangle+|100\rangle$, it is clear that $\mathrm{rank}_\otimes(W)\le 3$.

Conversely, assume $\mathrm{rank}_\otimes(W)\le 2$. Then there exists a decomposition
\[
W=u_1\otimes v_1\otimes w_1+u_2\otimes v_2\otimes w_2.
\]
View $W$ as an element of $H_1\otimes (H_2\otimes H_3)$ and consider the contraction map $W:H_1^\vee\to H_2\otimes H_3$. Then $\mathrm{Im}(W)$ is contained in $\mathrm{span}\{v_1\otimes w_1,\ v_2\otimes w_2\}$. In particular, if $W$ had rank $2$, its image would have to be generated by two linearly independent rank-$1$ (product) vectors (if they were dependent, $W$ would have rank $1$).

However, a direct computation shows that for the basis $\{|0\rangle,|1\rangle\}$ of $H_1$ and the dual basis $\{\varepsilon_0,\varepsilon_1\}\subset H_1^\vee$,
\[
\varepsilon_0(W)=|01\rangle+|10\rangle,\qquad
\varepsilon_1(W)=|00\rangle,
\]
hence
\[
\mathrm{Im}(W)=\mathrm{span}\{|00\rangle,\ |01\rangle+|10\rangle\}\subset H_2\otimes H_3.
\]
Within this $2$-dimensional subspace, the only rank-$1$ vectors are scalar multiples of $|00\rangle$. Indeed, identifying $H_2\otimes H_3\simeq \Mat_{2\times 2}$, a vector $a|00\rangle+b(|01\rangle+|10\rangle)$ corresponds to $\begin{psmallmatrix}a&b\\ b&0\end{psmallmatrix}$, whose determinant is $-b^2$. Thus it has rank $1$ ($\det=0$) only when $b=0$. This contradicts $\mathrm{rank}_\otimes(W)\le 2$.

Therefore $\mathrm{rank}_\otimes(W)\ge 3$, and hence $\mathrm{rank}_\otimes(W)=3$.
\end{remark}

This example of the $W$ state suggests the following.
\begin{itemize}
\item \textbf{Tensor rank} measures whether one can write the state \emph{exactly} as a sum of $r$ terms (but the locus is generally not closed).
\item \textbf{Secant varieties} measure whether one can \emph{approximate} the state by a sum of $r$
terms (they are closed and have a vanishing ideal).
\end{itemize}

\subsection{The ideal of a secant variety}
\label{subsec:secant-ideal}

\begin{definition}[Vanishing ideal of a secant variety]
For the secant variety $\sigma_r(\Sigma)\subset\PP(H)$, write its vanishing ideal as
\[
I\bigl(\sigma_r(\Sigma)\bigr)\subset S.
\]
This is the set of homogeneous polynomial equations that must be satisfied by states of border rank $\le r$.
\end{definition}

In general it is difficult to write down a complete set of generators for $I(\sigma_r(\Sigma))$, but the minors coming from flattenings remain a powerful tool, as follows.

If $\underline{\mathrm{rank}}_\otimes(\psi)\le r$, then for \emph{every} bipartition $A|B$, any
flattening must satisfy
\[
\rk\bigl(M_{A|B}(\psi)\bigr)\le r.
\]
Indeed, if $\psi$ is a limit of $r$ rank-$1$ tensors, then after any flattening it becomes a limit of matrices of rank $\le r$, and the matrix rank condition is closed. (More precisely: a flattening is induced by a linear map $H\to \Hom(H_B^\vee,H_A)$ and hence a morphism of projective spaces. The locus of matrices of rank $\le r$ is a determinantal Zariski closed subset, so its preimage is closed and contains $\sigma_r(\Sigma)$.)

Therefore, for any bipartition $A|B$,
\[
\text{all $(r+1)\times(r+1)$ minors vanish.}
\]
Equivalently,
\[
\sigma_r(\Sigma)\ \subseteq\ 
\bigcap_{A|B} V\bigl(I_{A|B}^{(r+1)}\bigr),
\]
where $I_{A|B}^{(r+1)}$ denotes the ideal generated by the $(r+1)$-minors of the flattening matrix $M_{A|B}$.

\smallskip
In the bipartite case this inclusion is an equality (Proposition~\ref{prop:secant=det}), but in the multipartite case the flattening minors alone generally do not generate $I(\sigma_r(\Sigma))$. Higher degree equations are needed for an exact description.

Secant varieties form an increasing sequence:
\[
\sigma_1(\Sigma)=\Sigma\ \subset\ \sigma_2(\Sigma)\ \subset\ \cdots\ \subset\ \PP(H).
\]
This yields a filtration by border rank. Since ideals are contravariant, we obtain a descending chain of ideals:
\[
I(\Sigma)=I(\sigma_1(\Sigma))\ \supset\ I(\sigma_2(\Sigma))\ \supset\ \cdots\ \supset\ (0).
\]
This is the additive analogue of the picture in Proposition~\ref{prop:refinement-inclusion}: coarser partitions give larger varieties and smaller ideals.

Let $\dim \Sigma = \sum_{i=1}^N(d_i-1)$. In general,
\[
\dim \sigma_r(\Sigma)\ \le\ \min\bigl\{\,r(\dim \Sigma+1)-1,\ \dim \PP(H)\,\bigr\}.
\]
The right-hand side is called the \emph{expected dimension}. Intuitively, it is the sum of the degrees of freedom to choose $r$ points on $\Sigma$ (namely $r\dim\Sigma$) and the degrees of freedom to choose a point in the projective subspace they span (namely $r-1$). When the actual dimension is strictly smaller than the expected dimension, $\Sigma$ is called \emph{secant defective} (see, e.g., \cite{HarrisAlgebraicGeometry,landsberg2011tensors}).

\section{Bipartite Quantum Entanglement}
\label{sec:5}

\subsection{Decomposition of $\End(H_A\otimes H_B)$}
As a basic isomorphism for finite-dimensional spaces, we have
\[
\End(H_A\otimes H_B)\simeq \End(H_A)\otimes \End(H_B),
\]
which, in a choice of bases, is represented by the Kronecker product of matrices.

Using the vector-space direct sum decompositions
\[
\End(H_A)=\C\cdot I_A \oplus \End_0(H_A),\qquad
\End(H_B)=\C\cdot I_B \oplus \End_0(H_B)
\]
here $\End_0$ denotes the trace-zero part, we can expand the tensor product as
\begin{align*}
\End(H_A)\otimes\End(H_B)
&=
(\C I_A \oplus \End_0(H_A))\otimes (\C I_B \oplus \End_0(H_B))\\
&=
(\C I_A\otimes \C I_B)
\oplus (\End_0(H_A)\otimes \C I_B)
\oplus (\C I_A\otimes \End_0(H_B))
\oplus (\End_0(H_A)\otimes \End_0(H_B)).
\end{align*}
Here $\C I_A\otimes \C I_B=\C(I_A\otimes I_B)=\C I_{AB}$ is the scalar part of the whole algebra. Therefore the trace-zero part $\End_0(H_A\otimes H_B)$ decomposes as
\[
\End_0(H_A\otimes H_B)
\simeq
(\End_0(H_A)\otimes I_B)\ \oplus\ (I_A\otimes \End_0(H_B))\ \oplus\ (\End_0(H_A)\otimes \End_0(H_B)).
\]
Since the last summand is trace-zero on both sides, it is typically the component that entangles the states (the \emph{entangling sector}).

\subsection{Schmidt rank}
\begin{definition}[Schmidt rank]
For $\psi\in H_A\otimes H_B$, define the \emph{Schmidt rank} $\mathrm{SR}(\psi)$ by
\[
\mathrm{SR}(\psi):=\min\left\{r\ \middle|\ \psi=\sum_{\alpha=1}^r a_\alpha\otimes b_\alpha\ \text{for some }a_\alpha\in H_A,\ b_\alpha\in H_B\right\}.
\]
For a projective state $[\psi]$, this does not change under nonzero scalar rescaling of the representative, so one may also write $\mathrm{SR}([\psi])$.
\end{definition}

From here on, we take the inner product $\langle\cdot,\cdot\rangle$ to be linear in the second argument (and conjugate-linear in the first). For a complex Hilbert space $H$, the Riesz isomorphism from the conjugate space $\overline H$ to the dual $H^\vee$,
\[
J_H:\overline H \xrightarrow{\ \simeq\ } H^\vee,\qquad
J_H(\overline y)(x):=\langle y,x\rangle,
\]
is $\mathbb C$-linear.

In general, for finite-dimensional complex vector spaces $V,W$, there is a canonical isomorphism
\[
V\otimes W \;\simeq\; \Hom(W^\vee,V),\qquad
(v\otimes w)\mapsto(\ell\mapsto v\,\ell(w)).
\]
Composing this with $J_{H_B}$ yields
\[
H_A\otimes H_B \;\simeq\; \Hom(\overline{H}_B,H_A).
\]

Fix bases $\{e_i\}$ of $H_A$, $\{f_j\}$ of $H_B$, and the dual basis $\{f^j\}\subset H_B^\vee$.
For
\[
\psi=\sum_{i,j}\psi_{ij}\,e_i\otimes f_j,
\]
define
\[
\widetilde T_\psi:H_B^\vee\to H_A,\qquad
\widetilde T_\psi(\ell):=(\id\otimes \ell)(\psi).
\]
Then
\[
\widetilde T_\psi(f^j)=\sum_i\psi_{ij}e_i,
\]
so the matrix representation of $\widetilde T_\psi$ is exactly $M(\psi)=(\psi_{ij})$. Also, setting
$T_\psi:=\widetilde T_\psi\circ J_{H_B}$ gives a map $T_\psi:\overline{H}_B\to H_A$.

\begin{theorem}\label{thm:SR_is_rk}
If $\psi\neq 0$, then
\[
\mathrm{SR}(\psi)=\rk(M(\psi)).
\]
\end{theorem}

\begin{proof}
Let $\widetilde T_\psi:H_B^\vee\to H_A$ be as above. By the matrix representation, $\rk(\widetilde T_\psi)=\rk(M(\psi))$.

(Step 1) If $\psi=\sum_{\alpha=1}^r a_\alpha\otimes b_\alpha$, then for any $\ell\in H_B^\vee$,
\[
\widetilde T_\psi(\ell)=(\id\otimes \ell)(\psi)=\sum_{\alpha=1}^r a_\alpha\,\ell(b_\alpha).
\]
Hence $\im(\widetilde T_\psi)\subset \mathrm{span}\{a_1,\dots,a_r\}$, so $\rk(\widetilde T_\psi)\le r$. Therefore $\rk(M(\psi))\le r$, and thus $\mathrm{SR}(\psi)\ge \rk(M(\psi))$.

(Step 2) Conversely, assume $\rk(M(\psi))=\rk(\widetilde T_\psi)\le r$. Choose a basis $a_1,\dots,a_s$ of $\im(\widetilde T_\psi)$ with $s\le r$. Then there exist unique linear functionals $c_\alpha\in (H_B^\vee)^\vee$ such that
\[
\widetilde T_\psi(\ell)=\sum_{\alpha=1}^s a_\alpha\,c_\alpha(\ell)
\qquad(\forall \ell\in H_B^\vee).
\]
Since we are in finite dimensions, $(H_B^\vee)^\vee\simeq H_B$, so each $c_\alpha$ can be identified with some $b_\alpha\in H_B$, and we can write $c_\alpha(\ell)=\ell(b_\alpha)$. Hence
\[
\widetilde T_\psi(\ell)=\sum_{\alpha=1}^s a_\alpha\,\ell(b_\alpha)
=(\id\otimes \ell)\Bigl(\sum_{\alpha=1}^s a_\alpha\otimes b_\alpha\Bigr).
\]
Since this holds for all $\ell$, we have $\psi=\sum_{\alpha=1}^s a_\alpha\otimes b_\alpha$. Therefore $\mathrm{SR}(\psi)\le s\le r$, hence $\mathrm{SR}(\psi)\le \rk(M(\psi))$.

Combining Steps 1 and 2 proves the claim.
\end{proof}

\subsection{Schmidt decomposition (SVD)}
\begin{theorem}[Schmidt decomposition]
Let $\psi\in H_A\otimes H_B$ with $\|\psi\|=1$. Then there exist orthonormal systems $\{u_\alpha\}\subset H_A$, $\{v_\alpha\}\subset H_B$, and nonnegative real numbers $\sigma_\alpha\ge 0$ such that
\[
\psi=\sum_{\alpha=1}^{m}\sigma_\alpha\,u_\alpha\otimes v_\alpha,
\qquad
\sigma_\alpha>0\ (1\le \alpha\le m),\ \sigma_\alpha=0\ (\alpha>m).
\]
Here $m=\mathrm{SR}(\psi)=\rk(M(\psi))$, and
\[
\sum_{\alpha=1}^{m}\sigma_\alpha^2=1.
\]
\end{theorem}

\begin{proof}
(Step 1) Fix bases and set $M:=M(\psi)$. Then $MM^\dagger$ is Hermitian and positive semidefinite. By the spectral theorem, there exist a unitary matrix $U$ and eigenvalues $\lambda_\alpha\ge 0$ such that $U^\dagger (MM^\dagger)U=\mathrm{diag}(\lambda_1,\dots,\lambda_{d_A})$. Writing the column vectors of $U$ as $u_\alpha$, we have $(MM^\dagger)u_\alpha=\lambda_\alpha u_\alpha$.

(Step 2) Set $\sigma_\alpha:=\sqrt{\lambda_\alpha}\ge 0$. For $\sigma_\alpha>0$, define
\[
v_\alpha:=\frac{1}{\sigma_\alpha}M^\dagger u_\alpha.
\]
Then $\|v_\alpha\|=1$, and for those $\alpha$ with $\sigma_\alpha>0$, the vectors $\{v_\alpha\}$ are orthogonal.

(Step 3) Moreover, $Mv_\alpha=\sigma_\alpha u_\alpha$. Hence for any $x$,
\[
Mx=\sum_{\alpha:\sigma_\alpha>0}\sigma_\alpha u_\alpha v_\alpha^\dagger x,
\]
so
\[
M=\sum_{\alpha:\sigma_\alpha>0}\sigma_\alpha u_\alpha v_\alpha^\dagger.
\]

(Step 4) Translating the rank-one matrices $u_\alpha v_\alpha^\dagger$ back into tensors gives
\[
\psi=\sum_{\alpha:\sigma_\alpha>0}\sigma_\alpha\,u_\alpha\otimes v_\alpha.
\]
The number of $\sigma_\alpha>0$ equals $\rk(M)$, so $m=\rk(M)=\mathrm{SR}(\psi)$. By orthogonality,
\[
\|\psi\|^2=\sum_{\alpha=1}^m\sigma_\alpha^2.
\]
Since $\|\psi\|=1$, we obtain $\sum_{\alpha=1}^m\sigma_\alpha^2=1$.
\end{proof}

\subsection{Determinantal varieties}\label{sec:variety}
\begin{theorem}
For a $d_A\times d_B$ matrix $M$, the following are equivalent:
\begin{enumerate}
  \item $\rk(M)\le k$,
  \item all $(k+1)\times(k+1)$ minors of $M$ vanish.
\end{enumerate}
\end{theorem}

\begin{proof}
\textbf{(1)$\Rightarrow$(2).}
If $\rk(M)\le k$, then any $(k+1)$ columns are linearly dependent, hence any $(k+1)\times(k+1)$ submatrix has determinant $0$.

\smallskip
\textbf{(2)$\Rightarrow$(1).}
We argue by contradiction. If $\rk(M)\ge k+1$, then at least $(k+1)$ columns are linearly independent. Form the $d_A\times(k+1)$ matrix consisting of those columns. It has column rank $k+1$, hence row rank also $k+1$. Therefore one can choose $(k+1)$ rows so that the resulting $(k+1)\times(k+1)$ submatrix has nonzero determinant, contradicting (2). Thus $\rk(M)\le k$.
\end{proof}

Consider the subset of nonzero vectors in $H_A\otimes H_B\simeq\C^{d_A d_B}$ given by
\[
Z_{\le k}:=\{\,\psi\neq 0\mid \rk(M(\psi))\le k\,\}.
\]
Since this is defined by homogeneous polynomials (the minors), its projectivization
\[
R_{\le k}:=\PP(Z_{\le k})\subset \PP(H_A\otimes H_B)
\]
is a closed subset (an algebraic variety) in projective space. We call $R_{\le k}$ the
\emph{determinantal variety} \cite{BrunsVetterDeterminantal,WeymanSyzygies,HarrisAlgebraicGeometry}.

\subsection{Dimension formula}
The following is a classical fact about determinantal varieties.
\begin{proposition}
Let $1\le k\le \min(d_A,d_B)$. The projective determinantal variety $R_{\le k}\subset \PP^{d_Ad_B-1}$ has dimension
\[
\dim(R_{\le k})=k(d_A+d_B-k)-1,
\]
and codimension
\[
\mathrm{codim}(R_{\le k})=(d_A-k)(d_B-k).
\]
\end{proposition}

\begin{proof}
(Step 1)
A $d_A\times d_B$ matrix $M$ of rank exactly $k$ can be written as
\[
M = U V,
\]
where
\[
U\in\mathrm{Mat}_{d_A\times k},\qquad V\in\mathrm{Mat}_{k\times d_B},
\]
with $U$ of full column rank $k$ and $V$ of full row rank $k$. Indeed, $\mathrm{im}(M)\subset \C^{d_A}$ is a $k$-dimensional subspace. Taking a basis and arranging it as columns produces $U$. Each column of $M$ is a linear combination of this basis, and the coefficients form the matrix $V$, giving $M=UV$.

\smallskip
(Step 2)
For any $g\in\GL_k$ we have
\[
(Ug^{-1})(gV)=UV=M,
\]
so $(U,V)$ and $(Ug^{-1},gV)$ determine the same $M$. Conversely, in the full-rank situation, any two pairs $(U,V)$ giving the same $M$ are related by this equivalence: if $UV=U'V'$ and $U,U'$ have column rank $k$, then their column spaces both equal $\mathrm{im}(M)$, so there exists $g\in\GL_k$ with $U'=Ug^{-1}$. Then $UV=U'V'=Ug^{-1}V'$ implies $V=g^{-1}V'$.

\smallskip
(Step 3)
The dimensions are
\[
\dim(\mathrm{Mat}_{d_A\times k})=d_Ak,\qquad
\dim(\mathrm{Mat}_{k\times d_B})=kd_B.
\]
Quotienting by the above equivalence subtracts $\dim(\GL_k)=k^2$, hence
\[
\dim(\{\text{rank }k\text{ matrices}\})
=
d_Ak+kd_B-k^2
=
k(d_A+d_B-k).
\]
The rank $\le k$ locus has the same dimension because the rank-$k$ locus is dense in it (and its closure is the rank $\le k$ locus).

\smallskip
(Step 4)
Passing to projective space removes the $1$-dimensional freedom of scalar rescaling, so
\[
\dim(R_{\le k})=k(d_A+d_B-k)-1.
\]
Since the ambient space is $\PP^{d_Ad_B-1}$, the codimension is
\[
(d_Ad_B-1)-\bigl(k(d_A+d_B-k)-1\bigr)
=
d_Ad_B-kd_A-kd_B+k^2
=
(d_A-k)(d_B-k).
\]
\end{proof}

\section{An Algebro--Geometric Treatment of Quantum Subsystems}\label{sec:6}

\subsection{Local identifications and transition functions (gluing)}
From here on, we consider a \emph{family of quantum systems} over a parameter space $X$ (which may be a manifold or an algebraic variety).

Take an open cover $\{U_i\}$ of $X$, and assume that on each $U_i$ we have local identifications
\[
H_x \simeq \C^n\quad (x\in U_i).
\]
Comparing the identifications on two patches $U_i$ and $U_j$, on the overlap $U_{ij}:=U_i\cap U_j$ there appears a change of basis
\[
\C^n \xrightarrow{\ g_{ij}(x)\ } \C^n.
\]

Since pure states identify vectors up to nonzero scalars, the matrix $g_{ij}(x)$ is meaningful only up to scalar rescaling, hence it naturally defines an element of
\[
g_{ij}(x)\in \PGL_n(\C)=\GL_n(\C)/\C^\times.
\]

Assume that the gluing on $U_{ij}$ is given by
\[
([v],x)\longmapsto([g_{ij}(x)\cdot v],x),\qquad g_{ij}:U_{ij}\to \PGL_n.
\]
When the cocycle condition $g_{ij}g_{jk}g_{ki}=1$ holds on triple overlaps $U_{ijk}$ (in $\PGL_n$), this datum is equivalent to a principal $\PGL_n$-bundle (a $\PGL_n$-torsor) $P\to X$, and the family of pure-state spaces can be written as the associated bundle
\[
\text{family of pure states}\quad \simeq\quad P\times_{\PGL_n}\PP^{n-1}.
\]

\begin{definition}[Severi--Brauer scheme]
A projective family
\[
\pi:SB\to X,\qquad \pi^{-1}(x)\simeq \PP^{n-1}
\]
whose fibers are fppf-locally isomorphic to $\PP^{n-1}$ is called a \emph{Severi--Brauer scheme}.
\end{definition}

\subsection{Observable algebra (a sheaf of matrix algebras)}
In quantum mechanics, observables are elements of $\End(H)$ (matrices). Traditionally one often restricts to self-adjoint Hamiltonians, but many recent physical settings consider Hamiltonians that are not self-adjoint. So, here we work with the full matrix algebra.

If we have a local identification $H|_{U_i}\simeq \C^n\times U_i$, then the observable algebra over $U_i$
looks like
\[
\End(\C^n)\times U_i \simeq \mathrm{Mat}_n(\C)\times U_i.
\]
On overlaps, matrices are glued by conjugation via the change of basis:
\[
A_j(x)\ \longmapsto\ A_i(x)=g_{ij}(x)\,A_j(x)\,g_{ij}(x)^{-1}.
\]
In this way we obtain a sheaf of matrix algebras that is locally a matrix algebra, with gluing by conjugation. This provides a basic motivation for why we consider Azumaya algebras \cite{GilleSzamuelyCSAGC,DeMeyerIngrahamSeparable}.

\subsection{Local subsystem decompositions and the Segre variety}
Now fix a factorization
\[
n=d_1\cdots d_r,
\]
which we regard as a \emph{subsystem type}. Suppose that on each patch $U_i$ we refine the local identification further to a tensor product decomposition
\[
H_x \simeq \C^{d_1}\otimes\cdots\otimes \C^{d_r}
\quad (x\in U_i).
\]
Then, fiberwise, the product states form the Segre variety
\[
\Sigma_d \subset \PP(\C^{d_1}\otimes\cdots\otimes \C^{d_r})\simeq \PP^{n-1}.
\]

If, on an overlap $U_{ij}$, the transition function $g_{ij}(x)\in\PGL_n$ satisfies
\[
g_{ij}(x)\cdot \Sigma_d = \Sigma_d,
\]
then the locus of product states seen over $U_i$ agrees with the one seen over $U_j$ and these glue compatibly. In this case one obtains a global product-state subfibration over $X$.

Conversely, if the transition functions move the Segre variety to a different position, then the locally defined product loci do not glue, and one cannot define the set of product states globally.

\subsection{Stabilizer}
\begin{definition}[Stabilizer $G_d$]\label{def:stabilizer}
As a subgroup of the projective linear automorphism group $\PGL_n$ of $\PP^{n-1}$, define
\[
G_d := \{\,g\in \PGL_n \mid g(\Sigma_d)=\Sigma_d\,\}.
\]
We call $G_d$ the \emph{stabilizer} of the Segre variety $\Sigma_d$ \cite{Westwick1967,GesmundoHanLovitz2024}.
\end{definition}

In the bipartite case $d=(d_A,d_B)$, $G_d$ is typically understood as corresponding to local operations ($A\otimes B$),
\[
\PGL_{d_A}\times \PGL_{d_B}
\]
along with the factor-swap symmetry when $d_A=d_B$. This matches what is called \emph{local invertible operations} (SLOCC-type symmetry) in quantum information. Thus, $G_d$ is the group of allowed gauge transformations that do not destroy the globally compatible subsystem structure.

\begin{proposition}[Subsystem structure as a gluing condition]
Assume that the transition functions of the projective bundle of pure states are given, with respect to some open cover, by
\[
g_{ij}:U_{ij}\to \PGL_n.
\]
For a given choice of bipartite subsystem type $d=(d_A,d_B)$, a global product-state subfibration exists if and only if (after possibly passing to a refinement of the cover) one can take
\[
g_{ij}(x)\in G_d
\]
on each overlap \cite{Ikeda:2026ojm}.
\end{proposition}

\begin{proof}
(Necessity).
Assume there exists a global product-state subfibration $\Sigma_d(\text{family})\subset SB$. On each patch $U_i$, identify $SB|_{U_i}\simeq \PP^{n-1}\times U_i$. Then the subfibration agrees (at least up to the chosen identifications) with $\Sigma_d\times U_i$. On overlaps, the gluing map is the automorphism of $\PP^{n-1}$ induced by $g_{ij}(x)$, so we must have
\[
g_{ij}(x)\cdot (\Sigma_d\times\{x\}) = \Sigma_d\times\{x\}.
\]
Thus $g_{ij}(x)\in G_d$.

\smallskip
(Sufficiency).
Conversely, assume $g_{ij}(x)\in G_d$. On each $U_i$ take $\Sigma_d\times U_i$ as a subfamily. On an
overlap $U_{ij}$ we have
\[
([v],x)\in \Sigma_d\times U_{ij}
\mapsto
([g_{ij}(x)\cdot v],x)\in \Sigma_d\times U_{ij},
\]
and since $\Sigma_d$ is preserved, these local pieces glue compatibly. Hence a global subfibration exists.
\end{proof}

\section{\label{sec:Brauer_intro}Obstructions}
\subsection{Brauer class}
In this section, in order to gain intuition for \emph{why a global subsystem may fail to exist}, we compute the scalar discrepancy that appears when we lift the gluing data by one level.

\smallskip
Choose $\GL_n$-valued functions $\tilde g_{ij}$ locally representing the projective transition functions $g_{ij}\in\PGL_n$ (so that the image of $\tilde g_{ij}$ in $\PGL_n$ is $g_{ij}$). On a triple overlap, we have the cocycle condition in $\PGL_n$:
\[
g_{ij}g_{jk}g_{ki}=1.
\]
However, for the $\GL_n$-lifts one can in general have a purely scalar discrepancy:
\[
\tilde g_{ij}\tilde g_{jk}\tilde g_{ki} = c_{ijk}\,I_n
\qquad (c_{ijk}:U_{ijk}\to\C^\times).
\]
In other words, a $\PGL_n$-cocycle need not lift to a $\GL_n$-cocycle.

\smallskip
Let us directly verify that $\{c_{ijk}\}$ forms a \v{C}ech $2$-cocycle. On a quadruple overlap
$U_{ijkl}$, consider products such as
\[
(\tilde g_{ij}\tilde g_{jk}\tilde g_{ki})
(\tilde g_{ik}\tilde g_{kl}\tilde g_{li})
\]
and regroup them in two different ways. The resulting scalar factors must agree. Indeed, the expression above equals
\[
(c_{ijk}I)(c_{ikl}I)=c_{ijk}c_{ikl}I.
\]
On the other hand, the same product can also be rewritten as
\[
(\tilde g_{ij}\tilde g_{jl}\tilde g_{li})
(\tilde g_{jk}\tilde g_{kl}\tilde g_{lj}),
\]
which equals
\[
(c_{ijl}I)(c_{jkl}I)=c_{ijl}c_{jkl}I.
\]
Therefore we must have
\[
c_{ijk}c_{ikl}=c_{ijl}c_{jkl},
\]
which is precisely the \v{C}ech $2$-cocycle condition
\[
(\delta c)_{ijkl}=1.
\]

The cohomology class $[c]$ encodes the obstruction that says: a projective bundle exists, but it does not glue to a vector bundle. This obstruction is recorded as a Brauer class.

\subsection{A relation to subsystem structures}
Now fix a subsystem type $d=(d_1,\dots,d_r)$ and assume that the transition functions take values in $G_d$ (so the subsystem structure is globally meaningful). In this case, locally the transitions can be represented as tensor-product changes of basis:
\[
\tilde g_{ij}
\sim
A^{(1)}_{ij}\otimes\cdots\otimes A^{(r)}_{ij}
\qquad (A^{(m)}_{ij}\in \GL_{d_m}).
\]
Moreover, after adjusting each $A^{(m)}_{ij}$ by a scalar so that it lies in $\SL_{d_m}$, the remaining scalar ambiguity is only up to $d_m$-th roots of unity. As a consequence, the discrepancy appearing on triple overlaps satisfies the constraint
\[
\tilde g_{ij}\tilde g_{jk}\tilde g_{ki}\in
\mu_{d_1}\times\cdots\times \mu_{d_r},
\]
where $\mu_d$ denotes the group of complex numbers whose $d$-th power is $1$. Hence the discrepancy becomes trivial after raising to the power $\ell:=\mathrm{lcm}(d_1,\dots,d_r)$ (its $\ell$-th power is $1$). This is the intuitive reason why if a subsystem structure exists, then the obstruction is $\ell$-torsion. We summarize this as follows.

\begin{proposition}\label{prop:lcm-torsion}
Fix a type $d=(d_1,\dots,d_r)$ and let $\ell=\mathrm{lcm}(d_1,\dots,d_r)$. If an Azumaya algebra (or the corresponding Brauer class) admits a subsystem structure of type $d$, then the associated Brauer class is $\ell$-torsion:
\[
[A]\in Br(X)[\ell]:=\ \{\,\beta\in Br(X)\mid \ell\beta=0\,\}.
\]
More precisely, for a suitable lift $\widetilde G_d\subset \GL_n$ of $G_d$, there is a central extension
\[
1\to \mu_\ell\to \widetilde G_d\to G_d\to 1,
\]
and the obstruction $2$-cocycle takes values in $\mu_\ell$.
\end{proposition}

\section{Summary of the first half}
\begin{itemize}
  \item A pure state is represented as a point $[\psi]$ of the projective space $\PP(H)$ (overall scaling is ignored).

  \item In the bipartite case, the set of all product states is the Segre variety, and it can be detected by rank $1$ of the flattening matrix (equivalently, by the vanishing of all $2\times 2$ minors).

  \item The Schmidt rank coincides with the rank of the flattening matrix, and the condition rank $\le k$ is given by the vanishing of $(k+1)$ minors (a determinantal variety).

  \item Even for a twisted quantum system over a parameter space $X$, the family of pure-state spaces (a projective bundle) can be constructed globally by gluing.

  \item However, a general quantum system also contains classical states (product states), and to discuss whether a state is genuinely quantum (entangled), one must fix a \emph{subsystem decomposition} (a tensor product structure).

  \item A subsystem decomposition need not exist globally. Only when the transition functions lie in the stabilizer subgroup $G_d$ of the Segre variety (the structure group can be reduced to $G_d$) do product states and entanglement hierarchies make global sense over $X$.

  \item (Analogy with number theory: \textbf{multiplicative side})
  Under a fixed subsystem decomposition, the product locus appears as an irreducible closed subset (a Segre variety), and the biproduct locus is the union of these loci. Hence its vanishing ideal can be written as an intersection of prime ideals (a primary decomposition), analogous to $V(n)=\bigcup_{p\mid n}V(p)$ and $I(V(n))=\bigcap_{p\mid n}(p)$ in $\Spec\mathbb Z$ (a comparison at the radical level, forgetting exponents).

  \item (Analogy with number theory: \textbf{additive side})
  The minimal number of product states needed to express a state as a sum is the tensor rank, and the filtration given by its closure (border rank) is geometrized by the secant varieties $\sigma_r(\Sigma)$ of the Segre variety $\Sigma$. In the bipartite case $\sigma_r(\Sigma)$ agrees with the determinantal variety, but in the multipartite case flattening minors give only necessary conditions, and additional higher degree equations are needed for an exact classification.
\end{itemize}

\section{Quantum Entanglement Geometry on the Severi--Brauer Scheme}
\subsection{Overview of the second half}
The goal of the second half is to lift the notions introduced in the first half,
\[
\Sigma_d\subset \PP(H),\qquad
R_{\le k}\subset \PP(H)\qquad\text{(fixed Hilbert space }H)
\]
namely the product locus and the entanglement filtration, to the setting of a \emph{twisted family over a parameter space $X$}.

\medskip
\noindent
Here, twisting means the following phenomenon: although at each point $x\in X$ one indeed sees a pure-state space
\[
\PP(H_x)\simeq \PP^{n-1},
\]
globally over $X$ one cannot bundle the $H_x$ into a single vector bundle $H$. (Physical meaning is that locally the
system looks like an $n$-level system, but the global identifications are obstructed by gauge twisting.)

\medskip
\noindent
The most natural language to describe this twisting is that of Azumaya algebras. For further details see
\cite{Grothendieck1968BrauerI,Grothendieck1968BrauerII,Grothendieck1968BrauerIII} and \cite{Azumaya1951MaximallyCentral,AuslanderGoldman1960Brauer}. Using an Azumaya
algebra, one can \emph{always} construct the Severi--Brauer scheme
\[
\pi:SB(A)\to X
\]
as a \emph{family of pure-state spaces}. However, when one wants to discuss entanglement on $SB(A)$
(or even before that), there are several crucial points:
\begin{quote}
\begin{enumerate}
    \item To discuss entanglement one must choose a subsystem decomposition, and geometrically
    this appears as a \emph{reduction of structure group}. (A convenient reduction does not necessarily
    exist.)
    \item Only once such a reduction exists do the product locus and the entanglement
    filtration (determinantal varieties) \emph{descend} as \emph{global closed subschemes} inside $SB(A)$.
\end{enumerate}
\end{quote}

In this section we explain this claim from two viewpoints:
\begin{itemize}
  \item \textbf{(Via \v{C}ech gluing)}: the concrete condition that the transition functions preserve the Segre
  embedding.
  \item \textbf{(Via the Hilbert scheme)}: Segre-type subfamilies are modularized as a parameter space, and
  one shows this is precisely the reduction.
\end{itemize}

\medskip
\noindent
\textbf{Why is this detour necessary?}\;
In quantum information, entanglement is defined assuming a subsystem decomposition $H=H_1\otimes\cdots\otimes H_r$. But for a twisted family such as $SB(A)$, since the $H_x$ cannot be assembled globally into a single $H$, the very task of \emph{choosing a tensor decomposition globally} becomes nontrivial. The correct geometric formalization of this nontriviality is the reduction, and the existence/nonexistence of such a reduction can be measured as an obstruction in terms of the Brauer class. This is the overall picture of the paper.

\subsection{The Azumaya algebra $A$ and the pure-state space $SB(A)$}

Let $A$ be a sheaf of $\cO_X$-modules (or an $\cO_X$-algebra) on $X$, and let $f:U\to X$ be a morphism. We write the pullback of $A$ along $f$,
\[
f^*A,
\]
as
\[
A|_U \quad \text{or}\quad A_U.
\]
In particular, when $X=\Spec R$ and $U=\Spec S$ are affine and $A$ comes from an $R$-algebra
(or $R$-module), we regard
\[
A|_U \;\simeq\; A\otimes_R S
\]
as the operation of changing the coefficient ring from $R$ to $S$.

\begin{definition}[\'etale cover]
An \emph{\'etale cover} of $X$ is a family of morphisms $\{f_i:U_i\to X\}$ such that
\begin{enumerate}
  \item each $f_i$ is \'etale,
  \item the images cover $X$ ($\bigcup_i f_i(U_i)=X$).
\end{enumerate}
\end{definition}

As an analytic/differential-geometric analogy, an \'etale morphism is like a map that is \emph{locally a homeomorphism (locally a diffeomorphism)}, but globally may cover the target by multiple sheets. Thus localizing by an \'etale cover means allowing finite (or locally finite) unramified covers in order to unwind twisting that is invisible from Zariski open sets alone.

\begin{example}
For a finite product field extension $k\subset k'$, the morphism $\Spec k'\to \Spec k$ is finite \'etale.
\end{example}

\begin{definition}[fppf cover]
\emph{fppf} is an abbreviation of the French \emph{fid\`element plat et de pr\'esentation finie}, i.e.\ \emph{faithfully flat and locally of finite presentation}.

An \emph{fppf cover} of $X$ is a family of morphisms $\{f_i:U_i\to X\}$ such that
\begin{enumerate}
  \item each $f_i$ is faithfully flat,
  \item each $f_i$ is (locally) of finite presentation,
  \item the images cover $X$.
\end{enumerate}
\end{definition}

Flatness expresses the idea that relations among equations do not collapse or artificially appear under base change (extension of scalars). More algebraically, in the affine case $X=\Spec R$, $U=\Spec S$,
\[
R\to S \text{ is flat} \quad \Longleftrightarrow \quad
(-)\otimes_R S \text{ preserves exactness}.
\]
Faithfully means there is no loss of information, for instance
\[
M\neq 0 \Rightarrow M\otimes_R S\neq 0,
\]
so tensoring does not annihilate nonzero data.

\begin{remark}[Why consider fppf?]
When one needs to treat twisting or gluing phenomena that are not visible \'etale-locally, the fppf topology is a very natural framework. However, for Azumaya algebras, since they can essentially be described via a $\PGL_n$-torsor, the \'etale topology is often sufficient, and some references even define an Azumaya algebra by \'etale-locally a matrix algebra.
\end{remark}

\begin{definition}[Azumaya algebra]
Let $X$ be a scheme and let $n\ge 2$.
A sheaf of $\cO_X$-algebras $A$ is an \emph{Azumaya algebra of degree $n$} if there exists an \'etale
(or fppf) cover $f:U\to X$ such that
\[
A|_U \;\simeq\; M_n(\cO_U),
\]
the $n\times n$ matrix algebra over $\cO_U$.
\end{definition}

Choose isomorphisms over a cover,
\[
\phi_i:\ A|_{U_i}\xrightarrow{\ \sim\ } M_n(\cO_{U_i}).
\]
On overlaps (fiber products) $U_{ij}:=U_i\times_X U_j$, we have two identifications
\[
A|_{U_{ij}}
\ \xrightarrow{\ \phi_i|_{U_{ij}}\ }\ 
M_n(\cO_{U_{ij}}),
\qquad
A|_{U_{ij}}
\ \xrightarrow{\ \phi_j|_{U_{ij}}\ }\ 
M_n(\cO_{U_{ij}}),
\]
so their difference defines an automorphism
\[
\alpha_{ij}
:=
(\phi_i|_{U_{ij}})\circ(\phi_j|_{U_{ij}})^{-1}.
\]
Locally, there exists an invertible matrix $g_{ij}\in \GL_n(\cO_{U_{ij}})$ such that
\[
\alpha_{ij}(T) \;=\; g_{ij}\,T\,g_{ij}^{-1}
\qquad (T\in M_n(\cO_{U_{ij}})).
\]
Since $g_{ij}$ and $\lambda I_n\cdot g_{ij}$ induce the same conjugation, the gluing is essentially encoded by
\[
[g_{ij}] \in \PGL_n(\cO_{U_{ij}}).
\]
Abstracting the transition functions are $\PGL_n$-valued yields a principal $\PGL_n$-bundle (torsor) $P\to X$. The Azumaya algebra $A$ and the $\PGL_n$-torsor $P$ may be regarded as equivalent data.

Now, even if $A$ is twisted, the projective family of pure-state spaces always exists globally. This is the Severi--Brauer scheme
\[
\pi:SB(A)\to X.
\]
Locally it looks like
\[
SB(A)|_U\simeq \PP^{n-1}\times U.
\]
Thus the basic quantum-information principle ``pure states = projective space'' is guaranteed \emph{in families} by $SB(A)$. (Algebraically, $SB(A)$ has the universality of representing rank-$1$ right ideals.)

If $A$ is split ($A\simeq \End(E)$), then
\[
SB(A)\simeq \PP(E),
\]
recovering the usual projectivization of a vector bundle. In this sense, $SB(A)$ is a canonical base that always supplies the pure-state space.

\subsection{The stabilizer $G_d$ and reduction}
This is the main topic of the second half. Fix a factorization
\[
n=\prod_{i=1}^r d_i,\qquad d=(d_1,\dots,d_r),
\]
as a subsystem type.

In the usual (split) model,
\[
\PP^{n-1}\simeq \PP\bigl(k^{d_1}\otimes\cdots\otimes k^{d_r}\bigr),
\]
the image of the Segre embedding
\[
\Sigma_d
=\PP^{d_1-1}\times\cdots\times\PP^{d_r-1}
\hookrightarrow \PP^{n-1}
\]
is the locus of product states. As we discussed in the first half, entangled states are the points outside this product locus.

The key point is that $\Sigma_d$ is defined only after choosing a tensor decomposition. Although $SB(A)$ is locally isomorphic to $\PP^{n-1}$, these identifications are glued by $\PGL_n$-valued transition functions. If the transitions are arbitrary elements of $\PGL_n$, they need not send a Segre variety to a Segre variety. Consequently, the judgment this point is product/entangled risks becoming dependent on the choice of local coordinates.

Therefore, to define product states globally, one must require that the transitions preserve the Segre variety. This requirement is exactly a reduction.

\vspace{0.3cm}
We now use the stabilizer group $G_d$ defined as in Definition~\ref{def:stabilizer}.

\begin{definition}[Subsystem structure \cite{Ikeda:2026ojm}]
Let $P\to X$ be the $\PGL_n$-torsor associated to an Azumaya algebra $A$. A \emph{subsystem structure of type $d$} is the datum that
\[
P\ \text{admits a reduction to a }G_d\text{-torsor }P_d\to X,
\]
i.e.\ that one can write
\[
P \simeq P_d\times_{G_d} \PGL_n.
\]
\end{definition}

If we trivialize $P$ over an open cover $\{U_i\}$, then the transition functions are
\[
g_{ij}:U_{ij}\to \PGL_n.
\]
In these terms, being reducible is equivalent to the condition
\[
g_{ij}(x)\in G_d\quad(\forall x),
\]
which is exactly the same as the condition stated in the first half:
\[
g_{ij}(x)\cdot\Sigma_d=\Sigma_d.
\]

Given a $G_d$-torsor $P_d$, one can define
\[
\Sigma_d(A):=P_d\times_{G_d}\Sigma_d\to X,
\]
and the Segre embedding glues to a closed immersion
\[
\iota_d:\Sigma_d(A)\hookrightarrow SB(A).
\]
This is the \emph{global product locus}.

\medskip
\noindent
While $SB(A)$ is always defined as a family of pure states, deciding \emph{which states are product (classical) and which are entangled (genuinely quantum)} has no meaning without introducing subsystems. Moreover, it is not automatic that a globally consistent choice of subsystems exists. The correct geometric translation of the subsystem choice is the reduction.

\subsection{Hilbert schemes and modularizing subsystem structures}
Up to this point we have discussed constructions in a concrete form. Let us now treat them uniformly using the Hilbert scheme.

Having a subsystem structure means that locally there exists a Segre-type subfamily
\[
(SB(A)|_{U_i},\ \Sigma_d(A)|_{U_i})
\simeq
(\PP^{n-1}\times U_i,\ \Sigma_d\times U_i).
\]
Conversely, if inside $SB(A)$ one can find a family of closed subschemes that fiberwise looks like the Segre variety, then the overlap transitions are forced to preserve the Segre variety. As a result, the transition functions take values in $G_d$, and one recovers the reduction.

Since $\pi:SB(A)\to X$ is projective, the relative Hilbert scheme $\Hilb(SB(A)/X)$ exists and
parameterizes all closed subscheme families in $SB(A)$ that are flat over $X$.
Inside it, we single out those families of closed subschemes that satisfy:
\begin{center}
$fppf$-locally, they are isomorphic to $(\PP^{n-1}\times T,\ \Sigma_d\times T)$.
\end{center}
The locus defined by this condition is a \emph{locally closed subscheme} of $\Hilb(SB(A)/X)$, and we call it the subsystem-structure locus (inside the relative Hilbert scheme).

\begin{example}
In the split case, inside the Hilbert scheme of $\PP^{n-1}$, take the point
\[
[\Sigma_d]\in \Hilb(\PP^{n-1}).
\]
Its $\PGL_n$-orbit is
\[
O_{\Sigma_d}\simeq \PGL_n/G_d.
\]
This expresses the simple structure that $\PGL_n$ moves a Segre variety around, and $G_d$ is the freedom that fixes it.
\end{example}

\vskip0.3cm
\begin{definition}
Since $\pi:SB(A)\to X$ is projective, the relative Hilbert scheme $\Hilb(SB(A)/X)$ exists. The locally closed subscheme consisting of Segre-type closed subfamilies that are $fppf$-locally isomorphic to $(\PP^{n-1}\times T,\Sigma_d\times T)$ is called the subsystem-structure locus
\cite{Ikeda:2026ojm}.
\end{definition}

\begin{theorem}[{\cite{Ikeda:2026ojm}}]\label{thm:hilb-rep}
Fix an Azumaya algebra $A$ and a type $d$. The moduli of subsystem structures ($G_d$-reductions) is represented by
\[
P/G_d,
\]
and this is identified with the subsystem-structure locus inside the relative Hilbert scheme. In particular,
\[
\boxed{\ \text{a subsystem structure exists}\ \Longleftrightarrow\ \Sigma_d(A)\subset SB(A)\ \text{exists}\ }.
\]
\end{theorem}

\medskip
\noindent
Summarizing the discussion so far:
\begin{itemize}
  \item One can decide, in a unified way, whether a global subsystem structure exists by the existence of a scheme representing it.
  \item Once the moduli object ($P/G_d$) is available, subsequent descent statements can be obtained systematically by forming associated bundles.
\end{itemize}
Using the Hilbert scheme makes this picture more transparent.

\section{Bipartite Quantum Entanglement Geometry}
\subsection{Entanglement filtration}

From here on, we again consider the bipartite case $d=(d_A,d_B)$ (so $n=d_Ad_B$). We discuss explicitly the relationship between the most familiar notions in quantum information---the Schmidt decomposition / Schmidt rank---and the algebro--geometric entanglement filtration. A more detailed technical discussion is given in \S\ref{subsec:incidence-from-flattening}.

Let $H_A\simeq k^{d_A}$ and $H_B\simeq k^{d_B}$. Then a tensor $\psi\in H_A\otimes H_B$ can be identified with a matrix, and the condition that the Schmidt rank is $\le r$ is equivalent to
\[
\text{all }(r+1)\times(r+1)\text{ minors vanish.}
\]
The closed subscheme defined by this condition is
\[
R_{\le r}\subset \PP(H_A\otimes H_B).
\]
Here $R_{\le 1}$ is the product locus, and as $r$ increases the filtration becomes thicker. We obtain the natural filtration (by closed subvarieties) of the determinantal varieties
\[
R_{\le 1}\subset R_{\le 2}\subset\cdots\subset\PP(H_A\otimes H_B),
\]
which we call the \emph{entanglement filtration} \cite{Ikeda:2026ojm}.

\medskip
\noindent
\begin{remark}[Why do minors detect entanglement?]
Minors are polynomial conditions that detect whether the rank is $\le r$, which can be understood as the \emph{vanishing of polynomial equations}. This compatibility with polynomial/algebraic conditions is precisely what makes the viewpoint well-suited to algebraic geometry. Our discussion addresses the question of when entanglement criteria can be globalized over $X$?    
\end{remark}

When a subsystem structure $P_d\to X$ is given, $R_{\le r}$ is $G_d$--invariant, so by forming the associated bundle one can define
\[
\Sigma_{\le r}(A,d)
:=P_d\times_{G_d}R_{\le r}
\subset
P_d\times_{G_d}\PP(H_A\otimes H_B)\simeq SB(A).
\]
Concretely, this can be understood as follows: locally, after identifying $SB(A)\simeq \PP(H_A\otimes H_B)$, insert $R_{\le r}$ and glue. Because the subsystem structure exists, this gluing is consistent. Moreover, $\Sigma_{\le r}(A,d)\to X$ is fppf-locally isomorphic to $R_{\le r}\times U\to U$, hence it is flat. Flatness guarantees that the geometry of the fibers (dimension, etc.) does not jump, stabilizing the behavior as a family.

For each point $x\in X$, inside the fiber $(SB(A))_x\simeq \PP^{n-1}$ we obtain a filtration
\[
(\Sigma_{\le 1})_x\subset (\Sigma_{\le 2})_x\subset\cdots\subset (SB(A))_x,
\]
which is exactly the filtration Schmidt rank $\le r$. Therefore, for a pure state $[\psi]\in (SB(A))_x$ we can define a discrete entanglement invariant
\[
SR_d([\psi]) := \min\{r\mid [\psi]\in (\Sigma_{\le r}(A,d))_x\}.
\]
The condition $SR_d([\psi])=1$ means product, and larger values agree with the intuition of being more entangled.

\subsection{\label{sec:sing}Entanglement filtration and singularities}
The determinantal variety $R_{\le r}$ is singular in general, and its singularities obey a strong and regular pattern. When a reduction exists, this pattern is inherited uniformly in families. Concretely, for $1\le r<\min(d_A,d_B)$,
\[
\Sing(R_{\le r})=R_{\le r-1}
\]
holds \cite{WeymanSyzygies,BrunsVetterDeterminantal}.
Intuitively, the lower the rank, the more the Jacobian matrix of the defining equations degenerates, so the boundary of the filtration itself becomes singular.

Since $\Sigma_{\le r}(A,d)$ is locally $R_{\le r}\times U$, the above description of singular loci descends
directly to give
\[
\Sing\bigl(\Sigma_{\le r}(A,d)\bigr)=\Sigma_{\le r-1}(A,d).
\]
Thus the property the boundary of the entanglement filtration is singular holds uniformly over $X$.

A standard resolution of singularities of $R_{\le r}$ can be constructed using Grassmannians.
Using
\[
\Gr_A(r)=\Gr(r,H_A),\qquad \Gr_B(r)=\Gr(r,H_B)
\]
and their tautological subbundles $U_A, U_B$, define
\[
\widetilde R_{\le r}:=\PP(U_A\boxtimes U_B)\to \Gr_A(r)\times \Gr_B(r).
\]
Then there is a natural morphism (an incidence resolution)
\[
\rho_r:\widetilde R_{\le r}\to R_{\le r},
\]
which is projective and birational, and is an isomorphism over the open locus of matrices of rank exactly $r$. In this sense it gives a canonical resolution of $R_{\le r}$. A concrete construction is explained in \S\ref{subsec:incidence-from-flattening}.

This construction is $G_d$--equivariant. Hence, by forming the associated bundle, we obtain
\[
\widetilde\Sigma_{\le r}(A,d)
:=P_d\times_{G_d}\widetilde R_{\le r},
\qquad
\rho_r:\widetilde\Sigma_{\le r}(A,d)\to \Sigma_{\le r}(A,d).
\]
The total space $\widetilde\Sigma_{\le k}(A,d)$ is smooth over $X$, and $\rho_k$ is projective and birational and is an isomorphism over the open locus of rank exactly $k$ \cite{Ikeda:2026ojm}.

\begin{remark}[Practical meaning]
Once a subsystem structure is fixed, $\Sigma_{\le r}$ is not merely a subset, but behaves well in that:
\begin{itemize}
  \item it is flat (so the family has stable geometry),
  \item it is compatible with base change, and
  \item it comes equipped with a canonical resolution of singularities.
\end{itemize}
    
\end{remark}

\subsection{Concrete reinterpretations from the viewpoint of quantum information}

\paragraph{(QI-1) The algebro--geometric version of the Schmidt decomposition = SLOCC normal form.}
In quantum information, taking $k=\mathbb{C}$ and expanding a pure state $|\psi\rangle\in \mathcal{H}_A\otimes\mathcal{H}_B$ in a basis gives
\[
|\psi\rangle=\sum_{i=1}^{d_A}\sum_{j=1}^{d_B}\psi_{ij}\,|i\rangle_A\otimes|j\rangle_B,
\]
and we can identify it with the coefficient matrix $\Psi=(\psi_{ij})\in \mathrm{Mat}_{d_A\times d_B}$
and with the associated linear map $\psi_B:\mathcal{H}_B^\vee\to \mathcal{H}_A$).
Then
\[
\mathrm{SR}(|\psi\rangle)=\mathrm{rank}(\Psi)=\mathrm{rank}(\psi_B).
\]
Moreover (once an inner product is chosen), a Schmidt decomposition exists:
\[
|\psi\rangle=\sum_{\ell=1}^{r}\lambda_\ell\,|a_\ell\rangle\otimes|b_\ell\rangle,
\qquad \lambda_\ell>0,
\]
and the integer $r$ (the number of nonzero Schmidt coefficients) agrees with the above matrix rank.

On the other hand, the natural group appearing in this note is not the unitary group, but rather
\[
\mathrm{GL}(\mathcal{H}_A)\times\mathrm{GL}(\mathcal{H}_B),
\]
i.e.\ local invertible operations, which correspond in quantum information to \textbf{SLOCC} (the invertible part of stochastic LOCC \cite{Bennett:2000fte}). The coefficient matrix transforms as
\[
\Psi\ \longmapsto\ (g_A)\,\Psi\,(g_B)^{\mathsf{T}}
\quad (g_A\in \mathrm{GL}(d_A),\,g_B\in\mathrm{GL}(d_B)),
\]
and the matrix rank is invariant. Hence the filtration
$R_{\le1}\subset R_{\le 2}\subset\cdots\subset \PP(H_A\otimes H_B)$ determined by Schmidt rank is an entanglement filtration invariant under local invertible operations. The $G_d$--invariance / descent in the main text is exactly the geometric incarnation of this SLOCC invariance.

\paragraph{(QI-2) Algebraic sets and entanglement (the two-qubit example).}
Any two-qubit pure state can be written as
\[
|\psi\rangle = a|00\rangle+b|01\rangle+c|10\rangle+d|11\rangle,\qquad
\Psi=\begin{pmatrix}a&b\\ c&d\end{pmatrix}.
\]
In this case, being product (Schmidt rank $1$) is equivalent to
\[
\det(\Psi)=ad-bc=0.
\]
Here $\det(\Psi)$ is the unique $2\times 2$ minor, so a single polynomial suffices to test separability. In the language of quantum information, this is the topological core of the concurrence for pure states ($C(|\psi\rangle)=2|ad-bc|$): \textbf{being zero} is equivalent to being product.

For general $(d_A,d_B)$, similarly, if some $(r+1)\times(r+1)$ minor is nonzero then
\[
\mathrm{SR}(|\psi\rangle)\ge r+1
\]
follows algebraically (and conversely, if all such minors vanish then $\le r$). The fact that the presence or absence of entanglement can be written using polynomials (elements of a coordinate ring) is the root of the preceding discussion about whether the product locus can be globalized as a subscheme over a base scheme.

\paragraph{(QI-3) Quantum-information meaning of the entanglement filtration.}
In the bilinear case, ``tensor rank = matrix rank,'' so Schmidt rank $\le r$ is equivalent to a linear combination of $r$ product states
\[
|\psi\rangle=\sum_{\ell=1}^{r} |a_\ell\rangle\otimes|b_\ell\rangle.
\]
Geometrically one can view
\[
R_{\le r}=\overline{\mathrm{Sec}_r(\mathrm{Segre})}
\]
as the closure of the $r$-secant variety of the Segre \cite{landsberg2011tensors}. In this sense, one can also read the entanglement filtration as encoding how many product states are needed in a superposition.

\paragraph{(QI-4) Quantum-information meaning of families and reductions.}
From the quantum information viewpoint, one thinks of $x\in X$ as a classical parameter (external field, coupling constant, defect configuration, etc.), and the fiber $\mathbb{P}^{n-1}\simeq \mathbb{P}(\mathcal{H}_x)$ as the pure-state space at that parameter.

\begin{itemize}
\item For a twisted Azumaya background, even if locally $\mathcal{H}_x\simeq k^n$, one cannot
  \emph{globally fix a single $\mathcal{H}$} (the coordinate changes twist by $PGL_n$).
\item To discuss entanglement (as emphasized repeatedly), one needs additional data asserting that, for each $x$, $\mathcal{H}_x$ decomposes \emph{with the same meaning} as $\mathcal{H}_{A,x}\otimes\mathcal{H}_{B,x}$.
\end{itemize}

The $G_d$--reduction in the main text is exactly this additional datum: in the quantum-information picture it means that the coordinate changes (gauge transformations) are restricted from the full $\PGL_n$ to local operations $g_A\otimes g_B$. Only then does the statement
\[
\text{this state is product / has Schmidt rank $\le r$}
\]
acquire an \emph{$x$-independent, gluable meaning}, and $\Sigma_{\le r}(A,d)$ becomes an entanglement ``phase diagram'' defined over the whole parameter space $X$.

\paragraph{(QI-5) Intuitive understanding of singularities and quantum phase transitions.}
As already noted, from the quantum information viewpoint, a Schmidt-rank-$r$ state is SLOCC-equivalent to
\[
\sum_{i=1}^{r}|i\rangle\otimes|i\rangle
\]
(which expresses that rank $r$ is a complete SLOCC invariant). However, when the rank drops, the local degrees of freedom (orbit dimensions) change within the same $R_{\le r}$, and the geometry develops a sharp boundary.

For example, in the $3\times 3$ case with $r=2$,
\[
R_{\le 2}=\{\det(\Psi)=0\}\subset \mathbb{P}^8
\]
is a hypersurface defined by a single cubic equation. Since the gradient of $\det$ vanishes precisely when the rank is $\le 1$ (all $2\times 2$ cofactors vanish), one has
\[
\mathrm{Sing}(R_{\le 2})=R_{\le 1}.
\]
This is the geometric version of the familiar phenomenon that when one Schmidt coefficient collapses
to $0$ and the rank drops, the Jacobian degenerates as well. In the language of condensed-matter or
statistical physics, such a sudden change in dimension corresponds to a change in degeneracy of energy eigenstates in a quantum many-body system, and hence to a quantum phase transition.

\paragraph{(QI-6) Physical meaning of the incidence resolution.}
A rank-$r$ state $|\psi\rangle$ naturally determines $r$-dimensional subspaces on each side:
\[
U_A(\psi):=\mathrm{im}(\psi_B)\subset \mathcal{H}_A,
\qquad
U_B(\psi):=\mathrm{im}(\psi_A)\subset \mathcal{H}_B.
\]
(With an inner product, this agrees with the support $\mathrm{supp}(\rho_A)$ of the reduced density
matrix $\rho_A=\mathrm{Tr}_B|\psi\rangle\langle\psi|$.)

The incidence resolution
\[
\widetilde R_{\le r}=\mathbb{P}(U_A\boxtimes U_B)\to R_{\le r}
\]
can be understood intuitively as the enhancement
\[
\text{keep not only the state $[\psi]$, but also its support subspaces $U_A,U_B$.}
\]
\begin{itemize}
\item If the rank is exactly $r$, then $U_A(\psi)$ and $U_B(\psi)$ are uniquely determined, so the resolution is an isomorphism over that open locus.
\item If the rank is $<r$, then there are multiple choices of $r$-dimensional subspaces containing $\psi$, and the fibers thicken by exactly this non-uniqueness, resolving the singularities.
\end{itemize}

The same picture holds in families: $\widetilde\Sigma_{\le r}(A,d)$ is a smooth space encoding not only states but also how the support subspaces vary along the parameters. In quantum-information terms, this provides a uniform geometric framework for the entanglement structure and the geometry of its local effective degrees of freedom.

\paragraph{(QI-7) A remark on rank.}
Since the locus
\[
\{\,\mathrm{rank}(\Psi)\le r\,\}
\]
is closed, in a parameter family the rank is typically large at a general point and drops suddenly on a discriminant locus (where all relevant minors vanish). For instance, for
\[
|\psi(t)\rangle = |00\rangle + t\,|11\rangle
\]
one has rank $2$ for $t\neq 0$ and rank $1$ for $t=0$, and the discriminant is $\det(\Psi)=t$. The fact that the entanglement filtration globalizes as subschemes is directly tied to the fact that such ``entanglement phase boundaries'' can be written globally as polynomial conditions.

\section{Classification of Brauer-theoretic obstructions of entanglement}
The phenomenon from the first half---that on triple overlaps one may pick up a scalar discrepancy
(\S\ref{sec:Brauer_intro})---can be formulated in the language of the Brauer group. The Brauer class of an Azumaya algebra lies in $Br(X)$. Among these classes, we collect those for which there exists a representative admitting a subsystem structure of type $d$, and define the subset 
\[
Br_d(X)=\mathrm{Im}\Bigl(H^1(X,G_d)\to H^1(X,\PGL_n)\xrightarrow{\ \delta\ } Br(X)\Bigr)
\subset Br(X).
\]

\begin{definition}[{\cite{Ikeda:2026ojm}}]\label{def:Br_d}
A class $\beta\in Br(X)$ is said to be obstruction class for a type $d$ subsystem structure if
\[
\beta\notin Br_d(X).
\]
\end{definition}
This means that, if $\beta\notin Br_d(X)$, no matter which Azumaya algebra representing the Brauer class $\beta$ one chooses, it admits no subsystem structure of type $d$. Therefore, if $\beta\notin Br_d(X)$, then
\[
SB(A)\to X\ \text{always exists, but the product-state family}\ \Sigma_d(A)\subset SB(A)\ \text{does not exist.}
\]
If a globally consistent product state does not exist, then it is intrinsically entangled.

\begin{definition}[$\ell$-torsion subgroup]
For an abelian group $A$ and an integer $\ell\ge 1$, define
\[
A[\ell]\ :=\ \{\,a\in A\mid \ell a=0\,\}.
\]
In particular,
\[
Br(X)[\ell]\ :=\ \{\,\beta\in Br(X)\mid \ell\beta=0\,\}.
\]
Here, $\beta\in Br(X)[\ell]$ means that the period $\per(\beta)$ divides $\ell$.
\end{definition}

\begin{proposition}[{\cite{Ikeda:2026ojm}}]
Let $\ell=\mathrm{lcm}(d_1,\dots,d_r)$. Then for any $X$ one has
\[
Br_d(X)\subset Br(X)[\ell].
\]
In particular, if $\beta\in Br_d(X)$ then $\per(\beta)\mid \ell$.
\end{proposition}

\medskip
\noindent
\begin{remark}[why does it become $\ell$-torsion?]
When one attempts to lift $\PGL$-valued transition functions to $\GL$, a scalar discrepancy appears on triple overlaps. If a subsystem structure exists (the transitions lie in $G_d$), then this scalar discrepancy is constrained to be compatible with the Segre decomposition, and in the end it can only take values in $\mu_\ell$ (the group of $\ell$-th roots of unity). Consequently the Brauer class becomes $\ell$-torsion. The intuition from the first half is made precise here in cohomological terms.
\end{remark}

\paragraph{Related literature: entanglement as an obstruction on a fixed background}
\label{rem:fixed-vs-twisted-obstruction}
In this article, we represent the background of a quantum system by an Azumaya algebra $A$ on a scheme $X$, and describe the family of pure-state spaces as the Severi--Brauer scheme $SB(A)\to X$. On the other hand, in order to decide globally whether a quantum state is entangled, one must specify concretely where the correlations lie. This is a relative notion that should be specified separately from the underlying space or group structure. In other words, to define globally whether a state is entangled, one must \emph{globalize} the subsystem type (the type of tensor product decomposition given by $\mathbf d=(d_1,\dots,d_N)$). The central viewpoint of this paper is that this globalization of subsystem structure is equivalent to the possibility of reducing the structure group of the corresponding $PGL_n$-torsor to the stabilizer subgroup $G_d$ that stabilizes the Segre variety, and that the feasibility of such a reduction appears as a \emph{geometric obstruction} (including twisting of the background). To organize this obstruction, we introduce $Br_d(X)\subset Br(X)$ (Definition~\ref{def:Br_d}). If $[A]\notin Br_d(X)$, then a subsystem structure of the given type does not exist globally.

On the other hand, even when the background
Hilbert space and the tensor product structure are fixed globally, entanglement can arise in a different sense as an obstruction to gluing. For example, in \cite{Ikeda:2025mgj}, the phenomenon that local data (local states and marginals on subsystems) look mutually consistent while the global state is (nonexistent / non-unique) is formulated as the \v{C}ech cohomology of a presheaf of states and witnesses, and a framework is proposed in which entanglement is regarded as an
obstruction to global reconstruction. (See also \cite{Abramsky:2011sbx} for related work.)

Therefore, the target of the term ``obstruction'' in quantum systems should be distinguished as follows:
\begin{itemize}
  \item[\textbf{(I)}] \textbf{(Obstruction of structure)}
  This article: the problem of whether one can define the \emph{subsystem structure (the type of tensor
  product decomposition itself)} on $X$, i.e.\ whether a reduction of the structure group is possible.
  \item[\textbf{(II)}] \textbf{(Obstruction of states)}
  \cite{Ikeda:2025mgj,Abramsky:2011sbx}: after fixing the background and the subsystem structure, the
  problem of whether \emph{locally given state data can be reconstructed (glued) into a global state, and whether it is unique}.
\end{itemize}
Both share the same philosophy of ``from local to global,'' but it is important not to confuse them, since the level at which the obstruction appears (\textbf{structure} vs.\ \textbf{states}) is different.

\section{Examples where $SB(A)$ exists but no globally consistent subsystem exists}
\subsection{Symbol Azumaya algebra}
\begin{definition}[{\cite{Ikeda:2026ojm}}]\label{def:kummer}
Let $p$ be a prime and set $m=p^2$. Assume that $k$ contains a primitive $m$-th root of unity $\zeta_m$. On
\[
X=\Gm^2=\Spec k[u^{\pm1},v^{\pm1}],
\]
define a sheaf of $\cO_X$-algebras $A(u,v)_m$ by generators $x,y$ and relations
\[
x^m=u,\qquad y^m=v,\qquad yx=\zeta_m xy.
\]
We call $A(u,v)_m$ a \emph{symbol Azumaya algebra}.
\end{definition}

\begin{theorem}[{\cite{Ikeda:2026ojm}}]\label{thm:kummer}
Let $\beta=[A(u,v)_{p^2}]\in Br(X)$ be the Brauer class of $A(u,v)_{p^2}$. Then $\per(\beta)=p^2$. On the other hand, for the bipartite type $d=(p,p)$ one has $\ell=\mathrm{lcm}(p,p)=p$, hence
\[
Br_{(p,p)}(X)\subset Br(X)[p].
\]
Therefore $\beta\notin Br_{(p,p)}(X)$: although $SB(A)\to X$ exists, the $(p,p)$-subsystem structure
(and hence the entanglement filtration) cannot be globalized.
\end{theorem}

\subsection{A concrete quantum example}\label{sec:holonomic_gate}

Fix a prime $p$, and set $m=p^{2}$. As an external parameter space, consider
\[
\mathcal{X}=(\mathbb{C}^{\times})^{2},\qquad (u,v)\in\mathcal{X},
\]
and let
\[
\zeta := e^{2\pi i/m}
\]
be a primitive $m$-th root of unity. Below we give a local presentation of the symbol Azumaya algebra in the language of quantum information (a concrete Hilbert space and operators), and we explicitly show that the holonomy is not a local operation but becomes entangling.

\paragraph{Step 1: Local trivialization and an operator realization depending on $(u,v)$.}
Take a simply connected open subset $U\subset\mathcal{X}$, and choose (analytically) branches of the $m$-th roots $u^{1/m},v^{1/m}$ on $U$. Let
\[
H:=\mathbb{C}^{m}
\]
and fix the basis $\{|r\rangle\}_{r\in\mathbb{Z}/m\mathbb{Z}}$. Define Weyl-type shift/clock operators by
\begin{equation}
\label{eq:Weyl}
\mathbf{X}|r\rangle = |r+1\rangle,\qquad
\mathbf{Z}|r\rangle = \zeta^{\,r}|r\rangle,
\end{equation}
so that $\mathbf{Z}\mathbf{X}=\zeta\,\mathbf{X}\mathbf{Z}$ holds.
For $(u,v)\in U$, set
\[
x(u,v):=u^{1/m}\mathbf{X},\qquad y(u,v):=v^{1/m}\mathbf{Z}.
\]
Then we immediately obtain
\[
x(u,v)^{m}=u\,\mathrm{id}_{H},\qquad y(u,v)^{m}=v\,\mathrm{id}_{H},\qquad
y(u,v)\,x(u,v)=\zeta\,x(u,v)\,y(u,v).
\]
Thus, locally, the generators $x,y$ of the symbol algebra are realized as concrete quantum operators depending on $(u,v)$.

\paragraph{Step 2: Gauge holonomy arising from multivaluedness ($m$-th roots).}
Fix a base point $(u_{0},v_{0})\in U$ and consider the loops
\[
\gamma_{u}(t)=(u_{0}e^{2\pi it},\,v_{0}),\qquad
\gamma_{v}(t)=(u_{0},\,v_{0}e^{2\pi it})\qquad (0\le t\le 1).
\]
Analytically continuing the chosen branches gives
\[
u^{1/m}\longmapsto \zeta\,u^{1/m}\quad(\text{along }\gamma_{u}),\qquad
v^{1/m}\longmapsto \zeta\,v^{1/m}\quad(\text{along }\gamma_{v}),
\]
so $x(u,v)$ picks up the phase $\zeta$ along $\gamma_{u}$, and $y(u,v)$ picks up the phase $\zeta$ along $\gamma_{v}$. However, using the conjugations
\[
yxy^{-1}=\zeta x,\qquad x^{-1}yx=\zeta y,
\]
these phase changes can be absorbed as gauge transformations (gluing of bases). Since pure states identify an overall non-zero scalar, the holonomy can be taken as an element of the projective group $PGL(H)=GL(H)/\mathbb{C}^{\times}$:
\[
g_{u}=[\,y(u_{0},v_{0})\,]\in PGL(H),\qquad
g_{v}=[\,x(u_{0},v_{0})^{-1}\,]\in PGL(H)
\]
(where $[\cdot]$ denotes projectivization).

Moreover, computing the commutator in a $GL(H)$-lift leaves a scalar:
\[
g_{u}g_{v}g_{u}^{-1}g_{v}^{-1}
\ \sim\
(\mathbf{Z})(\mathbf{X}^{-1})(\mathbf{Z}^{-1})(\mathbf{X})
=\zeta^{-1}\mathrm{id}_{H}.
\]
That is, it looks commuting in $PGL$, but in $GL$ a phase (central scalar) $\zeta^{-1}$ remains. This is the local picture of the Brauer-theoretic twist (discrepancy).

\paragraph{Step 3: A local $(p,p)$ decomposition $H\simeq\mathbb{C}^{p}\otimes\mathbb{C}^{p}$.} Write $r\in\mathbb{Z}/m\mathbb{Z}$ uniquely as
\[
r=a+pb,\qquad a,b\in\{0,1,\dots,p-1\}.
\]
This gives a local identification
\[
H \ \cong\ H_{A}\otimes H_{B},\qquad H_{A}\cong\mathbb{C}^{p},\ H_{B}\cong\mathbb{C}^{p},
\]
by
\[
|a\rangle_{A}\otimes |b\rangle_{B}\ \leftrightarrow\ |a+pb\rangle.
\]
Under this identification, the holonomy $g_{v}$ along $\gamma_{v}$ acts (up to an overall non-zero scalar) as $\mathbf{X}^{-1}$. In $(a,b)$-notation,
\[
\mathbf{X}^{-1}|a,b\rangle =
\begin{cases}
|a-1,\ b\rangle & (a\neq 0),\\[4pt]
|p-1,\ b-1\rangle & (a=0),
\end{cases}
\]
so $b$ changes only when $a=0$ (a borrow occurs). Therefore, it is also intuitively clear that $\mathbf{X}^{-1}$ cannot be of the form $U_{A}\otimes U_{B}$ (a local operation).

\paragraph{Step 4: Checking that a product state is sent to an entangled state.}
Take the product state
\[
|\psi\rangle := (|0\rangle_{A}+|1\rangle_{A})\otimes |0\rangle_{B}.
\]
Applying $\mathbf{X}^{-1}$ yields
\[
\mathbf{X}^{-1}|\psi\rangle
=
|0\rangle_{A}\otimes |0\rangle_{B}
+
|p-1\rangle_{A}\otimes |p-1\rangle_{B}.
\]
This has Schmidt rank $2$ (the coefficient matrix has two independent components), hence it is entangled. Consequently, the holonomy $g_{v}$ of $\gamma_{v}$ is \emph{entangling} with respect to this local $(p,p)$ decomposition, and it does not lie in the stabilizer (local operations group) $PGL_{p}\times PGL_{p}$.

\paragraph{Supplement ($p=2$ case).}
When $p=2$ (so $m=4$),
\[
(|0\rangle+|1\rangle)\otimes |0\rangle \ \longmapsto\ |00\rangle+|11\rangle,
\]
which is a Bell state after normalization. In this sense, the gauge holonomy arising from the twist in $(u,v)$ very concretely exhibits an entangling component.

Since pure states are points of the projective space $\mathbb{P}(H)$ and an overall non-zero scalar is physically identified, any $g\in \mathrm{PGL}(H)$ acts on pure states (up to phase) by $[\psi]\mapsto g[\psi]$. Hence the projective holonomy $\rho(\gamma)\in \mathrm{PGL}(H)$ associated to a loop $\gamma$ can be interpreted as a \emph{holonomic quantum gate}. Under a local decomposition $H\simeq H_A\otimes H_B$, the local operations group is $\mathrm{PGL}(H_A)\times \mathrm{PGL}(H_B)\subset \mathrm{PGL}(H)$ (the stabilizer), and the condition $\rho(\gamma)\notin \mathrm{PGL}(H_A)\times \mathrm{PGL}(H_B)$ is equivalent to $\rho(\gamma)$ being an entangling gate. In particular, when $p=2$, the holonomy $g_v$ along $\gamma_v$ is locally equivalent to the $\mathrm{CNOT}$ class, and if one can prepare a universal set of one-qubit gates on each factor, then together with $g_v$ they generate a dense subgroup of $\mathrm{PU}(H_A\otimes H_B)\cong \mathrm{PU}(4)$ (universal quantum computation is possible).

\paragraph{Related literature: relation to holonomic quantum computation}
In this article, we interpreted the following phenomenon in terms of a ``holonomic quantum gate'': the monodromy (holonomy) arising from gluing charts on a parameter space $X$ appears as an element of the projective group $\PGL(H)$, and when it goes outside the local operation group, entanglement can be generated after the gluing.

This viewpoint is conceptually parallel to holonomic quantum computation, which is based on the geometric phase (Berry phase) \cite{Berry1984QuantalPhase} and the Wilczek--Zee non-Abelian geometric phase \cite{WilczekZee1984GaugeStructure}, and uses the holonomy of a degenerate subspace as a logical gate
\cite{ZanardiRasetti1999HolonomicQC,PachosZanardiRasetti2000NonAbelianBerry}. It is also compatible with the line of nonadiabatic holonomic QC, which constructs holonomic gates without relying on adiabaticity \cite{SjoqvistTongAnderssonEtAl2012NonadiabaticHQC}.

However, while standard holonomic QC obtains unitary transformations as the holonomy of the Berry connection on the vector bundle of degenerate eigenspaces over a control-parameter manifold, in our framework monodromy arises from inconsistencies in the data of a Severi--Brauer scheme and in the subsystem structure (reduction of the structure group). Thus, our examples provide an algebro geometric viewpoint that connects the philosophy of holonomic QC with a twisted background and with the failure of globalizing subsystem structure.

\section{Subsystem reducibility is not determined solely by the Brauer class}

Even if the Brauer class is $0$ (split), a subsystem structure does \emph{not} automatically exist. In other words, being untwisted does not mean a subsystem decomposition exists automatically.

One can construct examples on $X=\PP^1$ where
\[
A=\End(E),\qquad A'=\End(E')
\]
are both split (Brauer class $0$), yet $E$ admits a tensor decomposition and hence a subsystem structure exists, while $E'$ admits no tensor decomposition and hence no subsystem structure exists.

On $X=\PP^1$, by the Grothendieck splitting theorem one has
\[
E\simeq \bigoplus_{m=1}^n \mathcal O_{\PP^1}(a_m)
\quad (a_1\le\cdots\le a_n).
\]
The projective bundle $P(E)$ (equivalently, the $\PGL_n$-torsor) is determined by this splitting type, up to the overall twist $E\mapsto E\otimes \mathcal O(t)$.

\begin{theorem}[{\cite{Ikeda:2026ojm}}]\label{thm:splitting}
The following are equivalent:
\begin{enumerate}
\item $P$ is reducible to $G_d$ (it admits a bipartite subsystem structure).
\item There exist a rank-$d_A$ bundle $F$, a rank-$d_B$ bundle $G$, and a line bundle $L$ such that
\[
E\simeq (F\otimes G)\otimes L.
\]
\item There exist integer sequences $b_1\le\cdots\le b_{d_A}$, $c_1\le\cdots\le c_{d_B}$ and an integer
$t$ such that
\[
\{a_m\}_{m=1}^n=\{\,b_i+c_j+t\mid 1\le i\le d_A,\ 1\le j\le d_B\,\}
\]
as multisets.
\end{enumerate}
\end{theorem}

In particular, for $(2,2)$ this becomes the parallelogram condition
\[
a_1+a_4=a_2+a_3.
\]
This reflects the phenomenon that, over $\PP^1$, the geometry of the product locus (= matrix rank $1$) can be translated into arithmetic conditions on the splitting type.

\section{Numerical invariants}
Once a subsystem structure is fixed and $\Sigma_{\le r}(A,d)$ is defined, each fiber is always the same determinantal variety $R_{\le r}$. As a consequence, numerical invariants such as dimension, degree, and Hilbert polynomial become uniform, independent of the point $x\in X$.

The dimension and codimension are
\[
\dim(R_{\le r})=r(d_A+d_B-r)-1,\qquad
\codim(R_{\le r})=(d_A-r)(d_B-r).
\]
Hence each fiber $(\Sigma_{\le r})_x$ has the same dimension and codimension.

The degree and the Hilbert polynomial (with respect to the standard polarization) are also constant, and are determined purely by $(d_A,d_B;r)$. In particular, for $r=1$ one has
\[
\deg(\Sigma_{A,B})=\binom{d_A+d_B-2}{d_A-1},
\qquad
h^0(\Sigma_{A,B},\mathcal O(t))
=
\binom{t+d_A-1}{d_A-1}\binom{t+d_B-1}{d_B-1}.
\]
Thus the entanglement-induced stratification provides constant geometry throughout the family, and its numerical invariants can be tracked explicitly.

\begin{proposition}[\cite{Ikeda:2026ojm}]
Assume $d_A\le d_B$. The projective degree of the determinantal variety $R_{\le r}\subset \PP^{d_Ad_B-1}$ is
\[
\deg(R_{\le r})=
\prod_{i=0}^{d_A-r-1}\frac{(d_B+i)!\,i!}{(r+i)!\,(d_B-r+i)!}.
\]
\end{proposition}

\begin{proposition}[\cite{Ikeda:2026ojm}]
Writing $R_{\le r}=\Proj(S/I_{r+1})$ (where $I_{r+1}$ is the ideal generated by the $(r+1)$-minors), the degree-$t$ piece admits the representation-theoretic decomposition
\[
(S/I_{r+1})_t \simeq \bigoplus_{\lambda\vdash t,\ \len(\lambda)\le r}
\Schur_\lambda(k^{d_A})\otimes \Schur_\lambda(k^{d_B}),
\]
and hence the Hilbert polynomial is determined solely by $(d_A,d_B,r)$.
\end{proposition}

\begin{theorem}[{\cite{Ikeda:2026ojm}}]
Given a bipartite subsystem structure, each $\Sigma_{\le r}(A,d)\subset SB(A)$ is flat over $X$, and for every point $x\in X$ the fiber $(\Sigma_{\le r})_x$ is isomorphic to the classical $R_{\le r}$. Consequently, the dimension, degree, and Hilbert polynomial are constant in $x$ and depend only on $(d_A,d_B,r)$.
\end{theorem}

\section{Physics example: a quantum spin system on a torus (gluing of the ground state and Hamiltonian)}
\label{sec:spin-chain-toy}

In this section, using a 4-site spin system with periodic boundary conditions, we explain as concretely as possible the following phenomenon: although on each chart (locally) the ground state can be made to look like a product state, once we glue over the parameter space, the transition (gauge transformation) may fail to be a local operation, and as a result the ground state appears entangled.
The key points are as follows:
\begin{enumerate}
  \item On a simply connected open set (a chart), we can locally identify the Hilbert space as
  $\mathcal{H}\cong \mathbb{C}^{2}\otimes\mathbb{C}^{2}$, and with respect to this identification we can arrange that the local ground state is a product state.
  \item However, the gluing (transition functions) when moving to another chart may fail to lie in the \emph{local operation group}
  ($ \mathrm{PGL}_2\times \mathrm{PGL}_2$) inside $\mathrm{PGL}(\mathcal{H})$.
  \item Consequently, the locally chosen set of product states (the Segre variety) does not glue globally,
  and even for the same ground-state line (projective state), the state becomes entangled when viewed in a different chart.
\end{enumerate}

The total Hilbert space of the original spin chain is $(\mathbb{C}^2)^{\otimes 4}$, which is the physically natural tensor-product decomposition. In contrast, in this section we work with the \emph{one-magnon subspace} $\mathcal{H}\cong\mathbb{C}^4$, and the further decomposition $\mathcal{H}\cong\mathbb{C}^2\otimes\mathbb{C}^2$ is a decomposition of an \emph{effective two-qubit system constructed by a change of basis}. Hence, the entanglement discussed here is not with respect to the physical decomposition $(\mathbb{C}^2)^{\otimes 4}$, but rather is relative to the local identifications of $\mathcal{H}$ (a local subsystem structure). The important point is that this subsystem structure cannot be transported globally in a compatible way.

\paragraph{Related literature: comparison with band topology (Berry phase, Chern number).}
In condensed matter physics, ``band topology'' studies how the eigenspaces of occupied bands over the Brillouin zone $T^d$
form a complex vector bundle over $T^d$ (the Bloch bundle), and how the gluing of local frames (gauge transformations) leads to the Berry connection and Berry curvature.
Universal quantities such as Chern numbers and $K$-theoretic classifications measure
\emph{whether a certain subbundle (the occupied bands) can be globally trivialized}
(e.g.\ \cite{Berry1984QuantalPhase,TKNN1982,Zak1989,HasanKane2010_RMP,QiZhang2011_RMP}).

By contrast, the \emph{global geometric entanglement} studied in this article measures (at least conceptually) a different level of structure:
\begin{itemize}
\item Band topology mainly concerns how an eigenstate line (or an occupied-band subspace) twists and glues globally.
The gluing group is typically a $U(r)$ gauge transformation (where $r$ is the rank of the occupied band),
and the obstruction appears as Chern numbers, etc.
\item In this article, we also allow the background (the quantum system itself) to globalize only \emph{projectively} via gluing of local identifications, and we furthermore ask whether a subsystem decomposition (tensor product structure) can be made globally compatible (whether a reduction of structure group is possible).
If this reduction fails, then the transitions (monodromy) between charts leave the local operation group, and the phenomenon can occur that \emph{a state that looks like a product state locally can produce entanglement after gluing}.
\end{itemize}
\[
\boxed{
\begin{minipage}{0.99\linewidth}
Therefore the obstruction classes discussed here are not Berry or Chern numbers. They are
invariants of a different kind. Berry/Chern data comes from globalizing eigenstates (or eigenbundles), whereas our classes come from globalizing the subsystem structure itself. When the monodromy of coordinate changes leaves the local operation group, it can generate entanglement. In that situation the associated obstruction classes are, at least in principle, accessible independently of Berry/Chern invariants (see \cite{Ikeda:2026ojm}).
\end{minipage}
}
\]

\subsection{Parameter space and the one-magnon Hilbert space}
As in \S\ref{sec:holonomic_gate}, we take the parameter space to be
\[
X:=(\mathbb{C}^{\times})^{2},\qquad (u,v)\in X.
\]
When considering a family of physically Hermitian (observable) Hamiltonians, we restrict to
\[
X_{\mathrm{phys}}:=(S^{1})^{2}\subset (\mathbb{C}^{\times})^{2},\qquad
u=e^{i\theta_u},\ v=e^{i\theta_v}.
\]
Below we fix $p=2$ and set
\[
m=p^{2}=4,\qquad \zeta:=e^{2\pi i/m}=i.
\]

\medskip
Consider a spin chain with periodic boundary conditions whose sites $r$ are spin-$1/2$ particles, with $r\in\mathbb{Z}/4\mathbb{Z}$. For simplicity, assume the chain consists of four particles. Let $\sigma_r^\pm$ denote the raising/lowering operators at site $r$, and define
\[
n_r:=\frac{1-\sigma_r^{z}}{2}
\]
to be the number operator for down spins (magnons). Taking the fully polarized state $|\!\uparrow\uparrow\uparrow\uparrow\rangle$ as a reference, define the one-magnon basis
\[
|r\rangle:=\sigma_r^-|\!\uparrow\uparrow\uparrow\uparrow\rangle,
\qquad r\in\mathbb{Z}/4\mathbb{Z}.
\]
This is the state where only site $r$ is down and the others are up. The one-magnon subspace
\[
\mathcal{H}:=\mathrm{span}\{|0\rangle,|1\rangle,|2\rangle,|3\rangle\}\cong \mathbb{C}^{4}
\]
is invariant under any Hamiltonian that preserves the total magnon number (the magnon moves while remaining a single magnon). Intuitively, $\mathcal{H}$ is a 4-dimensional effective space whose degrees of freedom are just the magnon position.

\medskip
On $\mathcal{H}$, define Weyl-type operators as in \eqref{eq:localH-U} by
\[
\mathbf{X}|r\rangle=|r+1\rangle,\qquad
\mathbf{Z}|r\rangle=\zeta^{\,r}|r\rangle\quad(\zeta=i),
\]
so that
\[
\mathbf{Z}\mathbf{X}=\zeta\,\mathbf{X}\mathbf{Z}.
\]
Physically, $\mathbf{X}$ can be viewed as the lattice translation (a one-site shift) restricted to the one-magnon subspace, and $\mathbf{Z}$ as a site-dependent phase (a phase gate implementable, for example, by local fields).

\subsection{Local Hamiltonian and local ground state}
Take a simply connected open set (chart) $U\subset X^{\mathrm{an}}$, and on it choose an analytic branch of $u^{1/4}$ (on a simply connected set one can choose a branch continuously). See Appendix \ref{sec:torus} for details. Fix constants $\Delta>J>0$, and define the following local Hamiltonian on the full Hilbert space of the spin chain:
\begin{equation}
\label{eq:localH-U}
H_U(u,v)
:=
-J\Big(u^{-1/4}\,\sigma_0^+\sigma_1^- + u^{1/4}\,\sigma_0^-\sigma_1^+\Big)
\;+\;\Delta\,(n_2+n_3).
\end{equation}

\medskip
\noindent
In words:
\begin{itemize}
  \item The first term is a hopping term moving the magnon between sites $0$ and $1$, with complex phases $u^{\pm1/4}$ in the hopping amplitude—i.e.\ a Peierls phase. For the Hamiltonian to be physically Hermitian, the coefficients must be complex conjugates of each other. If $|u|=1$ then $u^{-1/4}=\overline{u^{1/4}}$, so it is indeed Hermitian.
  \item The second term is a penalty term that assigns energy $\Delta$ if the magnon is at site $2$ or $3$. Hence, if $\Delta$ is large, low-energy states are pushed primarily into the subspace spanned by $\{|0\rangle,|1\rangle\}$.
\end{itemize}

\medskip
Since $H_U$ preserves the magnon number, it restricts to $\mathcal{H}$. In the basis $\{|0\rangle,|1\rangle,|2\rangle,|3\rangle\}$, the restriction is
\begin{equation}
\label{eq:matrixHU-JA}
H_U\big|_{\mathcal{H}}
=
\begin{pmatrix}
0 & -J u^{1/4} & 0 & 0\\
-J u^{-1/4} & 0 & 0 & 0\\
0 & 0 & \Delta & 0\\
0 & 0 & 0 & \Delta
\end{pmatrix}.
\end{equation}
The bottom-right $\Delta,\Delta$ means energy $\Delta$ if the magnon is at sites $2,3$. The top-left $2\times2$ block describes hopping between $|0\rangle$ and $|1\rangle$.

\begin{proposition}[Local ground state]
\label{prop:localGS-JA}
Assume $\Delta>J>0$ and $|u|=1$. Then the eigenvalues of $H_U|_{\mathcal{H}}$ are $-J,+J,\Delta,\Delta$, and the unique ground state (energy $-J$) is
\begin{equation}
\label{eq:GSU-JA}
|\mathrm{GS}\rangle_U
=
\frac{1}{\sqrt{2}}\Big(u^{1/4}|0\rangle + |1\rangle\Big).
\end{equation}
Moreover, the energy of the fully polarized state $|\!\uparrow\uparrow\uparrow\uparrow\rangle$ is $0$, and in higher-magnon sectors the energy is $\ge 0$ or $\ge \Delta$, so \eqref{eq:GSU-JA} with $-J<0$ is also the ground state of the full Hilbert space.
\end{proposition}

\noindent
\emph{Important point:} the state \eqref{eq:GSU-JA} depends on $u^{1/4}$, but $u^{1/4}$ is globally multi-valued. Thus, as written, it is difficult to define it as a single vector $|\mathrm{GS}(u,v)\rangle$ varying continuously over all of $X$. How to handle this is exactly the gluing discussion that follows. See Appendix \ref{sec:torus} for more details.

\subsection{The idea of gluing: how to globalize the Hamiltonian}
In general, if one changes the choice of basis in a quantum system (applies a unitary transformation), the Hamiltonian changes by conjugation:
\[
H\ \longmapsto\ U^{-1}HU.
\]
However, the eigenvalues (energy spectrum) and the projective eigenspaces do not change, so Hamiltonians representing the same physics are determined only up to unitary conjugacy. We use this identification up to conjugacy as a \emph{gluing rule} over the parameter space $X$.

\medskip
In this example, it is most natural to use lattice translations for gluing. Define the translation operator $T$ on the full spin-chain Hilbert space by
\[
T\sigma_r^\pm T^{-1}=\sigma_{r+1}^\pm\qquad(\text{indices mod }4).
\]
Restricting $T$ to the one-magnon subspace gives
\[
T|_{\mathcal{H}}=\mathbf{X},
\]
agreeing with the shift in \eqref{eq:Weyl}.

\medskip
On another simply connected chart $U'\subset X^{\mathrm{an}}$, define the Hamiltonian by conjugation:
\begin{equation}
\label{eq:localH-Up-JA}
H_{U'}(u,v):=T^{-1}H_U(u,v)\,T.
\end{equation}
Writing $H_{U'}$ in terms of spin operators yields
\begin{equation}
\label{eq:localH-Up-explicit-JA}
H_{U'}(u,v)
=
-J\Big(u^{-1/4}\,\sigma_3^+\sigma_0^- + u^{1/4}\,\sigma_3^-\sigma_0^+\Big)
\;+\;\Delta\,(n_1+n_2),
\end{equation}
which is still a \emph{local interaction} (only nearest-neighbor interactions). Thus $H_U$ couples the bond $(0,1)$, while $H_{U'}$ couples the bond $(3,0)$, shifted by one site.

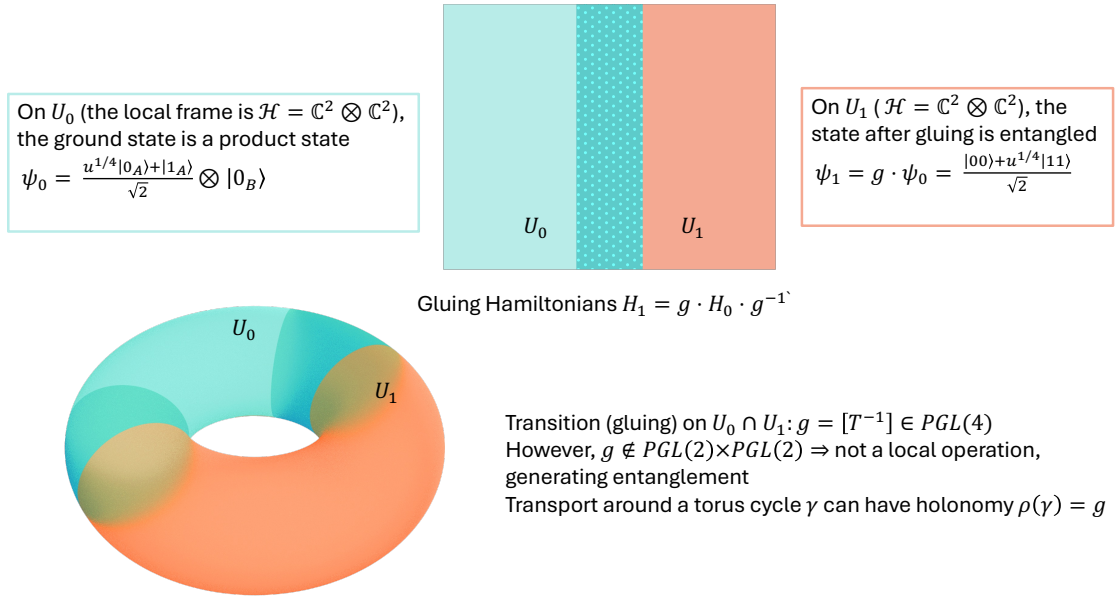
\begin{figure}[H]
    \centering
    \includegraphics[width=0.9\linewidth]{monodromy.pdf}
    \label{fig:placeholder}
\end{figure}

\medskip
\noindent
\begin{remark}
On the overlap $U\cap U'$, \eqref{eq:localH-Up-JA} shows that $H_U$ and $H_{U'}$ are related by unitary conjugation. Thus, rather than writing a single global Hamiltonian over $X$, we define a \emph{globally twisted} family of Hamiltonians by
\[
\text{taking local expressions }\{H_U,\ H_{U'},\dots\}\ \text{and gluing them on overlaps by conjugacy.}
\]
Mathematically, this corresponds to constructing a family that is determined only up to fiberwise conjugation (a family glued by gauge transformations in the projective group).
\end{remark}

\subsection{Gluing the ground-state line: transitions of projective states}
Next we explain how the ground state glues. Since $H_{U'}=T^{-1}H_UT$, if $|\psi\rangle$ is an eigenvector of $H_U$, then $T^{-1}|\psi\rangle$ is an eigenvector of $H_{U'}$ with the same eigenvalue. In particular, for the local ground state \eqref{eq:GSU-JA}, define
\[
|\mathrm{GS}\rangle_{U'}:=T^{-1}|\mathrm{GS}\rangle_U.
\]
On the one-magnon subspace, $T^{-1}=\mathbf{X}^{-1}$, so in the projective space $\mathbb{P}(\mathcal{H})$ we obtain the gluing rule
\[
[\mathrm{GS}]_{U'}=[\,\mathbf{X}^{-1}\mathrm{GS}\,]_{U}
\]
(where $[\cdot]$ denotes the state modulo overall non-zero scalar).

\medskip
A direct computation gives
\begin{equation}
\label{eq:GSUprime-JA}
|\mathrm{GS}\rangle_{U'}
=
\mathbf{X}^{-1}|\mathrm{GS}\rangle_{U}
=
\frac{1}{\sqrt{2}}\Big(u^{1/4}|3\rangle + |0\rangle\Big).
\end{equation}
This gluing is naturally formulated not for state vectors but for elements of projective space, so the gluing group is  $\mathrm{PGL}(\mathcal{H})$.

\subsection{A local two-qubit decomposition and the appearance of entanglement via gluing}
Now that the setup is in place, we check locally product, but becomes entangled after gluing. To view $\mathcal{H}\cong\mathbb{C}^4$ as two qubits, encode the basis $|r\rangle$ by a 2-bit string $(a,b)$:
\[
r=a+2b,\qquad a,b\in\{0,1\},
\]
and identify
\[
|a\rangle_A\otimes |b\rangle_B\ \longleftrightarrow\ |r=a+2b\rangle.
\]
Then
\[
|0\rangle=|00\rangle,\quad |1\rangle=|10\rangle,\quad
|2\rangle=|01\rangle,\quad |3\rangle=|11\rangle.
\]
(Again, this is a \emph{local identification inside the one-magnon space}, not the physical decomposition of the four-spin system itself.)

\medskip
Under this identification, \eqref{eq:GSU-JA} becomes
\begin{align}
|\mathrm{GS}\rangle_{U}
&=
\frac{1}{\sqrt{2}}\Big(u^{1/4}|00\rangle+|10\rangle\Big)
=
\Big(\tfrac{u^{1/4}|0\rangle_A+|1\rangle_A}{\sqrt{2}}\Big)\otimes |0\rangle_B.
\label{eq:GSU-product-JA}
\end{align}
Hence $|\mathrm{GS}\rangle_U$ is a \emph{product state} (in the sense of subsystems $A$ and $B$).

\medskip
On the other hand, the glued ground state \eqref{eq:GSUprime-JA} is
\[
|\mathrm{GS}\rangle_{U'}
=
\frac{1}{\sqrt{2}}\Big(|00\rangle+u^{1/4}|11\rangle\Big).
\]
This is immediate for readers familiar with quantum information, but one can see that this state is entangled as follows:
\begin{itemize}
\item In general, a two-qubit state
\(
|\phi\rangle=\sum_{a,b=0}^{1}c_{ab}|a\rangle_A\otimes|b\rangle_B
\)
corresponds to the coefficient matrix $C=(c_{ab})$, and
\emph{being a product state is equivalent to $C$ having rank $1$} (Proposition \ref{prop:matrixrank1}).
\item Here the coefficient matrix is
\(
C=\mathrm{diag}(1,u^{1/4}),
\)
so as long as $u^{1/4}\neq 0$, the rank is $2$. Hence the Schmidt rank is $2$, thus it is an \emph{entangled state}.
\end{itemize}
In particular, if $u=1$ then
\[
|\mathrm{GS}\rangle_{U'}=\frac{1}{\sqrt{2}}\big(|00\rangle+|11\rangle\big),
\]
which is exactly the normalized Bell state.

\subsection{Why entanglement appears after gluing: the transition element is not a local operation}
Finally, we summarize why this happens from the viewpoint of gauge transformations.

\medskip
The gluing on the overlap $U\cap U'$ was
\[
[\mathrm{GS}]_{U'}=[\,\mathbf{X}^{-1}\mathrm{GS}\,]_U.
\]
That is, the transition element is
\[
g_v=[\mathbf{X}^{-1}]\in \mathrm{PGL}(\mathcal{H}).
\]
If $g_v$ lay in the local operation group $\mathrm{PGL}(\mathbb{C}^2)\times\mathrm{PGL}(\mathbb{C}^2)$, then the locally chosen set of product states (the Segre variety) would be preserved on overlaps, and one should be able to glue a global family of product states consistently.

\medskip
However, $\mathbf{X}^{-1}$ is not a local operation with respect to this local decomposition.
Indeed, for the basis written as $r=a+2b$, it acts on a state (up to sign) as 
\[
\mathbf{X}^{-1}|a,b\rangle=
\begin{cases}
|a-1,\ b\rangle & (a=1),\\[2pt]
|1,\ b-1\rangle & (a=0),
\end{cases}
\]
so the change in $b$ occurs depending on $a$. This is an operation of the type $B$ changes depending on the state of $A$,
which is precisely the \emph{nonlocality} typical of controlled gates (CNOT-type gates). This nonlocal transition element is the mechanism that carries a locally product ground state to an entangled form when viewed in another chart.

\medskip
In summary, this example can be understood as follows:
\begin{quote}
\emph{Locally, one can introduce a two-qubit viewpoint ($\mathcal{H}\cong\mathbb{C}^2\otimes\mathbb{C}^2$) and make the ground state a product state. However, because the transition element defining the gluing between charts does not lie in the local operation group, this viewpoint cannot be extended globally. This failure of extension is observed as (geometric) entanglement arising from gluing.}
\end{quote}
This is an example that translates the viewpoint parallel transport along loops in the parameter space acts as a quantum gate (a holonomic gate), and it can be entangling (a physical translation of a geometric obstruction) into a simple model that is more directly implementable and has a more universal character. 

\medskip
\noindent
\[
\boxed{
\begin{minipage}{0.99\linewidth}
\textbf{Physical takeaway.}
For spin systems, it is reasonable to expect that many models with entangled eigenstates can also exhibit gluing-induced entanglement phenomena in parameter families. One reason such effects are rarely discussed is that many treatments work in a single coordinate patch, even when the parameter manifold requires multiple charts. For example, neither the torus nor the sphere can be covered by one chart, yet in many settings (such as topological insulators on the Brillouin zone) one has treated them on a single coordinate patch. Of course those are quantum many-body systems, so eigenstates (especially ground states) are already entangled on a local chart. But applying the viewpoint of this work naturally reveals that \textbf{new (global geometric) entanglement associated with coordinate changes exists}.
We formulate this systematically and illustrate it in basic examples. Moreover, the universal quantities (obstruction classes) accompanying such global geometric entanglement--quantities, to the best of our knowledge, that have not been emphasized in existing physics--can be newly observed (separately from the Berry numbers and Chern numbers of energy bands).
\end{minipage}
}
\]

\section{Spherical Hecke Spectra and an Algebraic Criterion for Entanglement}\label{sec:16}
So far, we have discussed a twisted family of quantum states. Here, let us consider measurement of entanglement as spectrum of some operators. To do so, we consider the map that sends the system to its dual space. The reasons are explained as follows.  

\subsection{Guiding philosophy}
The aim of this section is to set up a language for describing the following phenomenon on the geometric side,
\[
\text{``Locally one can see a tensor decomposition, but globally it does not glue.''}
\]
This \emph{local-to-global failure} (entangling monodromy) is to be encoded, on the dual side, as \emph{simultaneous eigenvalue data of commuting operators}.

The monodromy on the geometric side (a $\PGL_n$-valued holonomy) and the Satake parameter on the dual side (a semisimple conjugacy class in $G^\vee$) are \emph{not} obtained by identifying \emph{the same element of the same group}.

Nevertheless, the reason why the correspondence is ``visible'' is that the questions on both sides share essentially the same logical structure:
\[
\textbf{``Can one reduce (factorize) to a tensor type?''}
\]
More precisely:
\begin{itemize}
\item Geometric side: if the structure group cannot be reduced to $G_{\mathbf d}$, then transition functions leave $G_{\mathbf d}$,
  and holonomy appears that sends local product data to entangled data.
\item Satake side: if the Satake parameter does not lie in the image of $r_{\mathbf d}$, then
  the image conditions (polynomial relations) imposed on invariants fail.
\end{itemize}
Here we are comparing \emph{factorizability} on both sides.

\subsection{Terminology}

\subsubsection*{On ``spherical''}
The word ``spherical'' appears in two different senses in this section, so we first distinguish them:
\begin{itemize}
\item \textbf{(Sph-Hecke)} ``spherical'' in the spherical Hecke algebra, i.e.\ the unramified situation: take $G$ to be an unramified reductive group (in particular, it suffices that $G$ is split), fix a hyperspecial subgroup \(K\), and consider \(\mathcal H(G,K)\). This holds in general for (split) reductive $G$ and is unrelated to the sphericity of a homogeneous space $X=G/ H$.
\item \textbf{(Sph-var)} ``spherical'' in the sense of spherical varieties: for $X=G/ H$, a Borel subgroup has an open orbit. This is needed to interpret spectral conditions analytically, but in this section it appears only in \S\ref{sec:Sph-var} and is not used elsewhere.
\end{itemize}

From now on, the polynomial description of image conditions (the spectral defining ideal), and the discussion of reading them as Hecke eigenvalues via the Satake isomorphism, assume only (Sph-Hecke).

\medskip
\noindent
We also emphasize that the objects for which we discuss product/entangled are not the set of all quantum states, but rather semisimple conjugacy classes describable by (unramified) Satake parameters.

\subsubsection*{Notation and basic terms}
Let $F$ be a non-archimedean local field, $G$ a split reductive group, and $K\subset G$ a hyperspecial maximal compact subgroup. Write the spherical Hecke algebra as
\[
\mathcal H(G,K):=C_c^\infty(K\backslash G/K)
\]
(with convolution product).

Let $\pi$ be an irreducible smooth representation of $G$ with a nonzero $K$-fixed vector (a spherical representation). Then $\dim\pi^K=1$, so the action of $\mathcal H(G,K)$ on $V^K$ is by scalars, and one obtains the Hecke eigencharacter
\[
\lambda_\pi:\ \mathcal H(G,K)\to\C,\qquad
\pi(f)|_{\pi^K}=\lambda_\pi(f)\cdot \mathrm{id}.
\]
The ``parameters'' in this section are obtained by coordinatizing this eigencharacter via the Satake isomorphism.

\begin{definition}[Dual group $G^\vee$]
Let $G^\vee$ denote the connected complex reductive group defined by the dual root datum of $G$. In the basic example of this note,
\[
G=\PGL_n \quad\Longrightarrow\quad G^\vee=\SL_n(\C).
\]
\end{definition}

\begin{definition}[Dual torus]
Fix a split maximal torus $T\subset G$. The dual torus $T^\vee\subset G^\vee$ is the complex torus characterized by
\[
X^\ast(T^\vee)=X_\ast(T).
\]
\end{definition}

\begin{remark}[What we call a ``parameter'' in this section]
What we need is a point giving simultaneous eigenvalues for $\mathcal H(G,K)$. By the Satake isomorphism, this can be expressed as a point of $T^\vee/W_G$ (equivalently, a semisimple conjugacy class in $G^\vee$).
\end{remark}

\begin{definition}[Adjoint quotient (space of semisimple conjugacy classes)]
Define the adjoint (GIT) quotient for the conjugation action of $G^\vee$ by
\[
G^\vee\sslash G^\vee\ :=\ \Spec\bigl(\C[G^\vee]^{G^\vee}\bigr).
\]
This is the space of conjugacy classes identified by conjugation-invariant functions; over $\C$, semisimple conjugacy classes are naturally represented as points.
\end{definition}

\begin{theorem}[Chevalley restriction isomorphism]\label{thm:Chevalley}
Using a maximal torus $T^\vee\subset G^\vee$ and the Weyl group $W_G:=N_{G^\vee}(T^\vee)/T^\vee$, one has
\[
\C[G^\vee]^{G^\vee}\ \cong\ \C[T^\vee]^{W_G},
\qquad\text{hence}\qquad
G^\vee\sslash G^\vee\ \simeq\ T^\vee/W_G.
\]
\end{theorem}

\paragraph{Example: the case $G=\PGL_n$.}
Here $G^\vee=\SL_n(\C)$, $T^\vee$ is the diagonal torus ($\prod_{i=1}^n z_i=1$), and $W_G\simeq S_n$.
Thus
\[
\C[T^\vee]^{W_G}\ \cong\ \C[e_1,\dots,e_{n-1}]
\]
(where $e_k$ are the elementary symmetric polynomials and $e_n=\prod z_i=1$ is fixed),
and the ``Satake variables'' in this note correspond to generators of this invariant polynomial ring.

\subsubsection*{Segre variety and stabilizer (input from the geometric side)}
Fix $\mathbf d=(d_1,\dots,d_r)$ and set $n=\prod_{i=1}^r d_i$.
Viewing
\[
V \cong V_1\otimes\cdots\otimes V_r,\qquad \dim V_i=d_i,
\]
we obtain the Segre embedding
\[
\Sigma_{\mathbf d}:=\PP(V_1)\times\cdots\times \PP(V_r)\hookrightarrow \PP(V).
\]
Consider
\[
G:=\PGL(V)\simeq \PGL_n,\qquad
G_{\mathbf d}:=\mathrm{Stab}_G(\Sigma_{\mathbf d}),\qquad
X_{\mathbf d}:=G/G_{\mathbf d}.
\]
The space $X_{\mathbf d}$ is the space of subsystem structures of type $\mathbf d$.

\smallskip
\noindent
The shape of $G_{\mathbf d}$ and the dimension computations below are based on the standard fact that its identity component is essentially $\prod_i \PGL_{d_i}$ (the finite group part coming from permutations of equal-dimensional factors is inessential for the dimension and image-condition arguments, so we omit it from now on).

\subsection{Sphericity of $X_{\mathbf d}$}
\label{sec:Sph-var}
\begin{definition}[Spherical variety]
A $G$-variety $X$ is called spherical if some (hence any) Borel subgroup $B\subset G$ has an open orbit on $X$.
\end{definition}

\medskip
\noindent
For basic viewpoints in checking sphericity, see e.g.\ \cite{Brion1989}:
\begin{enumerate}
\item \textbf{(Existence of an open Borel orbit)}: does $B$ have an open orbit on $X$?
\item \textbf{(Necessary condition: a dimension inequality)}: if there is an open orbit, then $\dim(B\cdot x)=\dim X$, hence
\[
\dim X\ \le\ \dim B
\]
must hold (a \emph{necessary condition} for sphericity).
\item \textbf{(Rank and little Weyl group)}: a spherical $X$ has an associated Cartan part $A_X$ and little Weyl group $W_X$,
and one defines $\rank(X):=\dim A_X$.
\item \textbf{(Comparison with known classifications)}: when $G$ is classical and $H$ is reductive,
one can compare with classifications of reductive spherical subgroups (classically, Krämer \cite{Kraemer-SphericalSubgroupsClassical}, etc.).
\end{enumerate}

\begin{proposition}[Excluding sphericity by a dimension obstruction]\label{prop:sphericity-dim-test-clean}
Assume $d_i\ge2$, $r\ge2$, $n=\prod_i d_i$, and take $G=\PGL_n$, $X_{\mathbf d}=G/G_{\mathbf d}$.
The necessary condition for sphericity $\dim X_{\mathbf d}\le \dim B$ (for a Borel $B$) can hold only in the case
\[
\mathbf d=(2,2)\quad(n=4).
\]
Hence if $\mathbf d\neq(2,2)$, this particular homogeneous space $X_{\mathbf d}$ \emph{cannot be a spherical variety}
(at least it is excluded by the dimension obstruction).
\end{proposition}

\begin{proof}
First,
\[
\dim G=\dim \PGL_n=n^2-1.
\]
The dimension of a Borel subgroup of $\PGL_n$ is
\[
\dim B=\frac{n(n-1)}2+(n-1)=\frac{(n-1)(n+2)}2.
\]

Next, ignoring the finite group part (permutations of isomorphic factors), we may regard the identity component as
\[
G_{\mathbf d}^{\circ}\simeq \prod_{i=1}^r \PGL_{d_i}
\]
(a standard model coming from the evident symmetries of the Segre variety). Then
\[
\dim G_{\mathbf d}^{\circ}=\sum_{i=1}^r(d_i^2-1)=\sum_i d_i^2-r.
\]
Therefore
\[
\dim X_{\mathbf d}=\dim G-\dim G_{\mathbf d}^{\circ}
=(n^2-1)-\Bigl(\sum_i d_i^2-r\Bigr)
=n^2-\sum_i d_i^2+(r-1).
\]

The necessary condition $\dim X_{\mathbf d}\le \dim B$ is equivalent to
\[
n^2-\sum_i d_i^2+(r-1)\ \le\ \frac{(n-1)(n+2)}2,
\]
i.e.
\[
n^2-n-2\sum_i d_i^2+2r\ \le\ 0
\qquad(\ast).
\]

\smallskip\noindent
\textbf{(i) The case \(r=2\).}
Write \(\mathbf d=(a,b)\) with \(a,b\ge2\) and \(n=ab\). Then \((\ast)\) is equivalent to
\((a^2-2)(b^2-2)\le ab\). If \(a,b\ge3\), then \(a^2-2>a\) and \(b^2-2>b\), so the left-hand side is \(>ab\), a contradiction. Hence at least one of \(a,b\) equals \(2\). For instance, if \(a=2\), then
\(2(b^2-2)\le 2b\), i.e.\ \((b-2)(b+1)\le0\), so \(b=2\).

\smallskip\noindent
\textbf{(ii) The case \(r\ge 3\).}
Then \(n=\prod_i d_i\ge 2^3=8\). Also, since \(n/d_i\ge 2^{r-1}\), we have
\(d_i\le n/2^{r-1}\), hence \(d_i^2\le n^2/4^{r-1}\). Thus
\[
\sum_i d_i^2\le \frac{r}{4^{r-1}}\,n^2.
\]
Substituting into the left-hand side of \((\ast)\) gives
\[
n^2-n-2\sum_i d_i^2+2r
\ \ge\
n^2\Bigl(1-\frac{2r}{4^{r-1}}\Bigr)-n+2r.
\]
For $r\ge3$, one has $1-\frac{2r}{4^{r-1}}>0$, and substituting $n\ge8$ makes the right-hand side positive,
so \((\ast)\) cannot hold.
\end{proof}

\begin{remark}
Proposition~\ref{prop:sphericity-dim-test-clean} asserts that, for the construction $X_{\mathbf d}=G/G_{\mathbf d}$, sphericity survives essentially only for $(2,2)$.

\smallskip
On the other hand, the later arguments in this note that characterize image conditions (the $\mathbf d$-product locus) by invariant polynomials are purely algebraic and do not assume (Sph-var). Hence it is not a contradiction that one can obtain necessary and sufficient conditions even when $\mathbf d\neq(2,2)$ (e.g.\ Proposition~\ref{prop:finite-witness-revised} and Theorem~\ref{thm:3qubit-product-locus-note}).
\end{remark}

\subsection{Spherical Hecke algebra and the Satake isomorphism}
\textbf{Standing assumption (throughout this section)}:
Let $F$ be a non-archimedean local field, $G$ a split reductive group (in this note mainly $G=\PGL_n(F)$),
and $K\subset G$ a hyperspecial maximal compact subgroup (the unramified setting). In what follows we use only (Sph-Hecke).

\subsubsection*{The spherical Hecke algebra $\mathcal H(G,K)$}
From this point on we explain how the Hecke action is defined, and what the eigenvalues (Hecke eigenvalues) mean.

Let $F$ be a non-archimedean local field (e.g.\ a finite extension of $\mathbb Q_p$), $\cO$ its ring of integers,
$\varpi$ a uniformizer, and let $q$ be the size of the residue field.

In the concrete examples of this note we take
\[
G=\PGL_4(F),\qquad K=G(\cO).
\]
Here $K$ is a maximal compact subgroup and plays a central role in the spherical theory.

\begin{definition}[Spherical Hecke algebra]
Let
\[
\mathcal H(G,K):=C_c^\infty(K\backslash G/K)
\]
be the space of compactly supported, locally constant, bi-$K$-invariant complex-valued functions.
Normalize a Haar measure $dg$ by $\Vol(K)=1$. Then convolution
\[
(f_1*f_2)(g):=\int_{G} f_1(gh^{-1})f_2(h)\,dh
\]
makes $\mathcal H(G,K)$ into an algebra. This is called the (spherical) Hecke algebra.
\end{definition}

Since an $f$ is determined by its values on the double cosets $K\backslash G/K$, it is a finite linear combination of characteristic functions of finitely many double cosets. Convolution corresponds to composition of ``averaging first by $f_2$ and then by $f_1$'', so the algebra product encodes ``composition of operators''.

\begin{definition}[Convolution action (on a representation)]
Let $\pi$ be a smooth representation of $G$ with representation space $V$. For $f\in\mathcal H(G,K)$ define
\[
\pi(f)v := \int_G f(g)\,\pi(g)v\,dg,
\]
which gives an operator $\pi(f):V\to V$.
\end{definition}

Since $v$ is smooth and $f$ has compact support, the integral above is in practice a finite sum.

When \(\pi\) is irreducible and has a \(K\)-fixed vector (spherical / unramified), one has \(\dim \pi^K = 1\).
Then every $f\in\mathcal H(G,K)$ acts on $V^K$ by a scalar, and
\[
\pi(f)v_K=\lambda_\pi(f)\,v_K\qquad (0\neq v_K\in V^K).
\]
This scalar $\lambda_\pi(f)$ is called the \textbf{Hecke eigenvalue}.

Since $\mathcal H(G,K)$ is commutative, the $K$-spherical vector behaves like a \emph{simultaneous eigenstate of commuting observables}. The eigenvalue data are packaged by the Satake parameter in the next subsection.

\subsubsection*{The Satake isomorphism and Satake parameters}
The Satake isomorphism gives
\[
\mathrm{Sat}_G:\ \mathcal H(G,K)\xrightarrow{\sim}\C[T^\vee]^{W_G}.
\]
Hence the Hecke eigenvalues $\lambda_\pi$ of a spherical representation $\pi$ can be written as evaluation at a point $[s(\pi)]\in T^\vee/W_G$:
\[
\lambda_\pi(f)=\mathrm{Sat}_G(f)\bigl(s(\pi)\bigr)\qquad(f\in\mathcal H(G,K)).
\]
This $[s(\pi)]$ is called the (spherical) Satake parameter.

\begin{remark}[Operators, data, and parameters]\
\begin{itemize}
\item Hecke algebra: a commutative algebra of ``averaging (convolution) operators'' on $G(F)$.
\item Satake isomorphism: identifies the commutative algebra $\mathcal H(G,K)$ with symmetric polynomials on the dual torus.
\item Satake parameter: the ``point on the torus'' (a conjugacy class in the dual group) giving the eigenvalues.
\end{itemize}
\end{remark}

\subsection{The case $\mathbf d=(2,2)$}\label{sec:two-qubit_Langlands}

\subsubsection*{(1) Setup (tensor-product homomorphism and the torus embedding)}
For $\mathbf d=(2,2)$, consider the tensor-product homomorphism of complex groups
\[
r:\ \SL_2(\C)\times \SL_2(\C)\ \longrightarrow\ \SL_4(\C),\qquad (g_1,g_2)\longmapsto g_1\otimes g_2.
\]

For the standard torus elements in $\SL_2$,
\[
t(a)=\begin{pmatrix}a&0\\0&a^{-1}\end{pmatrix},\qquad
t(b)=\begin{pmatrix}b&0\\0&b^{-1}\end{pmatrix},
\]
the eigenvalues of $t(a)\otimes t(b)$ are
\begin{equation}\label{eq:eivenvalues}
(ab,\ ab^{-1},\ a^{-1}b,\ a^{-1}b^{-1}).
\end{equation}
Hence we obtain an embedding of tori
\[
\iota:\ (\C^\times)^2\to (\C^\times)^4,\qquad
(a,b)\mapsto (ab,ab^{-1},a^{-1}b,a^{-1}b^{-1}).
\]

\begin{definition}[Restriction map of the Satake isomorphism ($\mathbf d=(2,2)$)]
Define
\[
\mathrm{Res}_{(2,2)}:=\iota^\ast\circ\mathrm{Sat}_G:\ \mathcal H(G,K)\to \C[a^{\pm1},b^{\pm1}].
\]
By the symmetries ($a\mapsto a^{-1}$, $b\mapsto b^{-1}$, and swapping $a$ and $b$), the image lies naturally in
$\C[a^{\pm1},b^{\pm1}]^{W_{(2,2)}}$, so when needed we view it as
\[
\mathrm{Res}_{(2,2)}:\ \mathcal H(G,K)\to \C[a^{\pm1},b^{\pm1}]^{W_{(2,2)}}.
\]
\end{definition}

\subsubsection*{(2) Pullback of $e_1,e_2,e_3$}
We now compute $\iota^\ast(e_1),\iota^\ast(e_2),\iota^\ast(e_3)$ explicitly, aiming to write Hecke eigenvalues as functions of $(a,b)$.

Let $z_i$ be the image coordinates under $\iota$ (so $z_1=ab$, $z_2=ab^{-1}$, $z_3=a^{-1}b$, $z_4=a^{-1}b^{-1}$). Then:
\begin{enumerate}
\item Since $e_1=z_1+z_2+z_3+z_4$,
\[
\iota^\ast(e_1)=ab+ab^{-1}+a^{-1}b+a^{-1}b^{-1}
=(a+a^{-1})(b+b^{-1}).
\]
\item For $e_2=\sum_{i<j} z_i z_j$, compute term by term:
\[
\begin{aligned}
z_1z_2&=(ab)(ab^{-1})=a^2,\\
z_1z_3&=(ab)(a^{-1}b)=b^2,\\
z_1z_4&=(ab)(a^{-1}b^{-1})=1,\\
z_2z_3&=(ab^{-1})(a^{-1}b)=1,\\
z_2z_4&=(ab^{-1})(a^{-1}b^{-1})=b^{-2},\\
z_3z_4&=(a^{-1}b)(a^{-1}b^{-1})=a^{-2}.
\end{aligned}
\]
Summing gives
\[
\iota^\ast(e_2)=a^2+a^{-2}+b^2+b^{-2}+2.
\]
\item Since $e_3=\sum_{i<j<k} z_i z_j z_k$ and $\prod z_i=1$, we have
\[
e_3=(z_1z_2z_3z_4)\sum_{i=1}^4 z_i^{-1}=\sum_{i=1}^4 z_i^{-1}.
\]
In this case the set $\{z_1,z_2,z_3,z_4\}$ is closed under inversion ($z_1^{-1}=z_4$, $z_2^{-1}=z_3$), hence
\[
\iota^\ast(e_3)=z_1^{-1}+z_2^{-1}+z_3^{-1}+z_4^{-1}=z_1+z_2+z_3+z_4=\iota^\ast(e_1).
\]
\end{enumerate}

\subsubsection*{(3) Hecke eigenvalues for the spherical spectrum}
We normalized $T_k$ by $\mathrm{Sat}_G(T_k)=e_k$, hence
\[
\mathrm{Res}_{(2,2)}(T_k)=\iota^\ast(e_k).
\]
Therefore
\[
\mathrm{Res}_{(2,2)}(T_1)=(a+a^{-1})(b+b^{-1}),\quad
\mathrm{Res}_{(2,2)}(T_2)=a^2+a^{-2}+b^2+b^{-2}+2,\quad
\mathrm{Res}_{(2,2)}(T_3)=\mathrm{Res}_{(2,2)}(T_1).
\]
\begin{remark}[Normalization of the Satake isomorphism]
In the literature, one often uses a normalization in which the Satake image of the characteristic function $\mathbf 1_{K\varpi^\lambda K}$ carries a factor $q^{\langle\rho,\lambda\rangle}$ (or $\delta_B^{1/2}$). In this note we define $T_k:=\mathrm{Sat}_G^{-1}(e_k)$ and absorb that $q$-factor into the normalization.
\end{remark}

These results can be summarized as follows:
\begin{itemize}
\item $\SL_4(\C)\sslash \SL_4(\C)\simeq T^\vee/W_G$ has $\dim=3$ (one can take coordinates $e_1,e_2,e_3$).
\item On the other hand $(\SL_2(\C)\times \SL_2(\C))\sslash(\SL_2(\C)\times \SL_2(\C))$ has $\dim=2$ (parameters $(a,b)$).
\item Hence the image has codimension $1$, and indeed the minimal relation is $e_1=e_3$ (equivalently $T_1-T_3=0$).
\end{itemize}

\subsubsection*{(4) The product locus and an ``entanglement witness''}
Here, building on the results above, we discuss how product states and entangled states appear in the image of the Satake map.

\begin{definition}[$(2,2)$-product / entangled]
A semisimple conjugacy class $[s]\in \SL_4(\C)/\!\sim$ is called $(2,2)$-\emph{product} if it has the eigenvalue pattern
\[
(ab,\ ab^{-1},\ a^{-1}b,\ a^{-1}b^{-1})
\]
(up to reordering). Otherwise it is called $(2,2)$-\emph{entangled}.
\end{definition}

Let us check that this definition matches the product/entangled notion in quantum information. If $A\in \SL_2$ has eigenvalues $\{\alpha,\alpha^{-1}\}$ and $B\in \SL_2$ has eigenvalues $\{\beta,\beta^{-1}\}$, then the eigenvalues of $A\otimes B$ become pairwise products as in \eqref{eq:eivenvalues}.

This is the same basic fact in quantum information that the spectrum of a local gate
\[
U=U_A\otimes U_B\quad(\text{local gate})
\]
decomposes as a ``product of local spectra.'' Thus,
\[
s \text{ can be written as } s_A\otimes s_B
\quad\Longleftrightarrow\quad
\text{the spectrum decomposes into two ``local parameters''},
\]
which is exactly the product situation.

\begin{proposition}[Equivalent reformulations of the image condition]
Let $s\in \SL_4(\C)$ be semisimple with (unordered) eigenvalues $\{z_1,z_2,z_3,z_4\}$. The following are equivalent:
\begin{enumerate}
\item $s$ is $(2,2)$-product.
\item The eigenvalue set is invariant under inversion:
\[
\{z_1,z_2,z_3,z_4\}=\{z_1^{-1},z_2^{-1},z_3^{-1},z_4^{-1}\},
\]
i.e.\ it can be written as $\{u,u^{-1},v,v^{-1}\}$.
\item The characteristic polynomial is palindromic:
\[
\chi_s(T)=T^4-c_1T^3+c_2T^2-c_1T+1.
\]
\item As symmetric functions, $e_1(s)=e_3(s)$.
\end{enumerate}
\end{proposition}

\begin{proof}
(1)$\Rightarrow$(2):
Since $(ab)^{-1}=a^{-1}b^{-1}$ and $(ab^{-1})^{-1}=a^{-1}b$, the eigenvalue set on the image is closed under inversion.

(2)$\Leftrightarrow$(3):
The coefficients of the characteristic polynomial are symmetric polynomials, and for $\SL_4$ we have $e_4=z_1z_2z_3z_4=1$.
In general,
\[
\chi_s(T)=T^4-e_1T^3+e_2T^2-e_3T+e_4.
\]
Palindromic means $T^4\chi_s(T^{-1})=\chi_s(T)$, which by comparing coefficients is equivalent to $e_1=e_3$.

(3)$\Leftrightarrow$(4):
This is immediate from the general formula above.

(2)$\Rightarrow$(1):
Assume $\{z_i\}=\{u,u^{-1},v,v^{-1}\}$. Choose $a\in\C^\times$ such that $a^2=uv$ (always possible over $\C$). Then set $b:=u/a$, so $b\in\C^\times$ and $ab=u$ and $a/b=v$ hold. Thus, after reordering, $\{ab,ab^{-1},a^{-1}b,a^{-1}b^{-1}\}=\{u,v,u^{-1},v^{-1}\}$, so $s$ is $(2,2)$-product.
\end{proof}

From the computation above we have $\iota^\ast(e_1)=\iota^\ast(e_3)$. Hence
\[
e_1-e_3\in \ker(\iota^\ast)\subset \C[e_1,e_2,e_3].
\]
On the other hand, $\C[e_1,e_2,e_3]$ is a UFD, and the image of $\iota$ (the image on the adjoint quotient) is an irreducible set of dimension $2$, so $\ker(\iota^\ast)$ is a height-$1$ prime ideal. Therefore
\[
\ker(\iota^\ast)=(e_1-e_3).
\]
Via the Satake isomorphism $T_k:=\mathrm{Sat}_G^{-1}(e_k)$, we can write
\[
\ker(\mathrm{Res}_{(2,2)})=(T_1-T_3)\subset \mathcal H(G,K).
\]

\begin{remark}[Comparison with entanglement witnesses in quantum information]
In quantum information, an entanglement witness is defined as an inequality that separates the convex set of separable states in a Hahn--Banach sense (e.g.\ $\Tr(W\rho)<0$ implies entangled). Thus a single witness does not completely classify all states.

\smallskip
In contrast, the spectrum used in this note for the product/entanglement test is from the outset restricted to \textbf{(unramified) Hecke simultaneous eigenstates}. The eigenvalues are given by evaluating invariant polynomials on the Satake parameter $s(\pi)\in G^\vee\sslash G^\vee$. In this sense, in the $(2,2)$ case we have
\[
(T_1-T_3)v=\bigl(e_1(s(\pi))-e_3(s(\pi))\bigr)\,v,
\]
and
\[
e_1(s(\pi))=e_3(s(\pi))
\quad\Longleftrightarrow\quad
s(\pi)\ \text{lies in the $(2,2)$-product locus},
\]
so the product locus can be cut out exactly (within the unramified spectrum) by an equality condition.
\end{remark}

\section{Extension to General $\mathbf d$ Cases}
\label{sec:general_types}
In principle, the structure of the discussion in this section consists of the following two parts:
\begin{enumerate}
\item \textbf{(Algebraic)}: Define the $\mathbf d$-product locus as the image (adjoint quotient) of the tensor-product map on the dual group,
and detect it by finitely many invariant polynomials.
\item \textbf{(Hecke)}: Via the Satake isomorphism, pull those polynomials back to Hecke operators and detect the condition as unramified Hecke eigenvalues. Here the assumption (Sph-Hecke) is required.
\end{enumerate}

\subsection{Tensor-product map and $\mathbf d$-product (algebraic definition)}

In what follows, let $n=\prod_{i=1}^r d_i$, and we work in the local unramified, split setting.

\begin{definition}[Dual group in the general case and the $\mathbf d$-side dual group]
For a split type $\mathbf d=(d_1,\dots,d_r)$, set
\[
G^\vee=\SL_n(\C),\qquad
G_{\mathbf d}^\vee:=\prod_{i=1}^r \SL_{d_i}(\C)
\]
(and, if necessary, adjoin a finite group coming from permutations of equal factors). In this note, in order to describe the image condition of the tensor-product map, we use only this most direct model.
\end{definition}

\begin{definition}[Tensor-product homomorphism]
We call
\[
r_{\mathbf d}:\ G_{\mathbf d}^\vee\longrightarrow G^\vee,\qquad
r_{\mathbf d}(g_1,\dots,g_r)=g_1\otimes\cdots\otimes g_r
\]
the tensor-product homomorphism.
\end{definition}

If $s$ lies in the image of $r_{\mathbf d}$, then the eigenvalues of $s$ (with multiplicities) are of the form
\[
\Bigl\{\ \prod_{i=1}^r a_{i,j_i}\ :\ 1\le j_i\le d_i\ \Bigr\},
\]
where we write the eigenvalues of $g_i$ as $a_{i,1},\dots,a_{i,d_i}$. For $\mathbf d=(2,2)$ this could be written as in \eqref{eq:eivenvalues}.

\begin{definition}[$\mathbf d$-product / entangled]
A semisimple conjugacy class $[s]\in G^\vee\sslash G^\vee$ is called \emph{$\mathbf d$-product} if
\[
[s]\in \mathrm{Im}(r_{\mathbf d}\sslash),
\]
and otherwise it is called \emph{$\mathbf d$-entangled}.
\end{definition}

\subsection{Invariant polynomials and entanglement}
\begin{definition}[Symmetric group and the Weyl group of $\SL_m$]
For $m\ge 1$, let $S_m$ be the group of all permutations of the set $\{1,\dots,m\}$ (the symmetric group on $m$ letters).
For the standard maximal torus $T_m^\vee$ of $\SL_m(\C)$ (the diagonal matrices whose diagonal entries multiply to $1$),
the Weyl group
\[
W(\SL_m)\ :=\ N_{\SL_m(\C)}(T_m^\vee)\,/\,T_m^\vee
\]
acts by conjugation via permutation matrices, hence permutes the diagonal entries. Thus it can be naturally identified with
\[
W(\SL_m)\ \cong\ S_m.
\]
Below, for a splitting $\mathbf d$ of a subsystem, we set
\[
W_{\mathbf d}:=\prod_i W(\SL_{d_i})\cong \prod_i S_{d_i}
\]
(and, if necessary, adjoin the permutation group of equal factors via a semidirect product).
\end{definition}

\begin{remark}
Here $W(\SL_2)\cong \Z/2$ acts by $a\mapsto a^{-1}$, hence
\[
W_{\mathrm{dom}}:=W(\SL_2)\times W(\SL_2)\cong (\Z/2)^2
\]
acts on $(a,b)$. Since the two factors are isomorphic, one may also consider the extension by the swap action
\[
W_{(2,2)}^{\mathrm{ext}}:=W_{\mathrm{dom}}\rtimes S_2,
\]
but for the image condition in this article (the generators of the kernel), there is no essential difference whichever one is used.
\end{remark}

First, $r_{\mathbf d}$ induces a morphism on adjoint quotients
\[
r_{\mathbf d}\sslash:\ (G_{\mathbf d}^\vee\sslash G_{\mathbf d}^\vee)\ \longrightarrow\ (G^\vee\sslash G^\vee).
\]
On coordinate rings (rings of invariants), contravariantly we obtain
\[
r_{\mathbf d}^\ast:\ \C[G^\vee]^{G^\vee}\ \longrightarrow\ \C[G_{\mathbf d}^\vee]^{G_{\mathbf d}^\vee}.
\]
By the Chevalley restriction isomorphism (Theorem~\ref{thm:Chevalley}),
\[
\C[G^\vee]^{G^\vee}\cong \C[T^\vee]^{W_G},\qquad
\C[G_{\mathbf d}^\vee]^{G_{\mathbf d}^\vee}\cong \C[T^\vee_{\mathbf d}]^{W_{\mathbf d}},
\]
so in this note we write, via these identifications,
\[
r_{\mathbf d}^\ast:\ \C[T^\vee]^{W_G}\longrightarrow \C[T^\vee_{\mathbf d}]^{W_{\mathbf d}}
\]
(where $T^\vee_{\mathbf d}$ is a maximal torus of $G_{\mathbf d}^\vee$).

\begin{definition}[spectral defining ideal]
Define
\[
I_{\mathbf d}\ :=\ \ker(r_{\mathbf d}^\ast)\ \subset\ \C[T^\vee]^{W_G}
\]
and call it the \emph{spectral defining ideal} defining the $\mathbf d$-product locus.
\end{definition}

An element $f\in I_{\mathbf d}$ means that $f$ vanishes identically on the tensor-product image.
Hence, once generators of $I_{\mathbf d}$ are known, the image condition can be detected by finitely many equations.

\begin{proposition}[Algebraic fact]\label{prop:finite-witness-revised}
The ideal $I_{\mathbf d}$ is finitely generated. Hence there exist finitely many elements
$f_1,\dots,f_m\in \C[T^\vee]^{W_G}$ such that
\[
I_{\mathbf d}=(f_1,\dots,f_m).
\]
Then for any semisimple conjugacy class $[s]\in G^\vee\sslash G^\vee$,
\[
[s]\ \text{is $\mathbf d$-product}
\quad\Longrightarrow\quad
f_1(s)=\cdots=f_m(s)=0.
\]
Moreover, since the image of $r_{\mathbf d}\sslash$ is Zariski closed (Lemma~\ref{lem:rd-image-closed} below), the converse also holds:
\[
[s]\ \text{is $\mathbf d$-product}
\quad\Longleftrightarrow\quad
f_1(s)=\cdots=f_m(s)=0.
\]
\end{proposition}

\begin{proof}[proof sketch]
Since $\C[T^\vee]^{W_G}$ is Noetherian, $I_{\mathbf d}=\ker(r_{\mathbf d}^\ast)$ is finitely generated. The necessity and sufficiency depend on the image being Zariski closed, which is shown in Lemma~\ref{lem:rd-image-closed}.
\end{proof}

\begin{lemma}[Zariski-closedness of the image of $r_\mathbf d\sslash$ in Proposition~\ref{prop:finite-witness-revised}]\label{lem:rd-image-closed} Over $\C$, let $G^\vee=\SL_n$ and $G_{\mathbf d}^\vee=\prod_i \SL_{d_i}$, and consider the tensor-product homomorphism $r_{\mathbf d}:G_{\mathbf d}^\vee\to G^\vee$. Then the induced morphism on adjoint quotients
\[
r_{\mathbf d}\sslash:\ (G_{\mathbf d}^\vee\sslash G_{\mathbf d}^\vee)\longrightarrow (G^\vee\sslash G^\vee)
\]
has Zariski-closed image.
\end{lemma}

\begin{proof}
Let $T_{\mathbf d}^\vee\subset G_{\mathbf d}^\vee$ and $T^\vee\subset G^\vee$ be maximal tori, and let $W_{\mathbf d},W_G$ be the corresponding Weyl groups. By the Chevalley restriction isomorphism,
\[
G_{\mathbf d}^\vee\sslash G_{\mathbf d}^\vee \simeq T_{\mathbf d}^\vee/W_{\mathbf d},\qquad
G^\vee\sslash G^\vee \simeq T^\vee/W_G,
\]
and $r_{\mathbf d}\sslash$ corresponds to $T_{\mathbf d}^\vee/W_{\mathbf d}\to T^\vee/W_G$.

First, $r_{\mathbf d}$ restricts to a homomorphism of tori
\[
r_{\mathbf d}|_{T_{\mathbf d}^\vee}:\ T_{\mathbf d}^\vee\longrightarrow T^\vee.
\]
The image of a morphism of algebraic tori is an algebraic subgroup (a subtorus), hence Zariski closed. Therefore
\[
S:=r_{\mathbf d}(T_{\mathbf d}^\vee)\subset T^\vee
\]
is a closed subset.

Next, the quotient map by the finite group $W_G$, namely $q:T^\vee\to T^\vee/W_G$, is a finite morphism and hence a closed map. Thus $q(S)\subset T^\vee/W_G$ is closed.

On the other hand, $T_{\mathbf d}^\vee\to T_{\mathbf d}^\vee/W_{\mathbf d}$ is also surjective, so the image of the induced map $T_{\mathbf d}^\vee/W_{\mathbf d}\to T^\vee/W_G$ agrees with the image seen from $T_{\mathbf d}^\vee$:
\[
\operatorname{Im}\bigl(T_{\mathbf d}^\vee/W_{\mathbf d}\to T^\vee/W_G\bigr)
=\operatorname{Im}\bigl(T_{\mathbf d}^\vee\to T^\vee/W_G\bigr)
=q(S).
\]
Hence the image is closed.
\end{proof}

\subsection{Detection in terms of Hecke eigenvalues}
Via $\mathrm{Sat}_G:\mathcal H(G,K)\xrightarrow{\sim}\C[T^\vee]^{W_G}$, for each $f_i$ take the corresponding Hecke operator
\[
\mathsf W_i:=\mathrm{Sat}_G^{-1}(f_i)\in\mathcal H(G,K).
\]
Then for an unramified representation $\pi$ one obtains
\[
[s(\pi)]\in \mathrm{Im}(r_{\mathbf d}\sslash)
\quad\Longrightarrow\quad
\lambda_\pi(\mathsf W_1)=\cdots=\lambda_\pi(\mathsf W_m)=0
\]
(the necessity and sufficiency depend on the closedness of the image).

\begin{remark}
This part---detecting the image condition using Hecke eigenvalues---works using only (Sph-Hecke), and is independent of whether $X_{\mathbf d}$ is a spherical variety.
\end{remark}

\subsection{The 3-qubit case ($\mathbf d=(2,2,2)$)}\label{sec:3qubit}

For $\mathbf d=(2,2,2)$, the quotient $X_{\mathbf d}=G/G_{\mathbf d}$ is not a spherical variety by Proposition~\ref{prop:sphericity-dim-test-clean}. However, the $(2,2,2)$-product locus (the image condition on the adjoint quotient) can be described purely algebraically.

\medskip
Consider the tensor-product homomorphism
\[
r:=r_{(2,2,2)}:\ \SL_2(\C)^3\longrightarrow \SL_8(\C),
\]
and the induced pullback on coordinate rings of invariants
\[
r^\ast:\ \C[\SL_8]^{\SL_8}\ \longrightarrow\ \C[\SL_2^3]^{\SL_2^3}.
\]
By the Chevalley restriction isomorphism (Theorem~\ref{thm:Chevalley}), we may identify
\[
\C[\SL_8]^{\SL_8}\ \cong\ \C[e_1,\dots,e_7],\qquad
\C[\SL_2^3]^{\SL_2^3}\ \cong\ \C[x,y,z]
\]
(where, for instance, $x$ corresponds to $a+a^{-1}$, etc.).
Let
\[
I_{(2,2,2)}:=\ker(r^\ast)\subset \C[e_1,\dots,e_7]
\]
be the (scheme-theoretic) defining ideal of the $(2,2,2)$-product locus.

Let the quartic polynomial used below be
\begin{equation}\label{eq:F-3qubit}
F
:=e_3^2+2e_1e_3+e_1^4-(e_4+2e_2+1)e_1^2\ \in \C[e_1,e_2,e_3,e_4].
\end{equation}

\begin{lemma}[Irreducibility and primeness of the quartic $F$]\label{lem:F-irreducible}
The polynomial $F\in \C[e_1,e_2,e_3,e_4]$ is irreducible. Hence $(F)$ is a prime ideal and, in particular, radical.
Moreover,
\[
J:=\bigl(e_7-e_1,\ e_6-e_2,\ e_5-e_3,\ F\bigr)\subset \C[e_1,\dots,e_7]
\]
is also a prime ideal (hence radical).
\end{lemma}

\begin{proof}
Over the coefficient ring $A:=\C[e_1,e_2,e_4]$, view $F$ as a polynomial in $e_3$; then $F\in A[e_3]$ is a monic quadratic. A monic quadratic is reducible if and only if its discriminant is a square in $A$, so compute the discriminant:
\[
\Delta=(2e_1)^2-4\bigl(e_1^4-(e_4+2e_2+1)e_1^2\bigr) =4e_1^2\,(e_4+2e_2+2-e_1^2).
\]
Since $4e_1^2$ is already a square in $A$, for $\Delta$ to be a square it is necessary that
\[
L:=e_4+2e_2+2-e_1^2\in A
\]
be a square. But $L$ is linear in $e_4$, whereas any square in $A$ has even degree in $e_4$, so $L$ cannot be a square.
Therefore $F$ is irreducible, hence $(F)$ is prime.

Next, in the ideal $J$ we can eliminate $e_5,e_6,e_7$ using $e_5=e_3,\ e_6=e_2,\ e_7=e_1$, so
\[
\C[e_1,\dots,e_7]/J\ \cong\ \C[e_1,e_2,e_3,e_4]/(F).
\]
The right-hand side is an integral domain (since $(F)$ is prime), hence so is the left-hand side, and therefore $J$ is a prime ideal.
\end{proof}

\begin{theorem}[Product locus for 3-qubit: image condition and defining ideal]\label{thm:3qubit-product-locus-note}
For a semisimple element $s\in \SL_8(\C)$, the following are equivalent:
\begin{enumerate}
\item[\textnormal{(a)}]
$[s]\in \SL_8(\C)\sslash \SL_8(\C)$ is $(2,2,2)$-product, i.e.\ there exists $(a,b,c)\in(\C^\times)^3$ such that the multiset of eigenvalues is
\[
\{\,a^{\pm1}b^{\pm1}c^{\pm1}\,\}.
\]
\item[\textnormal{(b)}]
The following three palindromic relations and the quartic equation \eqref{eq:F-3qubit} hold:
\[
e_7(s)=e_1(s),\qquad e_6(s)=e_2(s),\qquad e_5(s)=e_3(s),
\]
\[
F(s)=0.
\]
\end{enumerate}
Moreover, scheme-theoretically,
\[
I_{(2,2,2)}=\ker(r^\ast)=\bigl(e_7-e_1,\ e_6-e_2,\ e_5-e_3,\ F\bigr).
\]
\end{theorem}

\begin{proof}
\textbf{(a)$\Rightarrow$(b).}
If $s$ is $(2,2,2)$-product, then its eigenvalue multiset is invariant under $z\mapsto z^{-1}$ (and has product $1$), so the characteristic polynomial is self-reciprocal and thus $e_k=e_{8-k}$ for $k=1,2,3$. Also, writing $x=a+a^{-1}$, $y=b+b^{-1}$, $z=c+c^{-1}$, and setting
\[
E(t):=\prod_{\varepsilon\in\{\pm1\}^3}\bigl(1+t\,a^{\varepsilon_1}b^{\varepsilon_2}c^{\varepsilon_3}\bigr)
=\sum_{k=0}^8 e_k\,t^k,
\]
a straightforward simplification (as in the earlier computation) yields $F=0$.

\textbf{(b)$\Rightarrow$(a).}
Assume (b). By the palindromic relations, $e_5,e_6,e_7$ are determined by $e_1,e_2,e_3$, so it is essentially enough to check the consistency of $(e_1,e_2,e_3,e_4)$. Then, by the reconstruction already given earlier (splitting into the cases $e_1\neq 0$ and $e_1=0$), one can construct some $(a,b,c)\in(\C^\times)^3$ and
\[
t=\mathrm{diag}(a,a^{-1})\otimes \mathrm{diag}(b,b^{-1})\otimes \mathrm{diag}(c,c^{-1})\in\SL_8(\C)
\]
having the same elementary symmetric polynomials (hence the same eigenvalue multiset) as $s$. Therefore $[s]=[t]\in\Im(r\sslash)$, so $s$ is $(2,2,2)$-product.

\medskip
\textbf{Proof of the ideal equality $I_{(2,2,2)}=J$.}
Since (a)$\Rightarrow$(b) was shown, the ideal
\[
J:=\bigl(e_7-e_1,\ e_6-e_2,\ e_5-e_3,\ F\bigr)
\]
vanishes on the image, hence $J\subset I_{(2,2,2)}$. On the other hand, (b)$\Rightarrow$(a) implies that on the adjoint quotient,
\[
V(J)=\Im(r\sslash)
\]
(i.e.\ the vanishing locus coincides with the image). Hence by the Nullstellensatz,
\[
\sqrt{J}=I_{(2,2,2)}.
\]
Since $I_{(2,2,2)}$ is the kernel of a homomorphism to an integral domain, it is prime and therefore radical. Finally, by Lemma~\ref{lem:F-irreducible}, $J$ is prime (in particular radical), so $\sqrt{J}=J$. Thus $I_{(2,2,2)}=J$.
\end{proof}

\begin{corollary}[Detection by Hecke eigenvalues]\label{cor:3qubit-hecke-note}
Let $\pi$ be an irreducible unramified representation of $\PGL_8(F)$, with Satake parameter $[s(\pi)]$ and Hecke eigenvalues $\lambda_k:=e_k(s(\pi))$. Then $\pi$ is $(2,2,2)$-product (i.e.\ $[s(\pi)]\in\Im(r_{(2,2,2)}\sslash)$) if and only if
\[
\lambda_7=\lambda_1,\qquad \lambda_6=\lambda_2,\qquad \lambda_5=\lambda_3,
\]
\[
\lambda_3^2+2\lambda_1\lambda_3+\lambda_1^4-(\lambda_4+2\lambda_2+1)\lambda_1^2=0.
\]
\end{corollary}

\begin{remark}
For $\mathbf d=(2,2)$, a single defining equation suffices because the codimension is $1$ and the defining ideal becomes principal. For $\mathbf d=(2,2,2)$, we have $\rank(\SL_8)=7$ while $\rank(\SL_2^3)=3$, so the codimension is $4$; it is therefore natural that multiple (and nonlinear) defining equations are needed.
\end{remark}

\section{Summary}\label{sec:summary}

The discussion in the first half about local properties of entanglement is upgraded in the second half to the following global geometry:

\begin{itemize}
  \item The fibration $SB(A)\to X$ exists \emph{for any} Azumaya background and provides a canonical
  family of pure-state spaces.
  \item However, the existence of the product locus / entanglement filtration is not automatic: one needs a \emph{subsystem structure} (a reduction to $G_d$), and there can be an obstruction.
  \item The moduli of subsystem structures is represented by $P/G_d$, and it agrees with the subsystem-structure locus inside the relative Hilbert scheme of $SB(A)$.
  \item In the bipartite case, the full determinantal properties for $R_{\le k}$ (flatness, singularities, incidence resolution, numerical invariants) descends directly to $\Sigma_{\le k}(A,d)\subset SB(A)$.
  \item Brauer classes measure whether a subsystem can exist globally in the first place, producing the subset $Br_d(X)\subset Br(X)$ and the torsion constraint $Br_d(X)\subset Br(X)[\ell]$.
  \item Nevertheless, subsystem reducibility is not determined by the Brauer class alone: it depends on finer information of the underlying $\PGL_n$-torsor (and can fail even in the split case).
\end{itemize}

\paragraph{Detecting product/entangled states via the subsystem-structure locus}
We can summarize how to decide whether entanglement exists as follows.

First, choose the subsystem dimensions (if they exist) as a type $d=(d_1,d_2,\cdots,d_r)$, and consider the locally closed subscheme
\[
\Hilb^{\Sigma_d}(SB(A)/X)
\]
(subsystem-structure locus) inside the relative Hilbert scheme of the Severi--Brauer scheme $\pi:SB(A)\to X$. This is the moduli of subsystem structures ($G_d$-reductions), and as an $X$-scheme it can be identified with
\[
\Hilb^{\Sigma_d}(SB(A)/X)\simeq P/G_d.
\]

A coordinate-independent notion of separability/entanglement can be defined globally if and only if the morphism $\Hilb^{\Sigma_d}(SB(A)/X)\to X$ admits a section $s:X\to \Hilb^{\Sigma_d}(SB(A)/X)$. In that case, the section $s$ produces a closed subfamily $\Sigma_d(A)\subset SB(A)$ (the product locus).

Hence, for a pure state $[\psi]\in SB(A)_x$ at each point $x\in X$, one can decide:
\[
[\psi]\ \text{is product} \ \Longleftrightarrow\  [\psi]\in \Sigma_d(A)_x,
\qquad
[\psi]\ \text{is entangled} \ \Longleftrightarrow\  [\psi]\notin \Sigma_d(A)_x.
\]
(Locally in coordinates, this agrees with the vanishing condition of $2\times 2$ minors explained in the
introduction (\S\ref{sec:elementary_example}).)

\subsubsection*{An entanglement criterion via Hecke operators}
Moreover, via the discussion of Satake parameters, the entangling monodromy can be related to eigenvalues of Hecke operators, and in the two-qubit case we saw that
\[
\text{product state}\ \Longleftrightarrow\ \text{(a certain) eigenvalue }=0.
\]
In the case $\mathbf d=(2,2)$ emphasized in this note, the $(2,2)$-product condition in the Satake parameter space reduces to a single relation
\[
e_1(s)=e_3(s),
\]
where $e_k$ denotes the $k$-th elementary symmetric polynomial in the eigenvalues.

This is completely parallel to the quantum-information side (as reviewed in the introduction), where
\[
\Sigma_{2,2}\subset \PP^3 \ \text{ is defined by a single equation }\ \det(C)=0.
\]
In other words, we obtain the following perfectly analogous dictionary:
\begin{center}
\begin{tabular}{@{}ll@{}}
\toprule
Quantum information (geometry of states) & Langlands (geometry of the Hecke spectrum) \\
\midrule
product locus $\Sigma_{2,2}$ & $(2,2)$-product locus $\mathrm{Im}(r)$ \\
single equation $\det(C)=0$ & single equation $e_1-e_3=0$ \\
complement is entangled states & complement is $(2,2)$-entangled parameters \\
\bottomrule
\end{tabular}
\end{center}

For the 3-qubit case, a necessary and sufficient condition is given by a combination of the nonlinear equation and linear equations (Theorem \ref{thm:3qubit-product-locus-note}). More generally, for $n$ qubits, the tensor-product morphism $r_\mathbf d$ on the dual side determines a \emph{spectral defining ideal}; by Noetherianity it is finitely generated. Hence the (Zariski) closure of the $\mathbf d$-product locus can be described as the common zero set of finitely many invariant polynomials (and if the image is Zariski closed, this gives an equivalence).

Furthermore, if the relative Langlands framework of Sakellaridis--Venkatesh \cite{2012arXiv1203.0039S,2006math......6130S,2017arXiv170208264K} is established, one may expect a more structural understanding in which these image conditions appear as representation-theoretic spectral selection rules. (However, as Proposition~\ref{prop:sphericity-dim-test-clean} shows, $X_{\mathbf d}=G/G_{\mathbf d}$ is not spherical in general, so applying the SV formalism directly beyond $\mathbf d=(2,2)$ requires additional ideas.)

Regardless of practical usefulness, it is intriguing that the basic question of whether a quantum state factorizes can be seen to connect to modern mathematics through Hecke eigenvalues and Satake parameters.
\appendix

\appendix
\section{Technical explanations}
\label{sec:supplement-3points}

In this section we supplement the following three points:
\begin{enumerate}
\item Why $P/G_d$ represents subsystem structures (reductions).
\item Why the Segre orbit becomes a locally closed subscheme of the Hilbert scheme.
\item How to derive the incidence resolution directly from flattening (matrixization).
\end{enumerate}

\subsection{Why $P/G_d$ represents subsystem structures (generalities on reductions)}
\label{subsec:PGd-reduction}

Here, in order to explain things in a way that is more intuitive than maximally scheme-theoretic, we treat $X$ as a parameter space and describe everything in terms of gluing (transition functions) over an open cover $\{U_i\}$.

\begin{definition}[Gluing description of a principal $G$-bundle]
Let $G$ be a group (here we have in mind $G=\PGL_n$, etc.). Suppose that over an open cover $\{U_i\}$ of $X$ we can identify locally
\[
P|_{U_i}\ \simeq\ U_i\times G,
\]
and that on overlaps $U_{ij}:=U_i\cap U_j$ the transition isomorphisms
\[
(U_{ij}\times G)\ \xrightarrow{\ \sim\ }\ (U_{ij}\times G)
\]
are given by
\[
(x,g)\longmapsto \bigl(x,\ g_{ij}(x)\,g\bigr)
\qquad(g_{ij}(x)\in G),
\]
with the cocycle condition on triple overlaps $U_{ijk}$:
\[
g_{ij}(x)\,g_{jk}(x)\,g_{ki}(x)=1.
\]
Then these data $\{g_{ij}\}$ glue to a global principal $G$-bundle (torsor) $P\to X$.
\end{definition}

Now let $H\subset G$ be a subgroup. A reduction of the principal $G$-bundle $P\to X$ to $H$ can be
thought of as the ability to choose the transition functions $g_{ij}(x)\in G$ defining $P$ so that (after
refining the cover if necessary) they always take values in $H$.

More geometrically, it is equivalent to the following data:
\begin{itemize}
\item a principal $H$-bundle $P_H\to X$ (locally identifiable with $U_i\times H$), and
\item an isomorphism that extends it back to $G$ and recovers $P$:
\[
P_H\times_H G\ \simeq\ P.
\]
\end{itemize}
Here $P_H\times_H G$ denotes the quotient of $P_H\times G$ by the equivalence relation
\[
(p,g)\sim (p\cdot h,\ h^{-1}g)\qquad(h\in H).
\]
If locally $P_H|_{U_i}\simeq U_i\times H$, then this is locally $U_i\times G$, so one may regard it as the operation of pushing an $H$-bundle up to a $G$-bundle.

The group $H$ acts on $P$ on the right ($p\mapsto p\cdot h$). Since each fiber $P_x$ is a torsor for $G$ (a twisted copy of $G$), we can form the quotient $P_x/H$. Collecting these fiberwise quotients gives the bundle $P/H$. Intuitively, this is the bundle whose fibers are left cosets by $H$, so that the fiber is of the form $G/H$.

\begin{proposition}
\label{prop:reduction-section}
For a principal $G$-bundle $P\to X$ and a subgroup $H\subset G$, the following are naturally
equivalent:
\begin{enumerate}
\item A reduction of $P$ to $H$ exists.
\item The quotient bundle $P/H\to X$ admits a section
\[
s:X\longrightarrow P/H.
\]
\end{enumerate}
\end{proposition}

\begin{proof}
\textbf{(Step 1) Reduction $\Rightarrow$ section.}
Assume a reduction exists. That is, there is a principal $H$-subbundle $P_H\subset P$ such that for each $x$ one has an $H$-orbit $P_{H,x}\subset P_x$. For each $x\in X$, choose any point
$p\in P_{H,x}$ and define
\[
s(x):=\text{the $H$-equivalence class of $p$}\in P_x/H.
\]

\smallskip
\textbf{(Step 2) Independence of the choice of representative.}
Any other point in the same fiber can be written as $p\cdot h$. But
\[
[p\cdot h]=[p]\in P_x/H,
\]
so $s$ is well-defined.

\smallskip
\textbf{(Step 3) Section $\Rightarrow$ reduction.}
Conversely, given a section $s:X\to P/H$, define
\[
P_H:=\{\,p\in P\mid \text{the $H$-class of $p$ equals $s(\pi(p))$}\,\}
\subset P,
\]
where $\pi:P\to X$ is the projection.

\smallskip
\textbf{(Step 4) $P_H\to X$ is a principal $H$-bundle.}
If $p\in P_H$ and $h\in H$, then $p\cdot h$ defines the same $H$-class, hence $p\cdot h\in P_H$. Thus $P_H$ is preserved by the right $H$-action.

Moreover, for each $x$ the fiber $P_{H,x}$ is exactly a single $H$-orbit in $P_x$, and hence carries a free and transitive $H$-action (locally). Therefore $P_H\to X$ is a principal $H$-bundle.

\smallskip
\textbf{(Step 5) Construct $P_H\times_H G\simeq P$.}
Consider the map
\[
P_H\times G\longrightarrow P,\qquad (p,g)\longmapsto p\cdot g
\]
(which makes sense since $P$ is a right $G$-torsor). This map is constant on the equivalence relation $(p\cdot h,\ h^{-1}g)\sim(p,g)$, so it descends to a map
\[
P_H\times_H G\longrightarrow P.
\]
Locally it becomes the identity $U_i\times G\to U_i\times G$, hence it is an isomorphism. Finally, starting from a reduction and forming the section as in Step 1, then reconstructing as in Step 3, recovers the original $P_H$, and conversely as well. This proves the claim.
\end{proof}

In the setting of the main text, we take $G=\PGL_n$ and $H=G_d$, the subgroup preserving the Segre variety $\Sigma_d$. Therefore
\[
P/G_d
\]
is the bundle obtained by collecting, fiberwise, the $\PGL_n$-degrees of freedom modulo the subgroup that preserves the Segre embedding. By Proposition~\ref{prop:reduction-section}, we obtain
\[
\text{subsystem structure ($G_d$-reduction)}\quad\Longleftrightarrow\quad
\text{a section of $P/G_d$}.
\]
Equivalently, a point (or section) of $P/G_d$ represents the choice of how to insert the product locus fiber-by-fiber over a twisted background.

\bigskip

\label{subsec:Segre-orbit-Hilb}

When a group $G$ acts on a set $S$, for a point $s\in S$ define
\[
\mathrm{Orb}(s)=\{g\cdot s\mid g\in G\},\qquad
\mathrm{Stab}(s)=\{g\in G\mid g\cdot s=s\}.
\]
Then one has the basic identification
\[
\mathrm{Orb}(s)\ \simeq\ G/\mathrm{Stab}(s).
\]

Now think of $S$ as the collection of all subschemes of $\PP^{n-1}$. The geometric object that parameterizes them is the Hilbert scheme
\[
\Hilb(\PP^{n-1})
\]
(more precisely, the representing scheme whose points correspond to subschemes).
The important point is:
\begin{quote}
\emph{If $\PGL_n$ acts on $\PP^{n-1}$, then it acts on all subschemes of $\PP^{n-1}$, hence it also
acts on the Hilbert scheme.}
\end{quote}

Fix the Segre variety (the product locus) $\Sigma_d\subset \PP^{n-1}$, and write $[\Sigma_d]$ for the corresponding point of the Hilbert scheme.

\begin{definition}[Segre orbit]
\[
\mathcal{O}_{\Sigma_d}:=\PGL_n\cdot[\Sigma_d]\ \subset\ \Hilb(\PP^{n-1})
\]
is called the \emph{Segre orbit}.
\end{definition}

\begin{proposition}
The stabilizer of $[\Sigma_d]$ is
\[
\mathrm{Stab}_{\PGL_n}([\Sigma_d])=\{g\in\PGL_n\mid g(\Sigma_d)=\Sigma_d\}=G_d.
\]
\end{proposition}

\begin{proof}
This is immediate from the definition. The point of the Hilbert scheme does not move is fixed if and only if the corresponding subscheme is invariant. So the stabilizer is the group of projective transformations preserving the Segre variety.
\end{proof}

\subsubsection{The orbit map factors through $\PGL_n/G_d$}
Define the orbit map
\[
\varphi:\PGL_n\longrightarrow \Hilb(\PP^{n-1}),\qquad
g\longmapsto g\cdot[\Sigma_d].
\]
If $h\in G_d$, then $h\cdot[\Sigma_d]=[\Sigma_d]$, hence
\[
\varphi(gh)=\varphi(g).
\]
Thus $\varphi$ depends only on the left coset $gG_d$. Therefore one obtains an induced morphism
\[
\bar\varphi:\PGL_n/G_d\longrightarrow \Hilb(\PP^{n-1}),
\]
whose image is exactly the orbit $\mathcal{O}_{\Sigma_d}$.

\subsubsection{Why it is \texorpdfstring{locally closed}{locally closed}}
The key points are the following.
\begin{enumerate}
\item $\bar\varphi$ does not identify distinct cosets (including in families).
\item As a consequence, $\mathcal{O}_{\Sigma_d}$ is locally closed in the Hilbert scheme. 
\end{enumerate}

\paragraph{(A) The case of points.}
If $g_1(\Sigma_d)=g_2(\Sigma_d)$, then
\[
g_2^{-1}g_1(\Sigma_d)=\Sigma_d,
\]
so $g_2^{-1}g_1\in G_d$, hence $g_1G_d=g_2G_d$. Therefore, set-theoretically, $\PGL_n/G_d \to \mathcal{O}_{\Sigma_d}$ is injective.

\paragraph{(B) Why we consider fppf covers.}
In scheme theory, one wants to distinguish not only points but also families varying continuously over a base $T$. Suppose $x_1(t),x_2(t)\in(\PGL_n/G_d)(T)$ define the same point of the orbit. Then locally on $T$ (after an fppf cover) one can choose representatives $g_1(t),g_2(t)\in\PGL_n(T)$ such that
\[
g_1(t)(\Sigma_d\times T)=g_2(t)(\Sigma_d\times T).
\]
It follows that
\[
g_2(t)^{-1}g_1(t)\in G_d(T),
\]
hence $x_1=x_2$. Moreover, the map is injective even on families: two coordinate changes that give the same Segre variety differ only by an element of the stabilizer.

\paragraph{(C) Outline of why the image is locally closed.}
If $\bar\varphi$ is an immersion, then its image is necessarily a locally closed subscheme (analogous to an embedded image is a submanifold in differential geometry). Intuitively, the orbit is the parameter space of all Segre varieties obtained by $\PGL_n$-changes of coordinates, so it forms a smooth homogeneous space inside the Hilbert scheme, appearing as a locally closed locus in the component containing $[\Sigma_d]$.

\bigskip

\subsection{Deriving the incidence resolution directly from flattening}
\label{subsec:incidence-from-flattening}

This subsection explicitly constructs the standard incidence resolution of the determinantal variety $R_{\le k}$ (\S\ref{sec:sing}) starting only from \emph{flattening} (matrixization of a tensor).

\subsubsection{Flattening (matrixization) of a tensor}
Let $H_A\simeq k^{d_A}$ and $H_B\simeq k^{d_B}$ be finite-dimensional vector spaces, and set
\[
P:=\PP(H_A\otimes H_B).
\]
Using the basic linear-algebra isomorphism
\[
H_A\otimes H_B\ \simeq\ \Hom(H_B^\vee,\,H_A)
\]
(contraction in the $H_B$-factor), a tensor $\psi\in H_A\otimes H_B$ can be identified with a linear
map
\[
\psi_B:H_B^\vee\longrightarrow H_A.
\]
After choosing bases, this becomes the flattening matrix $M(\psi)$ from the first half.

Similarly, swapping the factors gives
\[
\psi_A:H_A^\vee\longrightarrow H_B,
\]
which corresponds to transpose at the level of matrices.

\begin{definition}[Schmidt rank and the determinantal variety]
For $\psi\neq 0$, define
\[
\rk(\psi):=\rk(\psi_B)=\rk(\psi_A).
\]
Define the projective determinantal variety (the locus of matrix rank $\le k$) by
\[
R_{\le k}\ :=\ \{[\psi]\in P\mid \rk(\psi)\le k\}.
\]
\end{definition}

\subsubsection{Idea of the resolution: record the image spaces for rank \texorpdfstring{$k$}{k}}
When $\rk(\psi)=k$, the subspaces
\[
U_A:=\im(\psi_B)\subset H_A,\qquad
U_B:=\im(\psi_A)\subset H_B
\]
are uniquely determined as \emph{$k$-dimensional subspaces}.

On the other hand, when $\rk(\psi)<k$, there are many choices of $k$-dimensional subspaces $U_A$ containing $\im(\psi_B)$. This non-uniqueness of choices is the intuitive source of the singularities of the determinantal variety. The incidence resolution resolves the singularities by parameterizing those choices together with the tensor.

\subsubsection{Construction: Grassmannians and an incidence variety}
The set of all $k$-dimensional subspaces is parameterized by the Grassmannians
\[
\Gr_A(k):=\Gr(k,H_A),\qquad \Gr_B(k):=\Gr(k,H_B).
\]
Each Grassmannian carries the tautological subbundle:
\[
\mathcal U_A\subset H_A\otimes\mathcal O_{\Gr_A(k)},\qquad
\mathcal U_B\subset H_B\otimes\mathcal O_{\Gr_B(k)}.
\]
The fiber of $\mathcal U_A$ over $[U_A]\in\Gr_A(k)$ is exactly $U_A$.

\begin{definition}[Incidence variety (symmetric version)]
Define
\[
\widetilde R_{\le k}
:=
\PP(\mathcal U_A\boxtimes \mathcal U_B)
\ \longrightarrow\
\Gr_A(k)\times\Gr_B(k).
\]
Equivalently, a point is the same as
\[
(U_A,\ U_B,\ [\psi])
\quad\text{where}\quad
U_A\in\Gr_A(k),\ U_B\in\Gr_B(k),\ [\psi]\in\PP(U_A\otimes U_B).
\]
\end{definition}

\begin{definition}[Incidence resolution]
Using the inclusion $U_A\otimes U_B\subset H_A\otimes H_B$, define the forgetful map
\[
\rho_k:\widetilde R_{\le k}\longrightarrow \PP(H_A\otimes H_B)=P,
\qquad
(U_A,U_B,[\psi])\longmapsto [\psi].
\]
\end{definition}

\subsubsection{Showing \texorpdfstring{$\mathrm{Im}(\rho_k)\subset R_{\le k}$}{Im(rho\_k) subset R\_{\le k}} using only flattening}
\begin{lemma}
$\rho_k(\widetilde R_{\le k})\subset R_{\le k}$, namely any $\psi\in U_A\otimes U_B$ has rank $\le k$.
\end{lemma}

\begin{proof}
Let $\psi\in U_A\otimes U_B$, and consider the flattening map $\psi_B:H_B^\vee\to H_A$. Since $\psi$ lies in $U_A\otimes U_B$, contracting in the $H_B$-factor produces only linear combinations of vectors in $U_A$. Therefore
\[
\im(\psi_B)\subset U_A.
\]
Because $U_A$ has dimension $k$, we have
\[
\rk(\psi)=\rk(\psi_B)=\dim\im(\psi_B)\le \dim U_A=k,
\]
so $[\psi]\in R_{\le k}$.
\end{proof}

\subsubsection{It is an isomorphism over the rank-exactly-\texorpdfstring{$k$}{k} locus}
Let
\[
R_{=k}:=R_{\le k}\setminus R_{\le k-1}
\]
be the open locus of tensors of rank exactly $k$. We show that $\rho_k$ restricts to an isomorphism over $R_{=k}$. This is the reason the construction is birational, and it is the core of the resolution.

\begin{proposition}[Explicit inverse over the rank-$k$ locus]
\label{prop:inverse-on-rankk}
For a point $[\psi]\in R_{=k}$, set
\[
U_A:=\im(\psi_B)\in\Gr_A(k),\qquad
U_B:=\im(\psi_A)\in\Gr_B(k).
\]
Then
\[
[\psi]\in \PP(U_A\otimes U_B).
\]
Therefore the map
\[
s:R_{=k}\longrightarrow \widetilde R_{\le k},\qquad
[\psi]\longmapsto (U_A,U_B,[\psi])
\]
is well-defined and satisfies $\rho_k\circ s=\id_{R_{=k}}$. Moreover, on $\rho_k^{-1}(R_{=k})$ one also has $s\circ\rho_k=\id$, and therefore
\[
\rho_k:\rho_k^{-1}(R_{=k})\xrightarrow{\sim} R_{=k}
\]
is an isomorphism.
\end{proposition}

\begin{proof}
\textbf{(Step 1) $U_A$ and $U_B$ are $k$-dimensional.}
Since $[\psi]\in R_{=k}$, we have $\rk(\psi)=k$. Thus
\[
\dim U_A=\dim\im(\psi_B)=k,\qquad
\dim U_B=\dim\im(\psi_A)=k,
\]
so $U_A\in\Gr_A(k)$ and $U_B\in\Gr_B(k)$.

\smallskip
\textbf{(Step 2) Show $\psi\in U_A\otimes U_B$ (the essential point).}
We give a completely elementary proof by choosing bases and viewing $\psi$ as a matrix. Extend a basis $u_1,\dots,u_k$ of $U_A=\im(\psi_B)\subset H_A$ to a basis $u_1,\dots,u_{d_A}$ of $H_A$. Similarly, extend a basis $v_1,\dots,v_k$ of $U_B=\im(\psi_A)\subset H_B$ to a basis $v_1,\dots,v_{d_B}$ of $H_B$.

Let $M(\psi)$ be the flattening matrix of $\psi$ in these bases. The condition $\im(\psi_B)=U_A=\langle u_1,\dots,u_k\rangle$ means that the column space of $M(\psi)$ lies in the span of $u_1,\dots,u_k$, hence the \emph{bottom $(d_A-k)$ rows are all zero}. Likewise, $\im(\psi_A)=U_B$ (equivalently, looking at the transpose) means that the \emph{rightmost $(d_B-k)$ columns are all zero}.

Therefore $M(\psi)$ has the block form
\[
M(\psi)=
\begin{pmatrix}
\ast & 0\\
0 & 0
\end{pmatrix}
\]
(only the top-left $k\times k$ block can be nonzero).
Equivalently, as a tensor,
\[
\psi=\sum_{i=1}^k\sum_{j=1}^k c_{ij}\,u_i\otimes v_j,
\]
which lies in $U_A\otimes U_B$. Hence
\[
[\psi]\in \PP(U_A\otimes U_B).
\]
Thus $s$ is well-defined and $\rho_k\circ s=\id$, since $\rho_k$ simply forgets $(U_A,U_B)$.

\smallskip
\textbf{(Step 3) Show $s\circ\rho_k=\id$ over the rank-$k$ locus.}
A point of $\rho_k^{-1}(R_{=k})$ is of the form
\[
(U_A,U_B,[\psi])\quad\text{with}\quad \rk(\psi)=k,
\]
and $[\psi]\in\PP(U_A\otimes U_B)$.
By the same flattening argument, $\im(\psi_B)\subset U_A$. Since both sides have dimension $k$, we get
\[
\im(\psi_B)=U_A.
\]
Similarly, $\im(\psi_A)=U_B$. Therefore the construction of $s([\psi])$ recovers the original triple $(U_A,U_B,[\psi])$, so $s\circ\rho_k=\id$ on $\rho_k^{-1}(R_{=k})$.

Hence $\rho_k$ is an isomorphism over $R_{=k}$, and in particular is birational.
\end{proof}

Thus $\widetilde R_{\le k}$ is smooth, and $\rho_k:\widetilde R_{\le k}\to R_{\le k}$ is a projective birational morphism which is an isomorphism over the rank-$k$ open locus. Therefore $\rho_k$ is the (standard) incidence resolution of the determinantal variety $R_{\le k}$.

If one rewrites the above argument not in terms of matrices at each point but in terms of the universal tensor on the projective space $P$, one recovers the two universal flattening maps used later:
\[
\Phi_A:H_A^\vee\otimes \mathcal O_P(-1)\to H_B\otimes \mathcal O_P,\qquad
\Phi_B:H_B^\vee\otimes \mathcal O_P(-1)\to H_A\otimes \mathcal O_P.
\]
On the rank-constant open subset $R_{=k}$, the images become rank-$k$ subbundles, and by the universal property of Grassmannians one obtains $(U_A,U_B)$ and hence the map $s$ as above.

\section{On branches and coordinate gluing on the torus}\label{sec:torus}

\subsection{What is a branch?}
This appendix explains the coordinate gluing on the torus used in \S\ref{sec:holonomic_gate} and \ref{sec:spin-chain-toy}.
In those examples, we considered the external parameter space
\[
X=(\mathbb{C}^{\times})^{2},\qquad (u,v)\in X,
\]
where $\mathbb{C}^\times=\mathbb{C}\setminus\{0\}$, so both $u$ and $v$ are nonzero complex numbers. Physically, one may restrict to $|u|=|v|=1$ and view this as a parameter space on the torus $X_{\mathrm{phys}}=(S^1)^2\subset (\mathbb{C}^\times)^2$.

\medskip

The notation $X^{\mathrm{an}}$ denotes the \emph{analytification} of $X$, i.e.\ viewing $X$ in the setting of complex analysis (the usual topology, continuity, and holomorphicity). Intuitively, it is harmless to think of
\[
X^{\mathrm{an}}\cong(\mathbb{C}^{\times})^{2}
\]
as a \emph{complex manifold} (in the usual sense, a space with open sets, continuous paths, and loops).

We write $X^{\mathrm{an}}$ explicitly because we want to discuss analytic notions such as analytic continuation along loops and the resulting monodromy (the mismatch after going once around a loop).

\medskip
On this space, branching phenomena occur in general. For example, for a complex number $u\in\mathbb{C}^\times$, the fourth root $u^{1/4}$ has four values in general. Writing
\[
u=re^{i\theta}\qquad (r>0),
\]
we have
\[
u^{1/4}=r^{1/4}e^{\,i(\theta+2\pi k)/4}\qquad (k=0,1,2,3).
\]
Thus, viewed as a function, $u^{1/4}$ is originally a \emph{multivalued function}.

Hence, by choosing a branch we mean the following:
\begin{quote}
On a region $U$, choose one of the four values at each point $(u,v)\in U$ in a way that is consistently continuous (and preferably holomorphic) on $U$.
\end{quote}
If we denote the chosen branch by $u^{1/4}\big|_{U}$, then it becomes a \emph{single-valued} function on $U$ and satisfies
\[
\bigl(u^{1/4}\big|_{U}\bigr)^4=u.
\]

\subsection{How do we choose a branch?}
A standard way to construct a branch is to use $\log$. The map $u\mapsto \log u$ (the complex logarithm) is itself multivalued, but on a \emph{simply connected} domain $U\subset\mathbb{C}^\times$ one can choose a suitable branch cut and obtain a \emph{single-valued logarithm} (a branch of the logarithm)
\[
\mathrm{Log}_{U}:U\to\mathbb{C}\quad(\text{holomorphic}),\qquad e^{\mathrm{Log}_{U}(u)}=u.
\]
Then defining
\[
u^{1/4}\big|_{U}:=\exp\left(\frac{1}{4}\mathrm{Log}_{U}(u)\right)
\]
gives a holomorphic fourth root on $U$.

In terms of the physical torus $X_{\mathrm{phys}}=(S^1)^2$, if we write
\[
u=e^{i\theta_u},
\]
then $\theta_u$ is intrinsically $2\pi$-periodic ($\theta_u\sim \theta_u+2\pi$), and hence \emph{cannot be chosen globally as a single-valued function}. However, on a simply connected patch $U$ obtained by cutting open the torus (for example, restricting $\theta_u$ to an interval of length $<2\pi$), one can choose $\theta_u$ continuously. In that case,
\[
u^{1/4}\big|_{U}=e^{\,i\theta_u/4}
\]
is defined analytically.

In short,
\[
\text{take a simply connected }U
\text{ then a branch (a single-valued choice of }u^{1/4}\text{) exists on }U.
\]

\medskip

\subsection{Why simply connectedness is necessary: monodromy along a loop}
Since $\mathbb{C}^\times$ has a hole at the origin $0$, a loop going once around $0$ cannot be contracted to a point (it is not simply connected). Because of this hole, the value of $u^{1/4}$ generally does not return to the original value after going once around such a loop. It comes back multiplied by a constant phase factor.

Concretely, starting from a base point $(u_0,v_0)$ and considering the loop
\[
\gamma_u(t)=(u_0e^{2\pi it},\,v_0)\qquad (0\le t\le 1),
\]
analytic continuation of a chosen branch yields the transformation
\[
u^{1/m}\longmapsto \zeta\,u^{1/m}\qquad\bigl(\zeta=e^{2\pi i/m}\bigr).
\]
In the spin-chain example on the torus we have $m=4$, so $\zeta=i$, and thus
\[
u^{1/4}\longmapsto i\,u^{1/4}.
\]
That is, even though we return to the same point $(u_0,v_0)$, the value of the fourth root is multiplied by $i$. Therefore, one cannot define $u^{1/4}$ as a single-valued function on all of $X$ (or on any region containing such a loop). This is why one must restrict to simply connected charts $U$ in order to choose a branch.

\medskip

\subsection{Gluing the Hamiltonian and the ground state}
In \S\ref{sec:spin-chain-toy}, the local Hamiltonian $H_U(u,v)$ for the spin-chain example contains coefficients such as $u^{\pm 1/4}$. If $u^{1/4}$ remains multivalued, then $H_U$ cannot be specified as a single function of $(u,v)$, and computations such as matrix representations and eigenstates (including the ground state) become ambiguous (This ambiguity is intentional. It provides a standard way to obtain a well-defined Hamiltonian family on a nontrivial parameter manifold.).

Thus we proceed as follows:
\begin{enumerate}
    \item take a simply connected $U$;
    \item define a single-valued branch;
    \item get a well-defined $H_U(u,v)$ on $U$
\end{enumerate}
so that the Hamiltonian and the ground state can first be written correctly on each chart.

On overlaps, different branch choices are glued by a gauge transformation (here, conjugation by the translation operator), producing a twisted family. If the resulting monodromy leaves the local-operation group, a state that is locally a product state can become entangled after gluing. Confirming this phenomenon was the aim of the entire example.

In much of the physics literature (especially in condensed matter and quantum information), one often works within a single coordinate description, which can make global gluing effects less explicit. From the geometric viewpoint adopted here, such global entanglement arises naturally and can be studied systematically.

\section*{Acknowledgment}
I am deeply indebted to my collaborators for many years of close collaboration during which these ideas took shape, in particular Dmitri Kharzeev, Yaron Oz, and Steven Rayan.

\medskip
Illustrations used in page 57 were created by Megumi Ikeda. 

\medskip
This work was supported by the U.S. National Science Foundation (NSF), Office of Strategic Initiatives, under Grant No. OSI-2328774.

\bibliographystyle{JHEP}
\bibliography{ref}

\end{document}